%% file: ToNPaper.tex
\newif\ifreport\reporttrue
 \documentclass[lettersize,journal]{IEEEtran}
\usepackage{pdfpages}
\usepackage{amsmath,amssymb,amsfonts,amsthm}
\usepackage{subcaption}
\usepackage{color}
\usepackage{algorithm}
\usepackage{algpseudocode}

\newtheorem{theorem}{Theorem}
\newtheorem{corollary}{Corollary}
\newtheorem{definition}{Definition}
\newtheorem{lemma}{Lemma}

\newtheorem{example}{Example}

\def\blue{\color{blue}}
\def\red{\color{red}}
\def\violet{\color{violet}}
\usepackage{color}
\usepackage{tcolorbox}
\colorlet{blue}{black}
\colorlet{red}{black}
\usepackage{lipsum}
\usepackage{hyperref}
\usepackage[noadjust]{cite}
\DeclareMathOperator*{\argmax}{arg\,max}
\DeclareMathOperator*{\argmin}{arg\,min}
 \begin{document}

\title{Timely Communications for Remote Inference}
        
\author{Md Kamran Chowdhury Shisher,~\IEEEmembership{Member,~IEEE,}
        Yin~Sun,~\IEEEmembership{Senior~Member,~IEEE,} I-Hong~Hou,~\IEEEmembership{Senior~Member,~IEEE}\IEEEcompsocitemizethanks{\IEEEcompsocthanksitem This paper was presented in part at ACM MobiHoc, 2022 \cite{ShisherMobihoc}. 
        
        Md Kamran Chowdhury Shisher was with the Department of Electrical and Computer Engineering, Auburn University, Auburn, AL, 36849, USA. He is now with the Elmore Family School of Electrical and Computer Engineering, Purdue University, West Lafayette, IN, 47907, USA (e-mail: mzs0153@auburn.edu, mshisher@purdue.edu). 
        
        Yin Sun is with the Department of Electrical and Computer Engineering, Auburn University, Auburn, AL,
36849, USA (e-mail: yzs0078@auburn.edu). 

I-Hong Hou is with the Department of Electrical and Computer Engineering, Texas A\&M University, College Station, TX 77843 USA (e-mail: ihou@tamu.edu). 

The work of Md Kamran Chowdhury Shisher and Yin Sun was supported in part by the NSF grant CNS-2239677 and the ARO grant W911NF-21-1-0244. The work of I-Hong Hou was supported in part by NSF under Award Number ECCS-2127721 and in part by the U.S. Army Research Laboratory and the U.S. Army Research Office under Grant Number W911NF-22-1-0151.}
}

\newcommand{\ignore}[1]{{}}
\pagestyle{plain}
\def\blue{\color{blue}}
\maketitle

\begin{abstract}
In this paper, we analyze the impact of data freshness on remote inference systems, where a {\blue pre-trained} neural network {\blue infers} a time-varying target (e.g., the locations of vehicles and pedestrians) based on features (e.g., video frames) observed at a sensing node (e.g., a camera). 
One might expect that the performance of a remote inference system degrades monotonically as the feature becomes stale. Using an information-theoretic analysis, we show that this is true if the feature and target data sequence can be closely approximated as a Markov chain{\blue, whereas} it is not true if the data sequence is far from {\blue being} Markovian. Hence, the inference error is a function of Age of Information (AoI), where the function could be non-monotonic. 
To minimize the inference error in real-time, we propose a new ``selection-from-buffer'' model for sending the features, which is more general than the ``generate-at-will'' model used in earlier studies. In addition, we design low-complexity scheduling policies to improve inference performance. {\blue For single-source, single-channel systems, we provide an optimal scheduling policy.
In multi-source, multi-channel systems, the scheduling problem becomes a multi-action restless multi-armed bandit problem.} For this setting, we design a new scheduling policy by integrating Whittle index-based source selection and duality-based feature selection-from-buffer algorithms. This new scheduling policy is proven to be asymptotically optimal. 
These scheduling results hold for minimizing general AoI functions (monotonic or non-monotonic). {\blue Data-driven evaluations demonstrate the significant advantages of our proposed scheduling policies.}
\end{abstract}


\input{introduction}

\input{LearningModel}

\input{Interpretation}

\input{Scheduling_SingleSource}

\input{Scheduling_MultiSource}
\input{Simulations}

\input{Conclusions}
\bibliographystyle{IEEEtran}
\bibliography{refshisher}
\begin{IEEEbiography}
[{\includegraphics[width=1in,height=1.25in,clip,keepaspectratio]{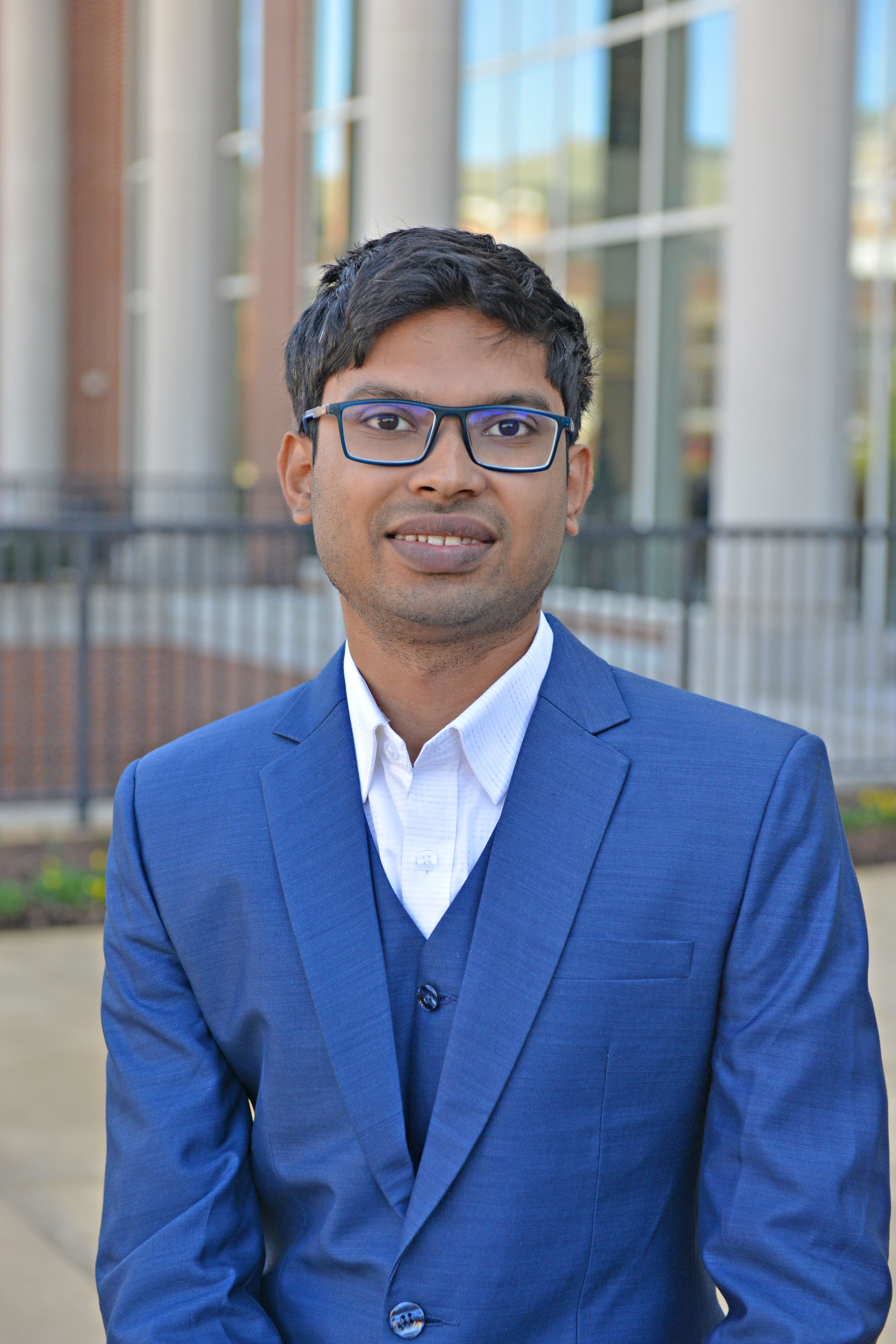}}]{Md Kamran Chowdhury Shisher}(Member, IEEE)
    received the B.Sc. degree in Electrical and Electronic Engineering from the Bangladesh University of Engineering and Technology, Dhaka, Bangladesh, in 2017. He received his M.Sc. and Ph.D. degrees in Electrical Engineering from Auburn University, AL, USA, in 2022 and 2024, respectively. He is currently a Postdoctoral Researcher with the Department of Electrical and Computer Engineering, Purdue University, West Lafayette, IN, USA. His research interests include Networking, Machine Learning, Information Freshness, and Information Theory. 
\end{IEEEbiography}
\vspace{5cm}
\begin{IEEEbiography}[{\includegraphics[width=1in,height=1.25in,clip,keepaspectratio]{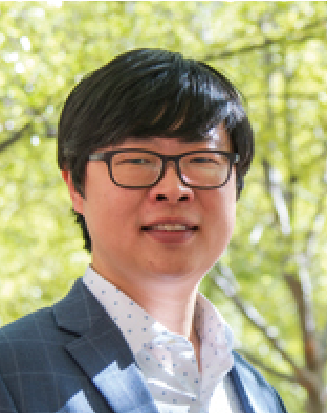}}]{Yin Sun}(Senior Member, IEEE) is the Bryghte D. and Patricia M. Godbold Endowed Associate Professor in the Department of Electrical and Computer Engineering at Auburn University, Alabama. He received his B.Eng. and Ph.D. degrees in Electronic Engineering from Tsinghua University, in 2006 and 2011, respectively. He was an Assistant Professor in the Department of Electrical and Computer Engineering at Auburn University from 2017 to 2023 and a Postdoctoral Scholar and Research Associate at the Ohio State University from 2011 to 2017. His research interests include Networking, Machine Learning, Semantic Communications, Age of Information, and Information Theory. His articles received the Best Student Paper Award of the IEEE/IFIP WiOpt 2013, the Best Paper Award of the IEEE/IFIP WiOpt 2019, runner-up for the Best Paper Award of ACM MobiHoc 2020, and the Journal of Communications and Networks (JCN) Best Paper Award in 2021. He received the Auburn Author Award in 2020, the National Science Foundation (NSF) CAREER Award in 2023, the Bryghte D. and Patricia M. Godbold Endowed Professorship in 2023, the Ginn Faculty Achievement Fellowship in 2023, and the College of Engineering's Research Award for Excellence (Senior Faculty) in 2024. 
\end{IEEEbiography}
\vspace{5cm}
\begin{IEEEbiography}[{\includegraphics[width=1in,height=1.25in,clip,keepaspectratio]{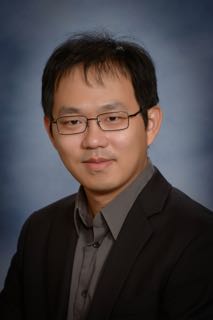}}]{I-Hong Hou}(Senior Member, IEEE) is an Associate Professor in the ECE Department of the Texas A\&M University. He received his Ph.D. from the Computer Science Department of the University of Illinois at Urbana-Champaign. His research interests include wireless networks, edge/cloud computing, and reinforcement learning. His work has received the Best Paper Award from ACM MobiHoc 2017 and ACM MobiHoc 2020, and Best Student Paper Award from WiOpt 2017. He has also received the C.W. Gear Outstanding Graduate Student Award from the University of Illinois at Urbana-Champaign, and the Silver Prize in the Asian Pacific Mathematics Olympiad.
\end{IEEEbiography}
\newpage
\begin{appendices}
\input{Appendix_Paper}
\end{appendices}
\end{document}

%% file: introduction.tex
\section{Introduction}\label{introduction}




{\blue Next-generation communications (Next-G), such as 6G, are expected to enable emerging networked intelligent systems, including autonomous driving, factory automation, digital twins, UAV navigation, and extended reality. These systems rely on timely inference, where a pre-trained neural network infers time-varying targets (e.g., vehicle and pedestrian locations) using features (e.g., video frames) transmitted from sensing nodes (e.g., cameras). Due to communication delay and transmission errors, the features delivered to the neural network may not be fresh. Traditionally, information freshness has not been a major concern in inference systems. However, in time-sensitive applications, it is critical to understand how information freshness impacts inference performance.}

{\blue  The concept of Age of Information (AoI), introduced in \cite{song1990performance, kaul2012real}, measures the freshness of information at the receiver. Consider packets sent from a source to a receiver. If $U(t)$ is the generation time of the freshest received packet by time $t$, then the AoI $\Delta(t)$ is defined as the time difference between the current time $t$ and the generation time $U(t)$ of the freshest received packet, expressed as
\begin{align}\label{AoIIntro}
    \Delta(t)=t-U(t).
\end{align} 
A smaller AoI indicates the presence of more recent information at the receiver. While existing AoI research extensively analyzes and optimizes linear and nonlinear functions of AoI \cite{yates2015lazy,sun2017update, YinUpdateInfocom, SunNonlinear2019, SunSPAWC2018, orneeTON2021, Tripathi2019, klugel2019aoi, bedewy2021optimal, kadota2018optimizing, hsu2018age, sun2019closed, Kadota2018}, there is a lack of clear understanding regarding the value of fresh information in real-time systems. This gap motivates us to seek a more suitable analysis for evaluating the significance of fresh information for systems using it.}

{\blue Previous studies \cite{yates2015lazy,sun2017update, YinUpdateInfocom, SunNonlinear2019, SunSPAWC2018, orneeTON2021, Tripathi2019, klugel2019aoi, bedewy2021optimal, kadota2018optimizing, hsu2018age, sun2019closed, Kadota2018} have assumed that the performance of real-time systems degrades monotonically as the AoI increases. However, our remote inference experiments (see Sec. \ref{Experimentation}) demonstrate that this assumption holds in some scenarios (e.g., video prediction), but not in others (e.g., temperature prediction or reaction prediction). 
For instance, when predicting the temperature for the next hour, the temperature recorded 24 hours ago is more relevant than the temperature recorded 12 hours ago due to daily weather patterns. Therefore, fresher temperature data is not always better for temperature prediction.
Similarly, in a reaction prediction system, if a vehicle A initiates braking, nearby vehicles do not react instantly due to the response delay of  drivers and  braking systems. Therefore, slightly outdated actions of vehicle A can be more relevant for predicting the reactions of nearby vehicles. These observations highlight that fresher data is not always better, emphasizing the need for a  theory to accurately assess the value of fresh information.}

 {\blue In this paper, we introduce a new information-theoretic framework for remote inference to analyze when fresher data is better and when it is not. We also study optimal communication system design for remote inference. Most existing studies in the AoI literature \cite{yates2015lazy,sun2017update, YinUpdateInfocom, SunNonlinear2019, SunSPAWC2018, orneeTON2021, Tripathi2019, klugel2019aoi, bedewy2021optimal, kadota2018optimizing, hsu2018age, sun2019closed, Kadota2018} have focused on designing transmission scheduling strategies to minimize monotonic functions of AoI. However, the design of efficient scheduling policies for optimizing general, potentially non-monotonic functions of AoI remains unexplored. To address this gap, we design the first transmission scheduling policies to minimize general functions of AoI, regardless of whether the functions are monotonic or non-monotonic. The contributions of this paper are summarized as follows:} 

\ignore{Timely communication of information sources is crucial for many networked intelligent systems, especially for tasks involving inference, estimation, and control. In \cite{song1990performance, kaul2012real}, the concept of \emph{Age of Information} was introduced to measure the  timeliness of communication. {\red what is the "timeliness of communication?"} Consider a sequence of status updating packets sent from an information source to a receiver through a communication channel. Let $U(t)$ be the generation time of the freshest packet delivered to the receiver by time $t$. The AoI $\Delta(t)$ is defined as the time difference between the current time $t$ and the generation time $U(t)$ of the freshest received packet, expressed as
\begin{align}\label{AoIIntro}
    \Delta(t)=t-U(t).
\end{align} 
A smaller AoI indicates the presence of more recent information at the receiver.

{\red the start of the following paragraph seems indirect or unclear.  } 

{The importance of fresh information depends on the purpose of communications and the nature of the information sources. Typically, there exists a correlation structure in information sources. {For instance, a strong correlation is evident between the current temperature and the temperature recorded $24$ hours ago. 
In predicting the temperature for the next hour, the temperature recorded 24 hours ago holds more significance than the temperature recorded 12 hours ago.} 
Moreover, consider an example of reaction prediction: when a vehicle initiates breaking, a nearby vehicle's braking system may not respond immediately due to inherent delay caused by the reaction time of the driver and the actuation delay of the brake system. Consequently, slightly outdated actions with AoI $>0$ can be more valuable than the current action with AoI = 0 in predicting the reactions of nearby vehicles. These examples underscore that the utility of fresh information is contingent on both the purpose of the information and the temporal correlation structure of the involved signals. 
However, the AoI $\Delta(t)$ does not specify the signal correlation structure or the goal of the communication system. Presently, there is a lack of comprehensive understanding regarding \emph{how to evaluate the value of fresh information in real-time systems}. This knowledge gap has prompted us to seek a more suitable metric for evaluating the significance of fresh information in goal-oriented  communications.} 


The evolution of artificial intelligence, control, and communications technologies has given rise to a new era of networked intelligent systems, which include autonomous driving, remote surgery, real-time surveillance, video analytics, and factory automation. Timely Inference is vital in these systems, where a trained neural network infers time-varying targets (e.g., the locations of vehicles and pedestrians) based on features (e.g., video frames) sent from a sensing node (e.g., camera). 
Due to communication delay and transmission errors, the features delivered to the neural network may not be fresh. In traditional inference systems where timeliness is not critical, information freshness is not a concern. However, in time-sensitive inference systems, it is crucial to understand the impact of information freshness on the inference performance. 

{In this paper, we first analyze the impact of fresh information  on remote inference. While one might assume that inference errors degrade monotonically as the data becomes stale, our remote inference experiments show that this assumption holds true in some scenarios (e.g., video prediction) and not in others (e.g., temperature prediction and reaction prediction). By developing a new information-theoretic analysis, we reveal {\blue when inference errors are monotonic in AoI, and when they are not.} inference errors are functions of the AoI, whereas the function is not necessarily monotonic. The second part of this paper focuses on optimal design of communication systems for remote inference. Most existing studies in AoI literature \cite{yates2015lazy,sun2017update, YinUpdateInfocom, SunNonlinear2019, SunSPAWC2018, orneeTON2021, Tripathi2019, klugel2019aoi, bedewy2021optimal, kadota2018optimizing, hsu2018age, sun2019closed, Kadota2018} have focused on designing transmission scheduling strategies to minimize monotonic functions of AoI. However, the design of efficient scheduling policies for optimizing general, potentially non-monotonic functions of AoI remains unexplored. To that end, we design the first transmission scheduling policies to minimize general functions of AoI, regardless of whether the functions are monotonic or non-monotonic. The contributions of this paper are summarized as follows:}} 

\begin{itemize}

\item We conduct five experiments to examine the impact of data freshness on remote inference. These experiments include (i) video prediction, (ii) robot state prediction in a leader-follower robotic system, (iii) actuator state prediction under mechanical response delay, (iv) temperature prediction, and (v) wireless channel state information (CSI) prediction. One might assume that the inference error degrades monotonically as the data becomes stale. Our experimental results show that this assumption is not always true. In some scenarios, even the fresh data with $\Delta(t)=0$ may generate a larger inference error than stale data with $\Delta(t)>0$ (see Figs \ref{fig:DelayedNetworkedControlled}-\ref{fig:TrainingCartVelocity}).



\item We develop two theoretical methods to interpret these counter-intuitive experimental results. 
First, by a local information geometric analysis, we show that the assumption ``fresh data is better than stale data'' is true when the time-sequence data used for remote inference can be closely approximated as a Markov chain; but it is not true when the data sequence is far from Markovian (Theorems \ref{theorem1}-\ref{theorem3}). Hence, the inference error is a function of the AoI, whereas the function is not necessarily monotonic. This analysis provides an information-theoretic interpretation of information freshness. Second, we construct two analytical models to analyze and explain when fresh data is better than stale data and when it is not (see Sec. \ref{ExperimentalObs}).

\item 
In the second part of the paper, we design transmission scheduling policies for minimizing the inference error. Because fresher data is not always better, we propose a new medium access model called the ``selection-from-buffer" model, where $B$ most recent features are stored in the source's buffer and the source can choose to send any of the $B$ most recent features. This model is more general than the ``generate-at-will'' model used in earlier studies, e.g., \cite{yates2015lazy, YinUpdateInfocom, sun2017update, SunNonlinear2019,  bedewy2021optimal, SunSPAWC2018, orneeTON2021}. 
If the inference error is an non-decreasing function of the AoI, 
the ``selection-from-buffer'' model achieves the same performance as the ``generate-at-will'' model; if the AoI function is non-monotonic, the ``selection-from-buffer'' model can potentially achieve better performance. 

\item {\blue For single-source and single-channel systems}, an optimal scheduling policy is devised to determine (i) when to submit features to the channel and (ii) which feature in the buffer to submit. This scheduling policy is capable of minimizing general functions of the AoI, regardless of whether the function is monotonic or not. By leveraging a new index function $\gamma(\Delta(t))$, the optimal scheduling policy can be expressed an \emph{index-based threshold policy}, where a new packet is sent out whenever $\gamma(\Delta(t))$ exceeds a pre-determined threshold (Theorems \ref{theorem5}-\ref{theorem6}). The threshold can be computed by using low complexity algorithms, e.g., bisection search. We note that the  function $\gamma(\cdot)$ is \emph{not necessarily monotonic} and hence its inverse function \emph{may not exist}. Consequently, this index-based threshold policy cannot be equivalently expressed as an AoI-based threshold policy,  i.e., a new packet is sent out whenever the AoI $\Delta(t)$ exceeds a pre-determined threshold. This is a key difference from prior studies on minimizing non-decreasing AoI functions
\cite{yates2015lazy,sun2017update, YinUpdateInfocom, SunNonlinear2019,SunSPAWC2018, orneeTON2021, Tripathi2019, klugel2019aoi, bedewy2021optimal, kadota2018optimizing, hsu2018age, sun2019closed, Kadota2018}, where the optimal scheduling policy is an AoI-based threshold policy. 

\item {{\blue In multi-source and multi-channel systems}, the scheduling problem is a restless multi-armed bandit (RMAB) problem with multiple actions. We propose a multi-source, multi-action scheduling policy that uses a Whittle index algorithm to determine which sources to  schedule and employs a duality-based selection-from-buffer algorithm to
decide which features to schedule from the buffers of these sources. By utilizing linear programming (LP)-based priority conditions \cite{verloop2016asymptotically, gast2021lp}, we establish the asymptotic optimality of this scheduling policy as the numbers of sources and channels tend to infinity, maintaining a constant ratio (see Theorem \ref{asymptotic_optimal}).}


\item The above results hold (i) for minimizing general AoI functions (monotonic or non-monotonic) and (ii) for random delay channels. 
Data-driven evaluations show that the optimal scheduler achieves up to $3$ times smaller inference error compared to ``generate-at-will” with optimal scheduling strategy and $8$ times smaller inference error compared to periodic feature updating (see Fig. \ref{fig:singlesourcedifferentsigma}). 
Numerical results further validate the asymptotic optimality of the proposed scheduling policy (see Figs. \ref{fig:multisourceN}-\ref{fig:multisourceweight1}). 
\item When the training and inference data have the same probabilistic distribution, remote inference reduces to signal-agnostic remote estimation. Hence, the results of the present work above also apply to signal-agnostic remote estimation.
\end{itemize}
\ignore{In this work, we analyze the impact of data freshness on real-time supervised learning. We also develop an optimal feature buffering and transmission strategy that minimizes the inference error and has low complexity.}

\ignore{\begin{figure}[t]
\centering
\includegraphics[width=0.20\textwidth]{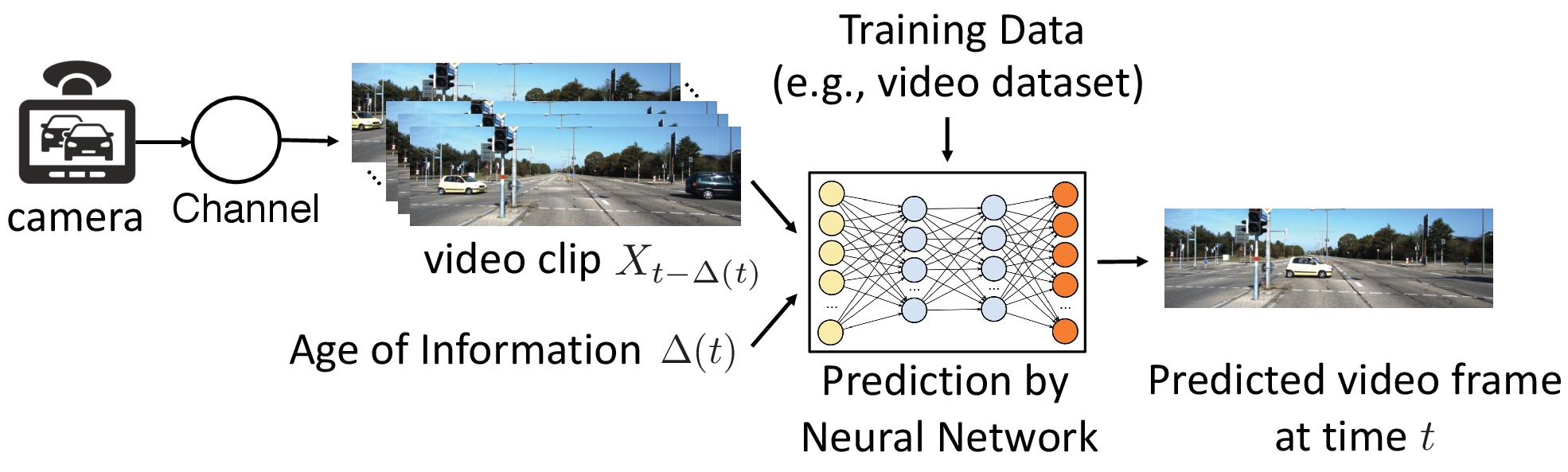}
\caption{\small A real-time supervised learning system. \label{fig:learning}
}
\end{figure}}

\subsection{Related Works}

\begin{figure*}[h]
  \centering
  
  \begin{subfigure}[t]{0.45\textwidth}
\includegraphics[width=\textwidth]{Learning_Model2.eps}
  \subcaption{Video prediction Task}
\end{subfigure}
%
\hspace{3mm} 
\begin{subfigure}[t]{0.20\textwidth}
\includegraphics[width=\textwidth]{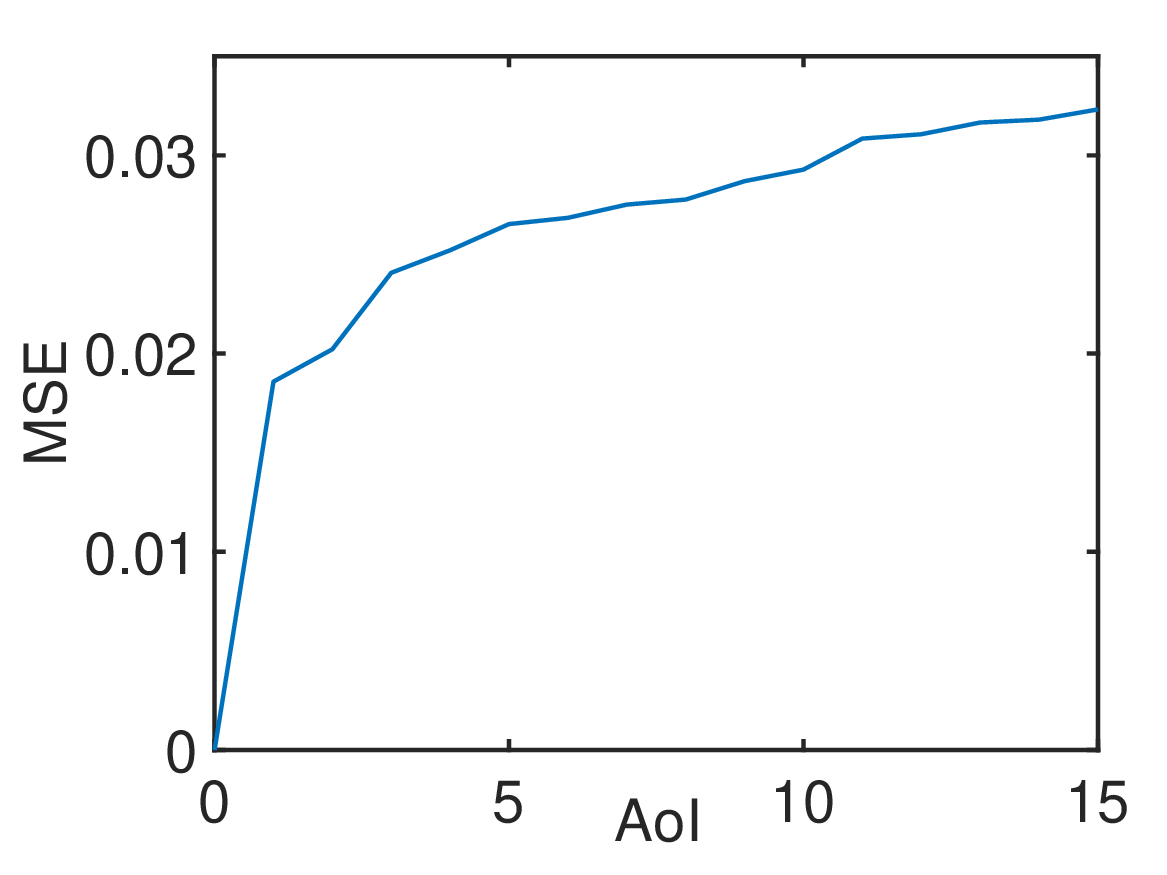}
\subcaption{Training Error vs. AoI}
\end{subfigure}
%
\hspace{3mm}
\begin{subfigure}[t]{0.20\textwidth}
\includegraphics[width=\textwidth]{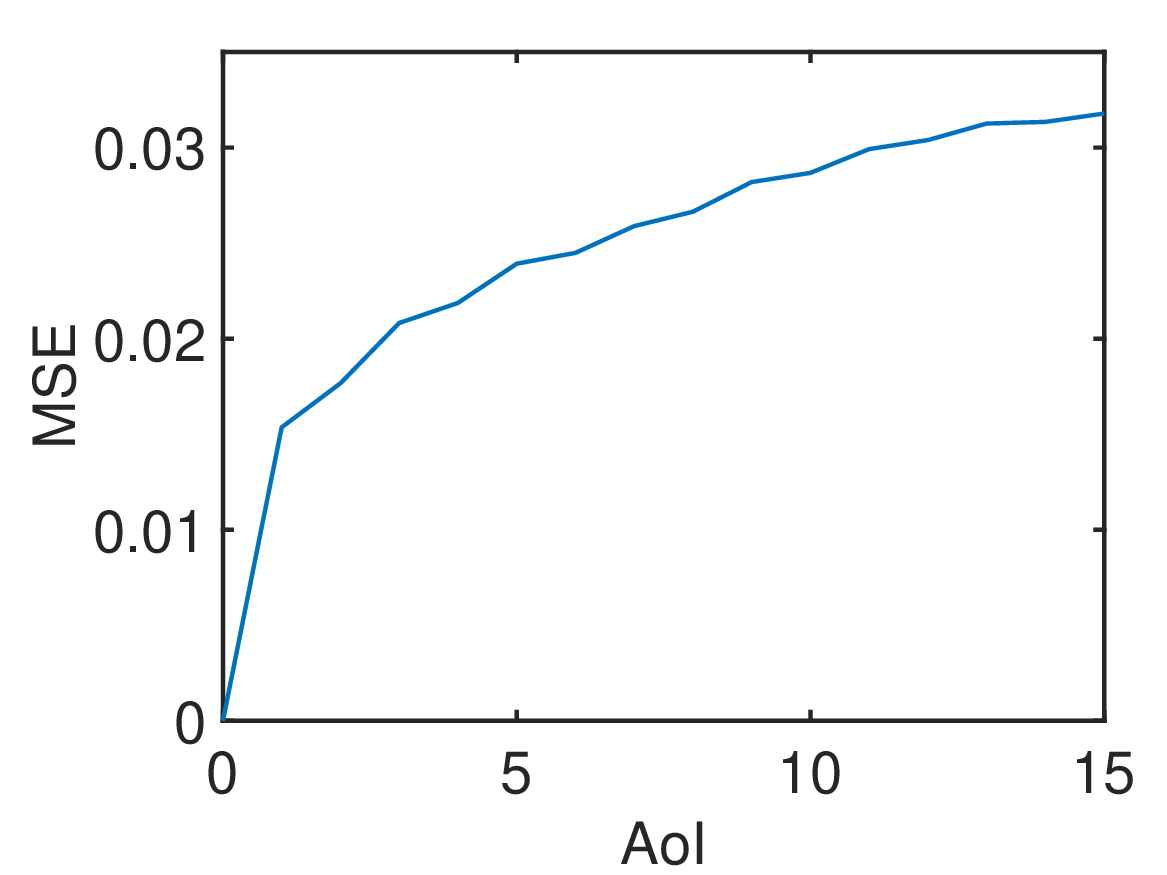}
\subcaption{Inference Error vs. AoI}
\end{subfigure}
\caption{\small Performance of video prediction experiment. The experimental results in (b) and (c) are regenerated from \cite{lee2018stochastic}. The training and inference errors are non-decreasing functions of the AoI. \label{fig:learning}}
\end{figure*}

{The concept of Age of Information (AoI) has attracted significant research interest; see, e.g., \cite{shisher2021age,kaul2012real,sun2017update, yates2015lazy, KamKompellaEphremidesTIT, kadota2018optimizing, soleymani2016optimal, SunSPAWC2018, SunNonlinear2019, chen2021uncertainty, wang2022framework, YinUpdateInfocom, orneeTON2021, Tripathi2019, klugel2019aoi, bedewy2021optimal,hsu2018age, sun2019closed, Kadota2018, ornee2023whittle, pan2023sampling, sun2023age, sun2022age, ornee2023context, ayan2023optimal} and a recent survey \cite{yates2021age}. Initially, research efforts were centered on analyzing and optimizing average AoI and peak AoI in communication networks \cite{kaul2012real, sun2017update, yates2015lazy, KamKompellaEphremidesTIT, kadota2018optimizing}. Recent research endeavors on AoI have shifted towards optimizing the performance of real-time applications, such as remote estimation \cite{orneeTON2021, ornee2023whittle, ornee2023context}, remote inference \cite{shisher2021age}, and control systems \cite{klugel2019aoi, ayan2023optimal}, by leveraging AoI as a tool. In \cite{soleymani2016optimal, SunSPAWC2018, SunNonlinear2019, wang2022framework}, information-theoretic metrics such as Shannon's mutual information (or Shannon's conditional entropy) has been used to quantify the amount of information carried by the received data  about the current source value (or the amount of uncertainty regarding the current source value) as the data ages. In addition, a Shannon’s conditional entropy term $H_{\mathrm{Shannon}}(Y_t|X_{t-\Delta(t)}=x, \Delta(t)=\delta)$ 
was used in \cite{chen2021uncertainty} to quantify information uncertainty given the most recent observation $X_{t-\Delta(t)}=x$ and AoI $\Delta(t)=\delta$. 
The information-theoretic metrics in these prior studies \cite{soleymani2016optimal, SunSPAWC2018, SunNonlinear2019, wang2022framework,chen2021uncertainty} cannot be directly used to evaluate real-world system performance.
To bridge the gap, in the present paper, we use an $L$-conditional entropy $H_L(Y_t|X_{t-\Delta(t)}, \Delta(t))$, to approximate and analyze the inference error in remote inference, as well as the estimation error in remote estimation. For example, when the loss function $L(y, \hat y)$ is chosen as a quadratic function $( y - \hat y)^2$, the $L$-conditional entropy $H_L(Y_t|X_{t-\Delta(t)}, \Delta(t))=\mathbb E[(Y_t-\mathbb E [ Y_t| X_{t-\Delta(t)}, \Delta(t)])^2]$ 
is exactly the minimum mean squared estimation error in signal-agnostic remote estimation. This approach takes a significant step to bridge the gap between AoI metrics and real-world applications, by directly mapping  AoI to the application performance metrics. 

In earlier AoI studies \cite{yates2015lazy,sun2017update, YinUpdateInfocom, SunNonlinear2019, SunSPAWC2018, orneeTON2021, Tripathi2019, klugel2019aoi, bedewy2021optimal, kadota2018optimizing, hsu2018age, sun2019closed, Kadota2018}, it was usually assumed that the observed data sequence is Markovian and the performance degradation caused by information aging was modeled as a monotonic AoI function.
Hence, the studies \cite{yates2015lazy, YinUpdateInfocom, sun2017update, SunNonlinear2019,  bedewy2021optimal, SunSPAWC2018, orneeTON2021} adopted  
the ``generate-at-will" status updating model, where the transmitter can only send the most recently generated signal value. However,  practical data sequence may not be Markovian \cite{SunNonlinear2019}. In the present paper, we propose a new local information geometric approach to analyze both Markovian and non-Markovian time-series data. For non-Markovian data, fresher data is not always better. To that end, we propose a new status updating model called the ``selection-from-buffer" model, where the transmitter has the option to send any of the $B$ most recent features stored in the source's buffer.

The optimization of linear and non-linear functions of AoI for multi-source scheduling is a restless multi-armed bandit (RMAB) problem. The multi-source problems in previous AoI studies \cite{Tripathi2019, kadota2018optimizing, hsu2018age, ornee2023whittle, sun2019closed, Kadota2018} are RMABs with binary actions and focused on monotonic AoI functions, where Whittle index policy \cite{whittle1988restless} is used to solve the problem. 
Our multi-source problem is an RMAB with multiple actions. 
Because of the multiple-action setup, the Whittle index alone can not be utilized to solve our problem. Consequently, we design a new asymptotically optimal policy for multi-action RMAB with general AoI functions (monotonic or non-monotonic). 

This paper is also related to the field of signal-agnostic remote estimation. The prior studies \cite{SunNonlinear2019, orneeTON2021, ornee2023whittle, Ornee2021performance, SunTIT2020, klugel2019aoi, pan2023sampling} in signal-agnostic remote estimation focused on Gaussian and Markovian processes. The results presented in the current paper are applicable to more general processes.}

%% file: LearningModel.tex
\section{Information Freshness in Remote Inference: Model and Performance}\label{RemoteInference}

\subsection{Remote Inference Model}
Consider the remote inference system illustrated in Fig. \ref{fig:learning}. In this system, a time-varying target $Y_t\in \mathcal Y$ (e.g., the position of the car in front) is predicted at time $t$, using a feature $X_{t-\Delta(t)}\in \mathcal X$ (e.g., a video clip) that was generated $\Delta(t)$ seconds ago at a sensor (e.g., a camera). The time difference $\Delta(t)$ between $X_{t-\Delta(t)}$ and $Y_t$ is the AoI defined in \eqref{AoIIntro}. Each feature $X_t = (V_t, V_{t-1}, \ldots, V_{t-u+1})$ is a time series of length $u$, extracted from the sensor's output signal $V_t$. For example, if $V_t$ is the video frame at time $t$, then $X_t$ represents a video clip consisting of $u$ consecutive video frames.


We focus on a class of popular supervised learning algorithms known as \emph{Empirical Risk Minimization (ERM)} \cite{goodfellow2016deep}. {In freshness-aware ERM supervised learning algorithms, a neural network is trained to generate an action $a = \phi(X_{t-\Delta(t)},\Delta(t)) \in \mathcal A$, where $\phi: \mathcal X \times \mathbb Z^{+} \mapsto \mathcal A$ is a function that maps a feature $X_{t-\Delta(t)}\in\mathcal X$ and its AoI $\Delta(t) \in \mathbb Z^{+}$ to an action $a\in \mathcal A$.} The performance of learning is evaluated using a  loss function $L: \mathcal Y \times \mathcal A \mapsto \mathbb R$, where $L(y,a)$ represents the loss incurred if action $a$ is selected when $Y_t=y$. It is assumed that both $\mathcal Y$ and $\mathcal X$ are discrete and finite sets. 

\begin{figure*}[ht]
  \centering
  \begin{subfigure}[t]{0.25\textwidth}
\includegraphics[width=\textwidth]{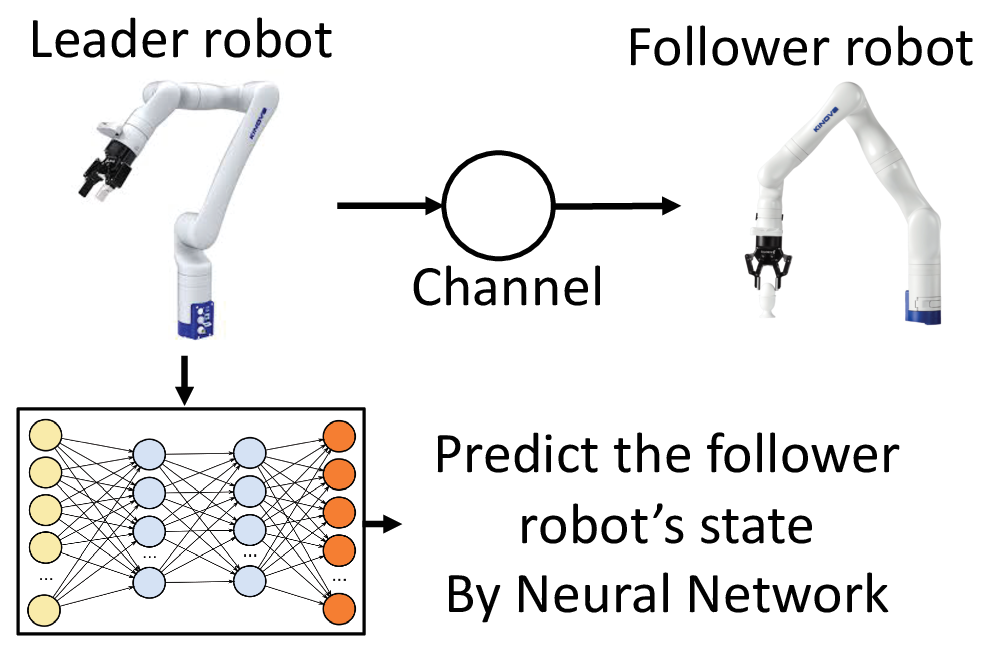}
  \subcaption{Prediction of Follower Robot}
\end{subfigure}
  \hspace{3mm}
\begin{subfigure}[t]{0.20\textwidth}
\includegraphics[width=\textwidth]{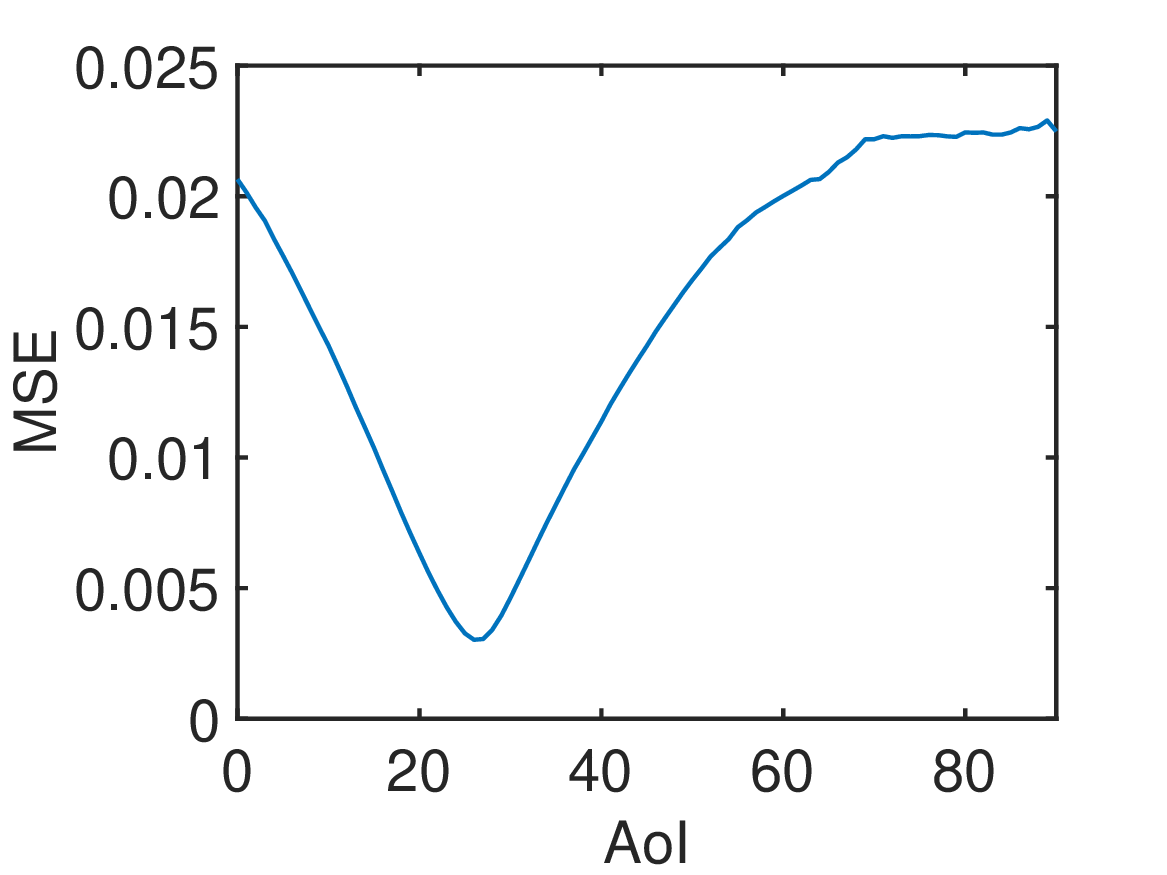}
  \subcaption{Training Error vs. AoI}
\end{subfigure}
  \hspace{3mm}
\begin{subfigure}[t]{0.20\textwidth}
\includegraphics[width=\textwidth]{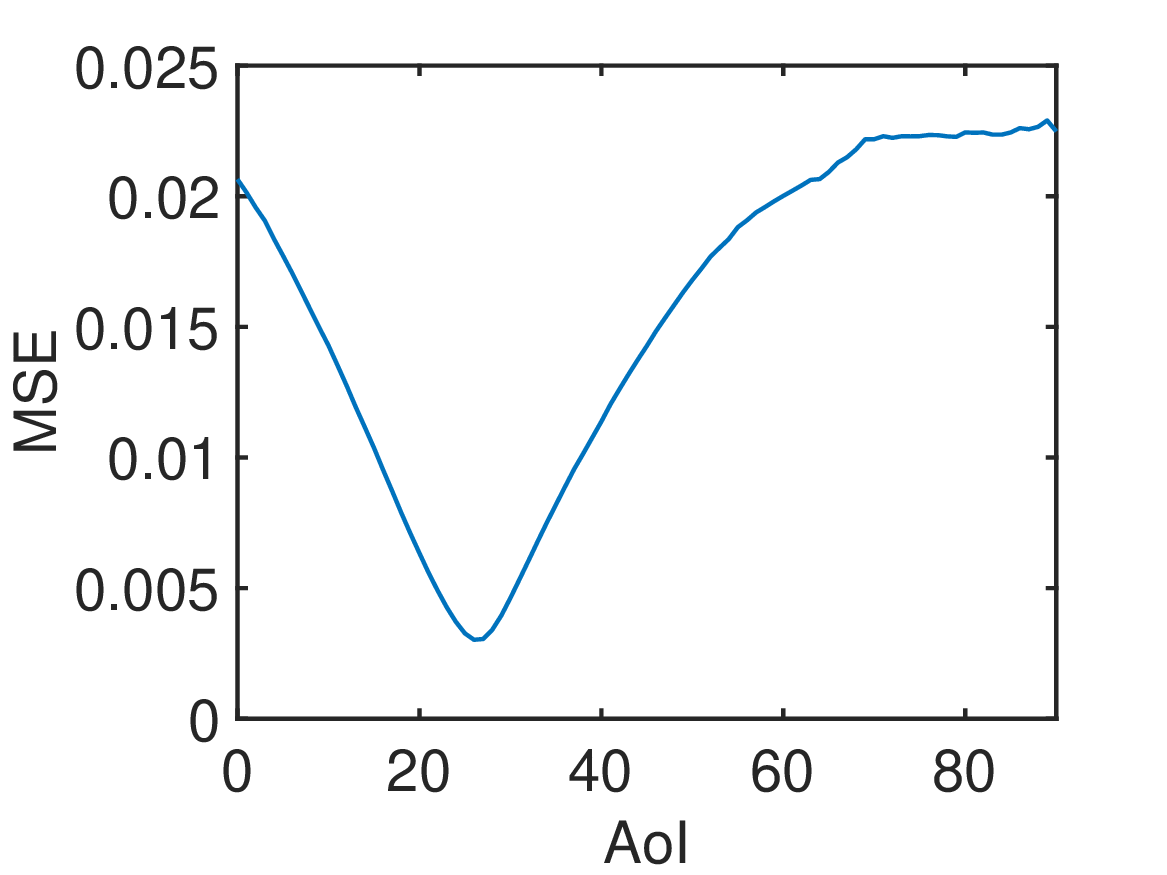}
  \subcaption{Inference Error vs. AoI}
\end{subfigure}
\caption{Robot state prediction in a leader-follower robotic system. The leader robot uses a neural network to predict the follower robot's state $Y_t$ by using the leader robot's state $X_{t-\delta}$ generated $\delta$ time slots ago $(u=1)$. The training and inference errors decrease in the AoI $\leq 25$ and increase when AoI $\geq 25$.}
\label{fig:DelayedNetworkedControlled}
\end{figure*}

The loss function $L$ is determined by the \emph{goal} of the remote inference system. For example, in neural network-based minimum mean-squared estimation, the loss function is $L_2(\mathbf y,\hat{\mathbf y}) =\|\mathbf y - \hat{\mathbf y}\|^2_2$, where the action $a = \hat {\mathbf y}$ is an estimate of the target $Y_t =\mathbf y$ and $\|\mathbf y\|^2_2$ is the Euclidean norm of the vector $\mathbf y$. In softmax regression (i.e., neural network-based maximum likelihood
classification), the action $a = Q_Y$ is a distribution of $Y_t$ and the loss function $L_{\text{log}}(y, Q_Y ) = - \text{log}~Q_Y (y)$ is the negative log-likelihood of the target value $Y_t = y$. 

\subsection{Offline Training and Online Inference}\label{learningerror}
A supervised learning algorithm consists of two phases: \emph{offline training} and \emph{online inference}. In the offline training phase, a neural network is trained using one of the following two approaches.  


In the first approach, multiple neural networks are trained independently, using distinct training datasets with different AoI values. The  neural network associated with an AoI value $\delta$ is trained by solving the following ERM problem:
\begin{align}\label{eq_trainingerror1}
\mathrm{err_{training,1}}(\delta) = \min_{\phi\in \Lambda} \mathbb{E}_{Y,X\sim P_{\tilde Y_0, \tilde X_{-\delta}}}[L(Y,\phi(X,\delta))],
\end{align} 
where $P_{\tilde Y_0, \tilde X_{-\delta}}$ is the empirical distribution of the label $\tilde Y_0$ and the feature $\tilde X_{-\delta}$ in the training dataset, the AoI value $\delta$ is the time difference between $\tilde Y_0$ and $\tilde X_{-\delta}$, and $\Lambda$ is the set of functions that can be constructed by the neural network. 

In the second approach, a single neural network is trained using a larger dataset that encompasses a variety of AoI values. The ERM training problem for this approach is formulated as 
\begin{align}\label{eq_trainingerror2}
\mathrm{err_{training,2}}= \min_{\phi\in \Lambda} \mathbb{E}_{Y,X,\Theta\sim P_{\tilde Y_0, \tilde X_{-\Theta},\Theta}}[L(Y,\phi(X,\Theta))],
\end{align}
where $P_{\tilde Y_0, \tilde X_{- \Theta}, \Theta}$ is the empirical distribution of the label $\tilde Y_0$, the feature $\tilde X_{-\Theta}$, and the AoI $\Theta$ within the training dataset, and the AoI $\Theta$ is the time difference between $\tilde Y_0$ and $\tilde X_{-\Theta}$.

{In the online inference phase, the pre-trained neural network is used to predict the target $Y_t$ in real-time. We assume that the process $\{(Y_t, X_t), t=0, 1, 2, \ldots\}$ is stationary and the processes $\{(Y_t, X_t), t=0, 1, 2, \ldots\}$ and $\{\Delta(t), t=0, 1, 2, \ldots\}$ are independent. Under these assumptions, if $\Delta(t)=\delta$, the inference error at time $t$ can be expressed as a function of the AoI value $\delta$, i.e.,  
\begin{align}\label{instantaneous_err1} 
\!\!\!\!\mathrm{err_{inference}}(\delta)\!=\!\mathbb E_{Y, X \sim P_{Y_{t}, X_{t-\delta}}}\!\!\left[L\!\left(Y,\!\phi^*(X,\delta)\right)\right],
\end{align}
where $P_{Y_t, X_{t-\delta}}$ is the distribution of the target $Y_{t}$ and the feature $X_{t-\delta}$, and $\phi^*$ is the trained neural predictor. The proof of \eqref{instantaneous_err1} is provided in Appendix \ref{pinferenceerror}. {\blue In Sections \ref{Scheduling}-\ref{Multi-scheduling}, to minimize inference error, we will develop signal-agnostic transmission scheduling policies in which scheduling decisions are determined without using the knowledge of the signal value of the observed process.} If the transmission schedule is signal-agnostic, then $\{(Y_t, X_t), t=0, 1, 2, \ldots\}$ is independent of the AoI process $\{\Delta(t), t=0, 1, 2, \ldots\}$.} 
\subsection{Experimental Results on Information Freshness}\label{Experimentation}  
We conduct five remote inference experiments to examine how the training error and the inference error vary as the AoI increases. These experiments include (i) video prediction, (ii) robot state prediction in a leader-follower robotic system, (iii) actuator state prediction under mechanical response delay, (iv) temperature prediction, and (v) wireless channel state information prediction. In these experiments, we consider the quadratic loss function $L(\mathbf y, \hat{\mathbf y}) = \|\mathbf y - \hat{\mathbf y}\|^2_2$. Detailed settings of these experiments can be found in Appendix \ref{Experiments}. We present the experimental results of the first training method  in Figs. \ref{fig:learning}-\ref{fig:Trainingcsi}. Related codes and datasets are accessible in our GitHub repository.\footnote{\url{https://github.com/Kamran0153/Impact-of-Data-Freshness-in-Learning}} To illustrate the training error of the second training method as a function of the AoI $\delta$, one can simply assess the training error using the training data samples with the AoI value $\delta$. {\red When the neural network is
sufficiently wide and deep, the results of the two training methods are similar. Hence, the experimental results of the second training method are omitted.}

Fig. \ref{fig:learning} presents the training error and inference error of a video prediction experiment, where a video frame $V_t$ at time $t$ is predicted using a feature $X_{t-\delta}=(V_{t-\delta}, V_{t-\delta-1})$ that is composed of two consecutive video frames. One can observe from Fig. \ref{fig:learning}(b)-(c) that both the training error and the inference error increase as the AoI $\delta$ increases.

Fig. \ref{fig:DelayedNetworkedControlled} plots the performance of robot state prediction in a leader-follower robotic system, where a leader robot uses a neural network to predict the follower robot's state $Y_t$ by using the leader robot's state $X_{t-\delta}$ generated $\delta$ time slots ago. As depicted in Fig. \ref{fig:DelayedNetworkedControlled}, the training and the inference errors decrease in AoI, when AoI $\leq 25$ and increase in AoI when AoI $\geq 25$. In this case, even a fresh feature with AoI $=0$ is not good for prediction. 

The performance of actuator state prediction under mechanical response delay is depicted in Fig. \ref{fig:TrainingCartVelocity}. We consider the OpenAI CartPole-v1 task \cite{brockman2016openai}, where the objective is to control the force on a cart and prevent the pole attached to the cart from falling over. The pole angle $\psi_t$ at time $t$ is predicted based on a feature $X_{t-\delta}=(v_{t-\delta}, \ldots, v_{t-\delta-u+1})$ that consists of a consecutive sequence of cart velocity with length $u$ generated $\delta$ milliseconds (ms) ago. As shown in Fig. \ref{fig:TrainingCartVelocity},  both the training error and the inference error exhibit non-monotonic variations as the AoI $\delta$ increases.

In Fig. \ref{fig:Training} and Fig. \ref{fig:Trainingcsi}, we plot the results of temperature prediction and wireless channel state information (CSI) prediction experiments, respectively. In both experiments, we observe non-monotonic trends in training error and inference error with respect to AoI, particularly when the length of the feature sequence $u$ is small.

In the AoI literature, it has been generally assumed that the performance of real-time systems degrades monotonically as the data becomes stale. However, Figs. \ref{fig:learning}-\ref{fig:Trainingcsi} reveal that this assumption is true in some scenarios, and not true in some other scenarios. Furthermore, Figs. \ref{fig:DelayedNetworkedControlled}-\ref{fig:TrainingCartVelocity} show that even the fresh data with AoI $= 0$ may generate a larger inference error than stale data with AoI $> 0$. These counter-intuitive experimental results motivated us to seek theoretical interpretations of information freshness in subsequent sections.




\section{An Information-theoretic Interpretation of Information Freshness in Remote Inference}\label{InformationAnalysis}
In this section, we develop an information-theoretic approach to interpret information freshness in remote inference.


\begin{figure*}
  \centering
  \begin{subfigure}[t]{0.20\textwidth}
\includegraphics[width=\textwidth]{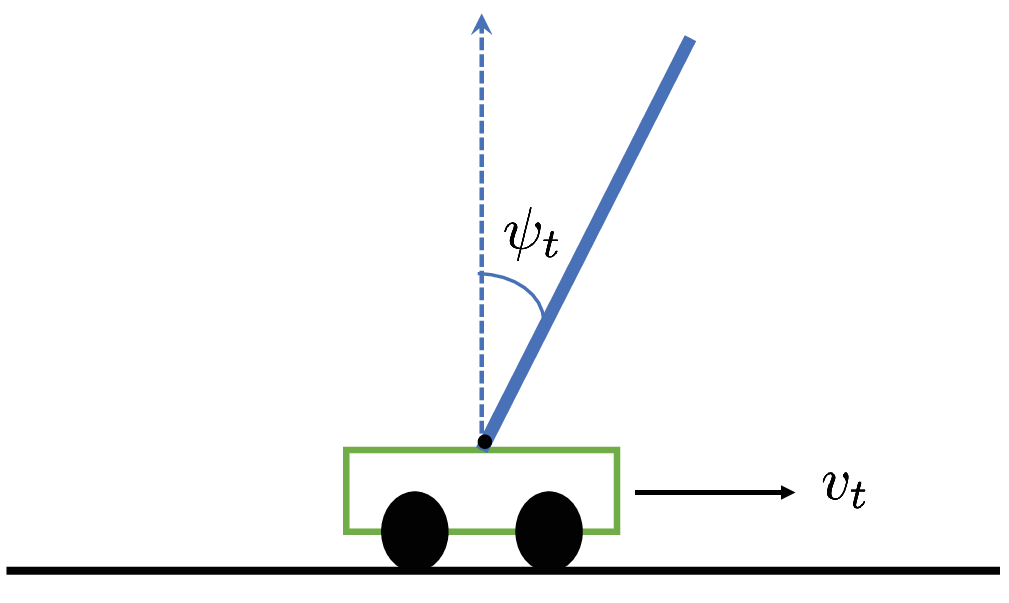}
  \subcaption{OpenAI Cart Pole  Task}
\end{subfigure}
%
\hspace{0mm}
\begin{subfigure}[t]{0.20\textwidth}
\includegraphics[width=\textwidth]{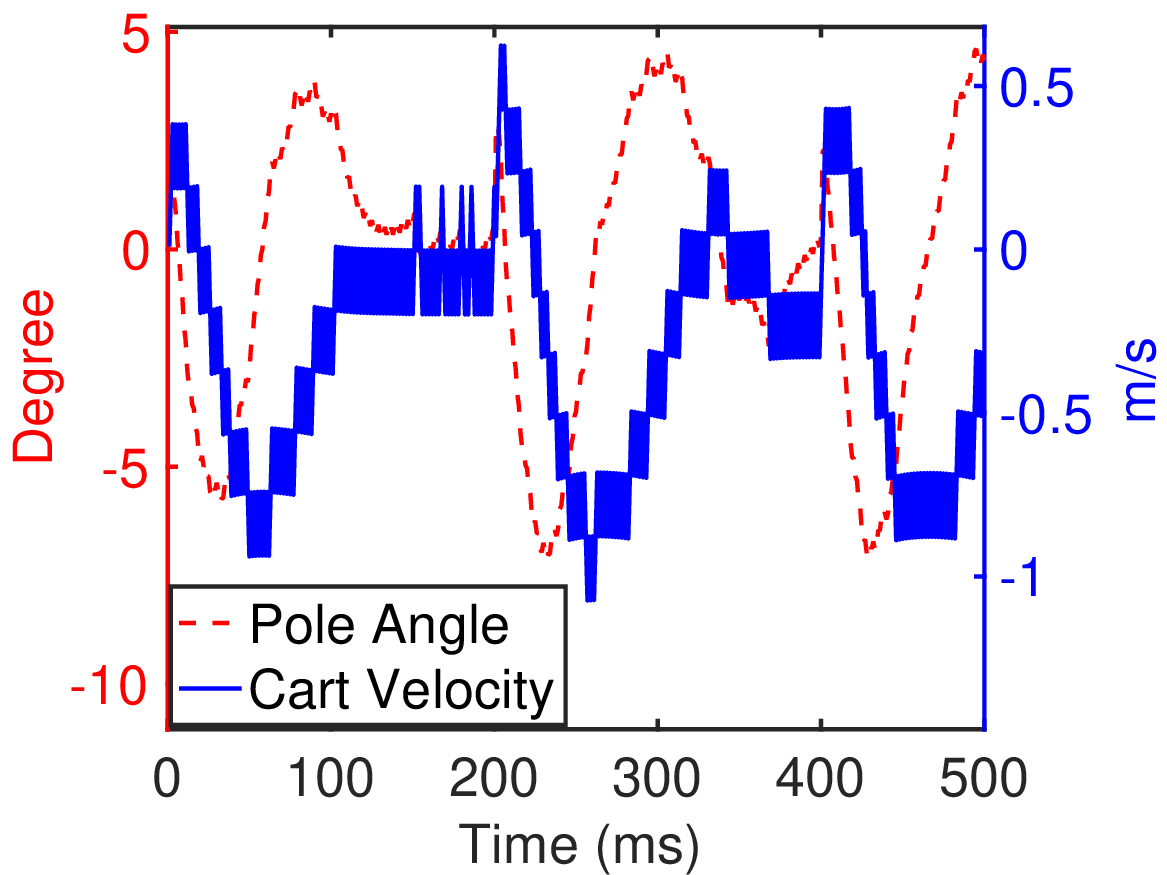}
  \subcaption{Data Traces}
\end{subfigure}
%
  \hspace*{0mm} 
\begin{subfigure}[t]{0.20\textwidth}
\includegraphics[width=\textwidth]{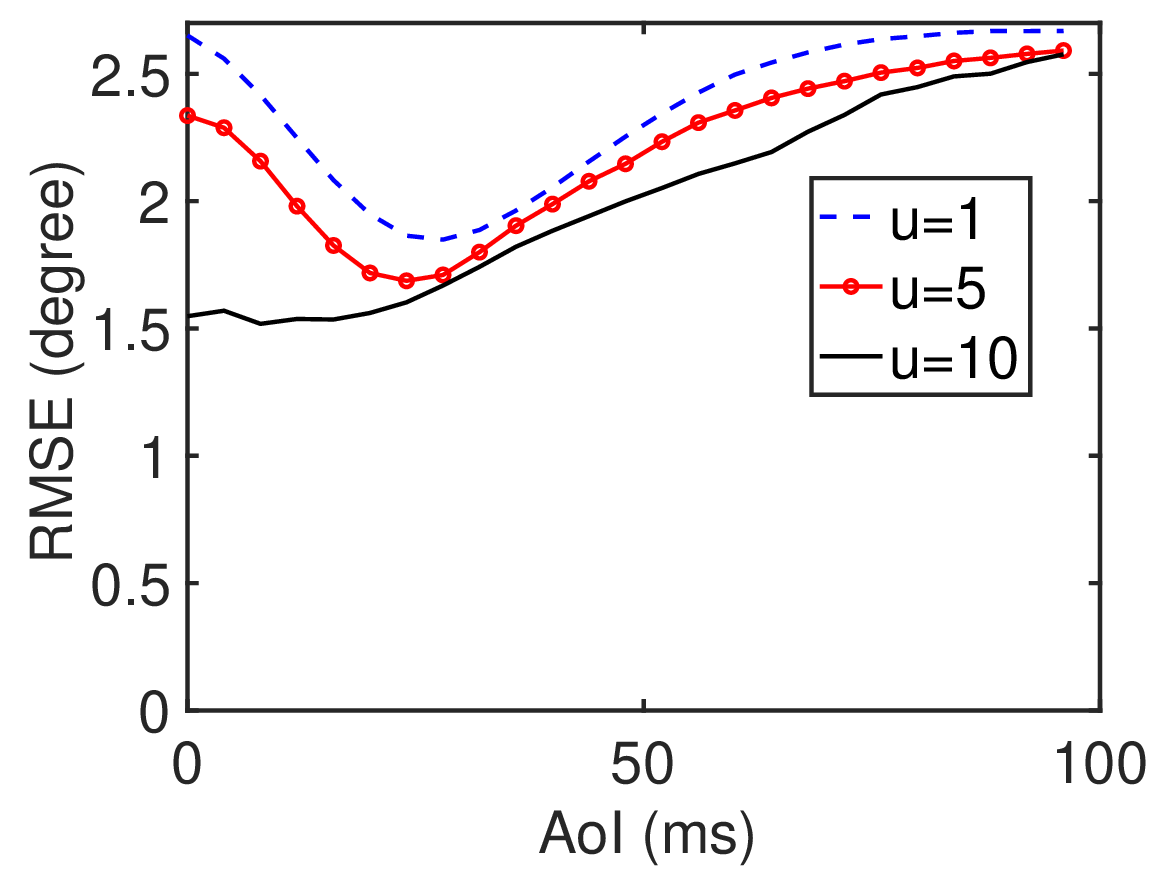}
  \subcaption{Training Error vs. AoI}
\end{subfigure}
%
\hspace{0mm}
\begin{subfigure}[t]{0.20\textwidth}
\includegraphics[width=\textwidth]{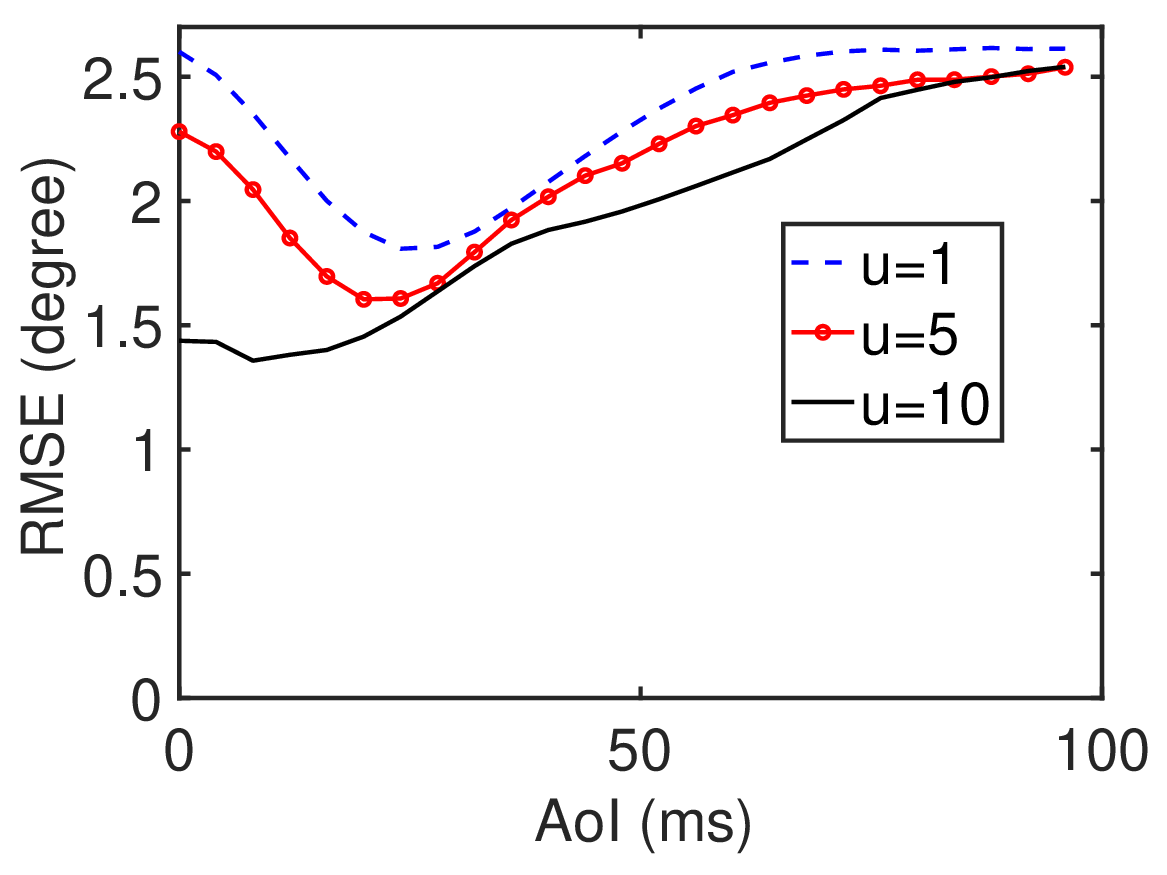}

  \subcaption{Inference Error vs. AoI}
\end{subfigure}
\caption{Performance of actuator state prediction under mechanical response delay. In the OpenAI CartPole-v1 task \cite{brockman2016openai}, the pole angle $\psi_t$ is predicted by using $X_{t-\delta}=(v_{t-\delta}, v_{t-\delta-1}, \ldots, v_{t-\delta-u-1})$, where $v_{t}$ is the cart velocity at time $t$. 
The training error and inference error are non-monotonic in the AoI.
}
\label{fig:TrainingCartVelocity}
\end{figure*}

\subsection{Information-theoretic Metrics for Training and Inference}\label{information-theoreticmetrics}
Because the set of functions $\Lambda$ constructed by the neural network  is complicated, it is difficult to directly analyze the training and inference errors by using \eqref{eq_trainingerror1}-\eqref{instantaneous_err1}. To overcome this challenge, we introduce information-theoretic metrics for the training and inference errors.
\subsubsection{Training Error of the First Training Approach}
Let $\Phi=\{f : \mathcal X \times \mathbb Z^{+} \mapsto \mathcal A\}$ be the set of all functions mapping
from $\mathcal X \times \mathbb Z^{+}$ to $\mathcal A$. Any function $\phi$ constructed by the neural network belongs to $\Phi$. 
Hence, $\Lambda \subseteq \Phi$. By relaxing the set $\Lambda$ in \eqref{eq_trainingerror1} as $\Phi$, we obtain the following lower bound of $\mathrm{err_{training,1}}(\delta)$:
\begin{align}\label{eq_TrainingErrorLBfixedAoi}
H_L(\tilde Y_0| \tilde X_{-\delta}) = \min_{\phi\in \Phi} \mathbb{E}_{Y,X\sim P_{\tilde Y_0, \tilde X_{-\delta}}}\![L(Y,\phi(X,\delta))], 
\end{align} 
where $H_L( Y| X)$ is a generalized conditional entropy of $ Y$ given $X$, defined by  \cite{Dawid2004, Dawid1998,farnia2016minimax}
\begin{align}\label{eq_TrainingErrorLB}
H_L( Y| X)= \min_{\phi(x, \delta) \in\mathcal A,~\! \forall x \in \mathcal X}  \mathbb E_{Y,X\sim P_{Y,X}} [L(Y, \phi(X, \delta))].
\end{align} 
Compared to $\mathrm{err_{training,1}}(\delta)$, its information-theoretic lower bound $H_L(\tilde Y_0| \tilde X_{-\delta})$ is mathematically more convenient to analyze. The gap between $\mathrm{err_{training,1}}(\delta)$ and $H_L(\tilde Y_0| \tilde X_{-\delta})$ was studied recently in \cite{shisher2022local}, where the gap is small if $\Lambda$ and $\Phi$ are close to each other, e.g., when the neural network is sufficiently wide and deep \cite{goodfellow2016deep}.

For notational convenience, we refer to $H_L( Y| X)$ as an \emph{L-conditional entropy}, because it is associated with a loss function $L$. The \emph{$L$-entropy} of a random variable $Y$ is defined as \cite{Dawid2004, farnia2016minimax}
\begin{align}\label{eq_Lentropy}
H_L(Y) = \min_{a\in\mathcal A} \mathbb{E}_{Y \sim P_{Y}}[L(Y,a)].
\end{align} 
The $L$-conditional entropy of $Y$ given $X=x$ is 
\begin{align}\label{given_L_condentropy}
H_L(Y| X=x)= \min_{a \in\mathcal A} \mathbb E_{Y\sim P_{Y| X=x}} [L(Y, a)].
\end{align}
{\blue Using \eqref{eq_TrainingErrorLB}, one can get 
\cite{Dawid2004, farnia2016minimax} 
\begin{align}\label{eq_cond_entropy1}
H_L(Y|X)&=\sum_{x \in \mathcal X} P_X(x) \min_{a \in\mathcal A} \mathbb E_{Y\sim P_{Y| X=x}} [L(Y, a)] \nonumber\\
&=\sum_{x \in \mathcal X} P_X(x) H_L(Y| X=x). 
\end{align}}
\subsubsection{Training Error of the Second Training Approach}
A lower bound of the training error $\mathrm{err_{training,2}}$ in \eqref{eq_trainingerror2} is 
\begin{align}\label{eq_TrainingErrorLB1}
&H_L(\tilde Y_0| \tilde X_{-\Theta},\Theta)\nonumber\\
=&\min_{\phi\in \Phi} \mathbb{E}_{Y,X,\Theta\sim P_{\tilde Y_0, \tilde X_{-\Theta},\Theta}}\![L(Y,\!\phi(X,\Theta))], 
\end{align} 
where $H_L(\tilde Y_0| \tilde X_{-\Theta},\Theta)$ is a $L$-conditional entropy of $\tilde Y_0$ given $(\tilde X_{-\Theta},\Theta)$. Using \eqref{eq_cond_entropy1}, $H_L(\tilde Y_0| \tilde X_{-\Theta},\Theta)$ can be decomposed as
\begin{align}\label{eq_TrainingErrorLB1}
\!\!\!&H_L(\tilde Y_0| \tilde X_{-\Theta},\Theta) \nonumber\\
=&\sum_{\substack{x \in \mathcal X, \delta \in \mathbb Z^{+}}} \!\!\!\! P_{\tilde X_{-\Theta}, \Theta}(x, \delta) H_L(\tilde Y_0|\tilde X_{-\delta}=x, \Theta=\delta).\!\!\!
\end{align}
{Similar to Sec. \ref{learningerror}, we assume that the label and feature $(\tilde Y_0,\tilde X_{-k})$ in the training dataset are independent of the training AoI $\Theta$ for every $k \geq 0$. Under this assumption, \eqref{eq_TrainingErrorLB1} can be simplified as (see Appendix \ref{pfreshness_aware_cond} for its proof)
\begin{align}\label{freshness_aware_cond}
\!\! H_L(\tilde Y_0| \tilde X_{-\Theta},\Theta)\!=\sum_{\delta \in \mathbb Z^{+}} P_{\Theta}(\delta)~H_L(\tilde Y_0| \tilde X_{-\delta}),
\end{align}
which connects the information-theoretic lower bounds of $\mathrm{err_{training,1}}(\delta)$ and $\mathrm{err_{training,2}}$.
}
\subsubsection{Inference Error}
Let $a_{P_Y}$ be an optimal solution to \eqref{eq_Lentropy}, called a \emph{Bayes action} \cite{Dawid2004}.
If the neural predictor in \eqref{instantaneous_err1} is replaced by the Bayes action $a_{P_{\tilde Y_0|\tilde X_{-\delta}=x}}$, 
then, for both training methods, $\mathrm{err_{inference}}(\delta)$ becomes the following $L$-conditional cross entropy
\begin{align}\label{L-CondCrossEntropy}
&H_L\left(P_{Y_{t}|X_{t- \delta}}; P_{\tilde Y_{0} |\tilde X_{- \delta}} \Big| P_{X_{t-\delta}}\right)\!\! \nonumber\\
=&\sum_{x \in \mathcal X} \!\!P_{X_{t-\delta}}(x)\mathbb{E}_{Y \sim P_{Y_t| X_{t- \delta}=x}}\!\!\left[ L\left(Y,a_{P_{\tilde Y_0|\tilde X_{-\delta}=x}}\right)\right]\!,\!\!\!\!
\end{align} 
where the \emph{$L$-cross entropy} is defined as 
\begin{align} \label{cross-entropy}
H_L(P_Y; P_{\tilde Y})= \mathbb{E}_{Y \sim P_{Y}}\left[L\left(Y, a_{P_{\tilde Y}}\right)\right],
\end{align}
and the \emph{$L$-conditional cross entropy} is defined as
\begin{align} \label{cond-cross-entropy}
&H_L(P_Y; P_{\tilde Y} | P_X)\nonumber\\
=& \sum_{x \in \mathcal X} P_X(x) \mathbb{E}_{Y \sim P_{Y|X=x}}\!\!\left[L\!\left(Y, a_{P_{\tilde Y|\tilde X=x}}\right)\right]\!.
\end{align}
If the function spaces $\Lambda$ and $\Phi$ are close to each other, the difference between $\mathrm{err_{inference}}(\delta)$ and the $L$-conditional cross entropy $H_L(P_{Y_{t}|X_{t- \delta}}; P_{\tilde Y_{0}|\tilde X_{- \delta}} | P_{X_{t-\delta}})$ is small.

\ignore{We assume that $\{(Y_t , X_t ), t \in \mathbb Z\}$ is a stationary process that is independent of $\{\Delta(t), t \in \mathbb Z\}$. If the neural predictor in \eqref{eq_inferenceerror} is replaced by the optimal solution to \eqref{eq_TrainingErrorLB1}, then $\mathrm{err}_{\mathrm{inference}}$ becomes an $L$-conditional cross-entropy
\begin{align}\label{L-CondCrossEntropy}
&H_L(Y_{t}; \tilde Y_{t} | X_{t-\Delta}, \Delta) \nonumber\\
=&\mathbb{E}_{Y,X,\Delta\sim P_{Y_t, X_{t- \Delta},\Delta}}\left[L\left(Y,\hat \phi_{P_{\tilde Y_t, \tilde X_{t-\Theta},\Theta}}(X,\Delta)\right)\right],
\end{align} 
where $\hat \phi_{P_{\tilde Y_t, \tilde X_{t-\Theta},\Theta}}$ is the optimal solution to \eqref{eq_TrainingErrorLB1} and $\Delta$ is a random variable that follows the empirical distribution of $\{\Delta(t), t \in \mathbb Z\}$.}
\ignore{If the function space $\Lambda$ is sufficiently large, the difference between $\mathrm{err}_{\mathrm{inference}}$ and $H_L(\tilde Y_{t}; Y_{t} | X_{t-\Delta}, \Delta)$ is small.} 

Examples of loss function $L$, $L$-entropy, and $L$-cross entropy are provided in Appendix \ref{InformationTheory1}. Additionally, the definitions of $L$-divergence $D_L(P_Y || Q_Y)$, $L$-mutual information $I_L(Y; X)$, and $L$-conditional mutual information $I_L(Y; X| Z)$ are provided in Appendix \ref{otherLmetrics}.  
The relationship among $L$-divergence, Bregman divergence \cite{dhillon2008matrix}, and $f$-divergence \cite{csiszar2004information} is discussed in Appendix \ref{InformationTheory2}.
\begin{figure}[t]
  \centering
\begin{subfigure}[b]{0.20\textwidth}
\includegraphics[width=1\linewidth]{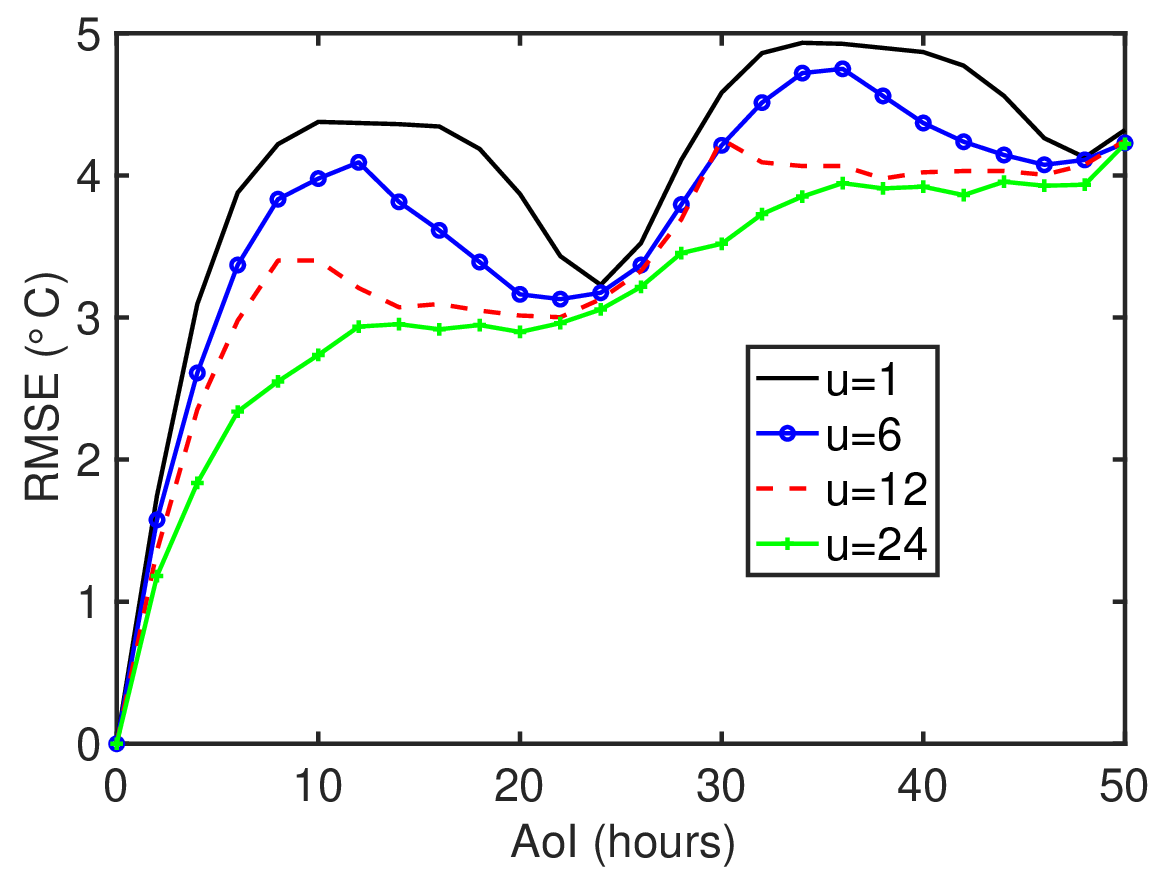}
  \subcaption{Training Error vs. AoI}
\end{subfigure}
\begin{subfigure}[b]{0.20\textwidth}
\includegraphics[width=1\linewidth]{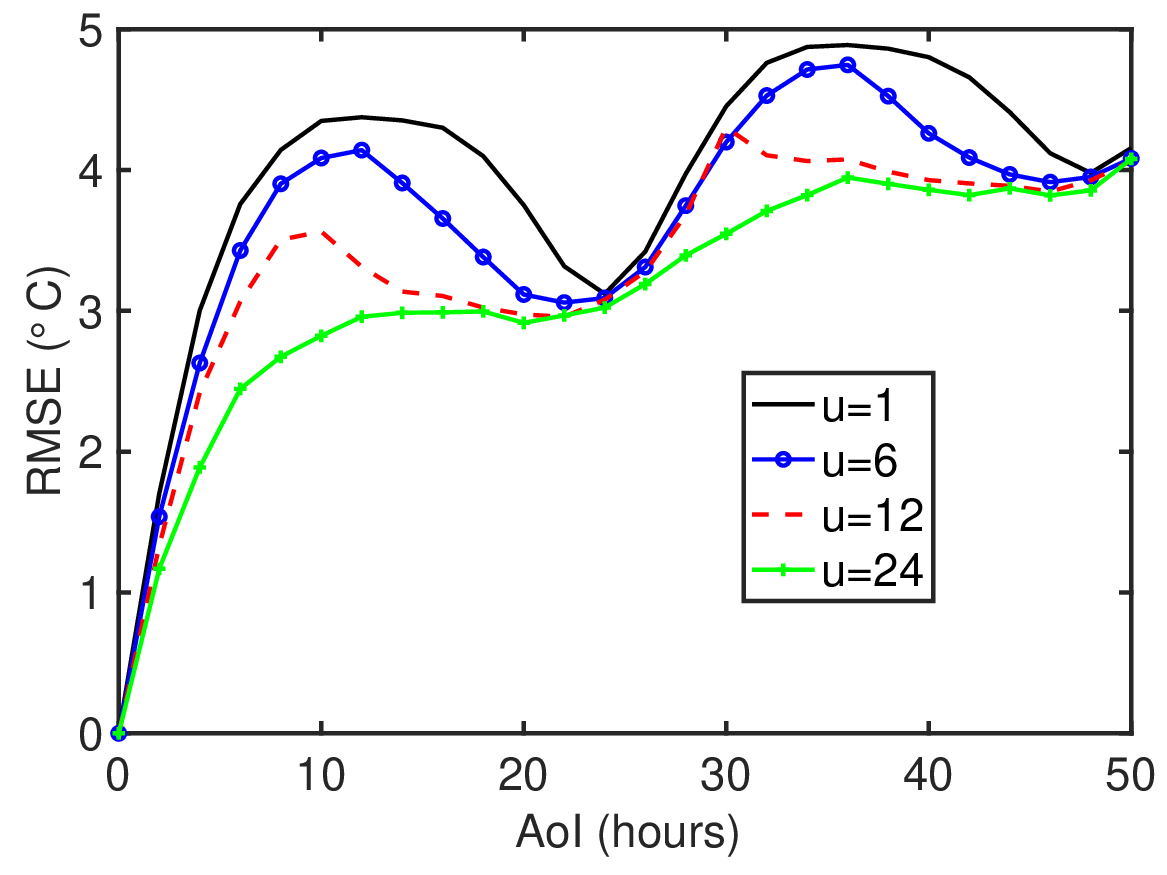}
  \subcaption{Inference Error vs. AoI}
\end{subfigure}
%
\caption{Performance of temperature prediction. The training error and inference error are non-monotonic in AoI. As the feature sequence length $u$ increases, the errors tend closer to non-decreasing functions of the AoI.
}
\label{fig:Training}
\end{figure}

\ignore{
\ifreport
The $L$-\emph{divergence} $D_L(P_{Y} || P_{\tilde Y})$ of $P_{Y}$ from $P_{\tilde Y}$ can be expressed as \cite{Dawid2004, farnia2016minimax}
\begin{align}\label{divergence}
&D_L(P_{Y} || P_{\tilde Y}) \nonumber\\
 &=\!\mathbb E_{Y \sim P_{Y}}\left[L\left(Y, a_{P_{Y}}\right)\right]-\mathbb E_{Y \sim P_{Y}}\left[L\left(Y, a_{P_{\tilde Y}}\right)\right].
\end{align}
Because $a_{P_{Y}}$ is an optimal to \eqref{eq_Lentropy}, from \eqref{divergence}, we have 
\begin{align}
D_L(P_{Y} || P_{\tilde Y}) \geq 0.
\end{align}
The \emph{$L$-mutual information} $I_L(Y;X)$ is defined as \cite{Dawid2004, farnia2016minimax}
\begin{align}\label{MI}
I_L(Y; X)=& \mathbb E_{X \sim P_{X}}\left[D_L\left(P_{Y|X}||P_{Y}\right)\right]\nonumber\\
=&H_L(Y)-H_L(Y|X) \geq 0,
\end{align}
which measures the performance gain in predicting $Y$ by observing $X$. In general, $I_L(Y;X)$ $\neq$ $I_L(X;Y)$. The $L$-conditional mutual information $I_L(Y; X | Z)$ is given by 
\begin{align}\label{CMI}
I_L(Y; X|Z)=& \mathbb E_{X, Z \sim P_{X, Z}}\left[D_L\left(P_{Y|X, Z}||P_{Y | Z}\right)\right]\nonumber\\
=&H_L(Y | Z)-H_L(Y|X, Z) \geq 0.
\end{align}
\fi}

\ignore{\ifreport
Since $(Y_0, X_{\Theta})$ is independent of $\Theta$ and $\{(\tilde Y_t, \tilde X_t), t \in \mathbb Z\}$ is independent of $\Delta$, the $L$-conditional cross entropy in \eqref{L-CondCrossEntropy} can be decomposed as 
(See Appendix \ref{pDecomposed_Cross_entropy} for its proof)
\begin{align}\label{Decomposed_Cross_entropy}
H_L(\tilde Y_t;Y_t| \tilde X_{t-\Delta}, \Delta) = \sum_{\delta \in \mathcal D} P_\Delta(\delta)~H_L(\tilde Y_t;Y_t| \tilde X_{t-\delta}).
\end{align}
{\violet Similar to \eqref{divergence}, we can get
\begin{align}\label{L-ConCrossEntropy1}
\!\!\!\!\!H_L(\tilde Y_{t}; Y_{t} | \tilde X_{t-\Delta}, \Delta) =&H_L(\tilde Y_{t} | \tilde X_{t-\Delta}, \Delta)+\sum_{x \in \mathcal X, \delta \in \mathcal D} P_{\tilde X_{t- \Delta},\Delta}(x, \delta)\nonumber\\
&\!\!\times D_L\left(P_{\tilde Y_t | \tilde X_{t-\Delta}=x, \Delta=\delta} || P_{Y_t | X_{t-\Theta}=x, \Theta=\delta}\right).\!\!\!
\end{align}
Because the $L$-divergence in \eqref{L-ConCrossEntropy1} is non-negative, we have}
\begin{align}\label{lowerbound_inference}
H_L(\tilde Y_{t}; Y_{t} | \tilde X_{t-\Delta}, \Delta) \geq H_L(\tilde Y_{t} | \tilde X_{t-\Delta}, \Delta).
\end{align}
\fi}

\ignore{Consider the real-time forecasting system illustrated in Fig. \ref{fig:learning}, where the goal is to predict a fresh label $Y_t$ (e.g., location of the car in front) of time $t$ based on {\blue an earlier}  observation $X_{t-\Delta(t)}$ that was generated $\Delta(t)$ seconds ago.~The observation, a.k.a., feature, $X_{t-\Delta(t)}$$=$$(s_{t-\Delta(t)}$, $\ldots, s_{t-\Delta(t)-u+1})$ is a {\blue time sequence of length $u$ (e.g., $u$ consecutive video frames). We consider} a class of supervised learning algorithms called Empirical Risk Minimization (ERM), which is a standard approach for supervised learning \cite{vapnik2013nature, goodfellow2016deep, mohri2018foundations}. In ERM, the decision-maker predicts {\blue the label} $Y_t\in\mathcal Y$ by taking an action $a$ $=$ $\phi(X_{t-\Delta(t)},\Delta(t))$$\in \mathcal A$ based on  the observation $X_{t-\Delta(t)}\in\mathcal X$ and {\blue its AoI $\Delta(t) \in \mathcal D$, where we assume that $\mathcal Y$, $\mathcal X$, and $\mathcal D$ are discrete sets.}~The learning performance is measured by a loss function $L$, where $L(y,a)$ is the incurred loss if action $a$ is chosen when $Y_t=y$. The loss function $L$ is specified by the ERM supervised learning algorithm.~For example, {\blue $L$ is a quadratic function $L_2(y,\hat y) = (y - \hat y)^2$ in linear regression and a logarithmic function $L_{\log}(y,P_Y) = - \log P_Y(y)$ in logistic regression, where $P_Y$ denotes the distribution of $Y$.}

Supervised learning based real-time forecasting consists of two phases: \emph{offline training}
and \emph{online inference}.~In the offline training phase, a training dataset is collected and is used to train a neural network. Let $P_{Y_t, X_{t- \Theta}, \Theta}$ denote the empirical distribution of the training data $(Y_t,X_{t-\Theta})$ and training AoI $\Theta$, {\blue and $(Y_t,X_{t-\Theta},\Theta)$ are random variables following this empirical distribution. Here, the training AoI $\Theta \geq 0$ is the time difference between the observation $X_{t-\Theta}$ and the label $Y_t$. We assume that the training label and feature $(Y_t ,X_t )$ are stationary over time and are independent of the training AoI $\Theta$.} The objective of training in ERM-based real-time forecasting is to solve the following problem:  
\begin{align}\label{eq_trainingerror}
\mathrm{err}_{\mathrm{training}} = \min_{\phi\in \Lambda} \mathbb{E}_{Y,X,\Theta\sim P_{Y_t,X_{t-\Theta},\Theta}}[L(Y,\phi(X,\Theta))],
\end{align} 
where $\phi: \mathcal X \times \mathcal D \mapsto \mathcal A$ is selected from a family of decision functions $\Lambda$ that can be implemented by a neural network, $$\mathbb{E}_{Y,X,\Theta\sim P_{Y_t,X_{t-\Theta},\Theta}}[L(Y,\phi(X,\Theta))]$$ is {\blue the expected loss over the empirical distribution of training data and training AoI, and} $\mathrm{err}_{\mathrm{training}}$ is {\blue called} the \emph{minimum training error} \footnote{In this paper, we focus on \emph{freshness-aware inference}, in which the AoI is fed as an input of the neural predictor (see Fig. \ref{fig:learning}) and is used to predict the label. {\violet The complimentary case of \emph{freshness-agnostic inference}, where the AoI is unknown at the neural predictor, is out of the scope of this paper. The difference between \emph{freshness-aware inference} and \emph{freshness-agnostic inference} was briefly discussed in \cite{shisher2021age} and will be further studied in our future work.}}.~The optimal solution to \eqref{eq_trainingerror} is denoted by $\phi^*_{P_{Y_t, X_{t-\Theta},\Theta}}$.

Within the {online inference} phase, {\blue the trained neural predictor $\phi^*_{P_{Y_t, X_{t-\Theta},\Theta}}$ is used to predict the target in real-time.} The \emph{inference error} is the expected loss on the inference data and inference AoI using the trained predictor, i.e.,
\begin{align}\label{eq_inferenceerror}
\mathrm{err}_{\mathrm{inference}} = \mathbb{E}_{Y,X,\Delta\sim P_{\tilde Y_t,\tilde X_{t- \Delta},\Delta}}[L(Y,\phi^*_{P_{Y_t, X_{t-\Theta},\Theta}}(X,\Delta))],
\end{align} 
where $P_{\tilde Y_t,\tilde X_{t- \Delta},\Delta}$ is the {\violet empirical distribution} of the inference data $(\tilde Y_t,\tilde X_{t- \Delta})$ and inference AoI $\Delta$, and $(\tilde Y_t,\tilde X_{t- \Delta},\Delta)$ are random variables following this distribution. 

During online inference, the observations (e.g., video frames for prediction) are sent to the trained neural predictor in real-time. As a result, the inference AoI is a random process $\{\Delta(t), t \in \mathbb Z \}$ governed by the communications from a data source (e.g., a sensor or camera) to the neural predictor, and $\Delta$ follows the empirical distribution of the AoI process $\{\Delta(t),t\in \mathbb Z\}$. {\violet We assume that the inference label and feature process $\{(\tilde Y_t , \tilde X_t), t\in \mathbb Z \}$ is stationary over time and is independent of the inference AoI $\Delta(t)$ and $\Delta$. On the other hand, the training AoI $\Theta$ (i.e., the time difference between  $X_{t-\Theta}$ and  $Y_t$) can be arbitrarily chosen ahead of time because the training dataset is prepared offline.} In Section \ref{InformationAnalysis}, we will discuss how to choose the training AoI $\Theta$. In Section \ref{Scheduling}, we will study how to optimally control the inference AoI process $\Delta(t)$ by scheduling the transmissions of {\violet observation features} to the neural predictor.}

%% file: Interpretation.tex

\subsection{Training Error vs. Training AoI}\label{SecMinTrainingError}

We first analyze the monotonocity of the $L$-conditional entropy $H_L(\tilde Y_0| \tilde X_{-\delta})$ as $\delta$ increases. If $\tilde Y_0 \leftrightarrow \tilde X_{-\mu}  \leftrightarrow \tilde X_{-\mu-\nu}$ is a Markov chain for all $\mu,\nu\geq 0$, by the data processing inequality for $L$-conditional entropy \cite[Lemma 12.1] {Dawid1998}, $H_L(\tilde Y_0| \tilde X_{-\delta})$ is a non-decreasing function of $\delta$. 
Nevertheless, the experimental results in Figs. \ref{fig:learning}-\ref{fig:Trainingcsi} show that the training error is a growing function of the AoI $\delta$ in some systems (see Fig. \ref{fig:learning}), whereas it is a non-monotonic function of $\delta$ in other systems (see Figs. \ref{fig:DelayedNetworkedControlled}-\ref{fig:Trainingcsi}). As we will explain below, a fundamental reason behind these phenomena is that practical time-series data for remote inference could be either Markovian or non-Markovian. For non-Markovian $(\tilde Y_0, \tilde X_{-\mu}, \tilde X_{-\mu -\nu})$, $H_L(\tilde Y_0| \tilde X_{-\delta})$ is not necessarily monotonic in $\delta$.  

We propose a new relaxation of the data processing inequality to analyze information freshness for both Markovian and non-Markovian time-series data. To that end, the following generalization of the standard Markov chain model is needed, which is motivated by the $\epsilon$-dependence concept used in \cite{huang2019universal}.

\ignore{{\blue Let us first consider the case of deterministic training AoI $\Theta=\delta$. In this case, $H_L(Y_t| X_{t-\Theta},\Theta=\delta)$ can be simply written as $H_L(Y_t| X_{t-\delta})$. Because $\{(Y_t, X_t)\}_{t \in \mathbb Z}$ is stationary, $H_L(Y_t| X_{t-\delta})$ is a function of $\delta$.} If $Y_t \leftrightarrow X_{t-\mu}  \leftrightarrow X_{t-\mu-\nu}$ is a Markov chain for all $\mu,\nu\geq 0$, then the data processing inequality \cite[Lemma 12.1] {Dawid1998} implies that $H_L(Y_{t} | X_{t-\delta})$ is a {\blue non-decreasing function of $\delta$.} However, our experimental results in Fig. \ref{fig:Training} show that {\blue the training error $\mathrm{err}_{\mathrm{training}}$ is not always monotonic in the training AoI $\delta$. This implies that the training data may not satisfy the Markov property. In fact, practical time-series data is usually non-Markovian \cite{Kampen1998Non-Markov, hanggi1977time, guo2019credibility, wang2021framework}, which hinders the use of data processing inequality.} Hence, novel analytical tools for interpreting information freshness in non-Markovian models are in great need. {\blue To resolve this challenge, we propose a new $\epsilon$-\emph{Markov chain model} that generalizes the standard Markov chain, and develop an $\epsilon$-\emph{data processing inequality} to characterize the relationship between training/inference errors and AoI.}}

\ignore{\subsubsection{$\epsilon$-Markov Chain Model and $\epsilon$-Data Processing Inequality}\label{Def_eMarkov}
We develop a unified framework {\blue to analyze information freshness for both} Markovian and non-Markovian time-series data. Towards that end, {\blue we introduce the following relaxation of the standard Markov chain model: }}

\begin{definition}[\textbf{$\epsilon$-Markov Chain}]\label{epsilonMarkovChain}
Given $\epsilon \geq 0$, a sequence of three random variables $Z, X,$ and $Y$ is said to be an \emph{$\epsilon$-Markov chain}, denoted as $Z \overset{\epsilon} \rightarrow X \overset{\epsilon} \rightarrow Y$, if
\begin{align}\label{epsilon-Markov-def}
\!\!I_{\mathrm{log}}(Y;Z|X)=  D_{\mathrm{log}}\!\left(P_{Y,X,Z} || P_{Y|X} P_{Z|X} P_{X}\right)  \leq  \epsilon^2,
\end{align}
where\ifreport \footnote{In \eqref{epsilon-Markov-def}, if $P_{Y|X=x}(y) = 0$, then $P_{Y|X=x,Z=z}(y) = 0$ which leads to a term $0\mathrm{log}\frac{0}{0}$ in the KL-divergence $D_{\mathrm{log}} (P_{Y|X=x,Z=z} || P_{Y|X=x})$. We adopt the convention in information theory \cite{polyanskiy2014lecture} to define $0\mathrm{log}\frac{0}{0}=0$.}\else \fi
\begin{align}\label{chi-divergence-def}
D_{\mathrm{log}}(P_Y ||Q_Y)=\sum_{y \in \mathcal{Y}} P_Y(y) \mathrm{log} \frac{P_Y(y)}{Q_Y(y)}
\end{align}
is KL-divergence and $I_{\mathrm{log}}(Y;Z|X)$ is Shannon conditional mutual information.
\end{definition}


A Markov chain is an $\epsilon$-Markov chain with $\epsilon= 0$. If $Z \rightarrow X \rightarrow Y$ is a Markov chain, then $Y \rightarrow X \rightarrow Z$ is also a Markov chain \cite[p. 34]{cover1999elements}. A similar property holds for the $\epsilon$-Markov chain.
\begin{lemma}\label{Symmetric}
 If $Z \overset{\epsilon}  \rightarrow X \overset{\epsilon} \rightarrow Y$, then $Y \overset{\epsilon}  \rightarrow X \overset{\epsilon} \rightarrow Z$.
\end{lemma}
\ifreport
\begin{proof}
See Appendix \ref{PSymmetric}.
\end{proof}
\else
Due to space limitation, all the proofs are relegated to our technical report \cite{technical_report}.
\fi

\begin{figure}
  \centering
\begin{subfigure}[b]{0.20\textwidth}
\includegraphics[width=1\linewidth]{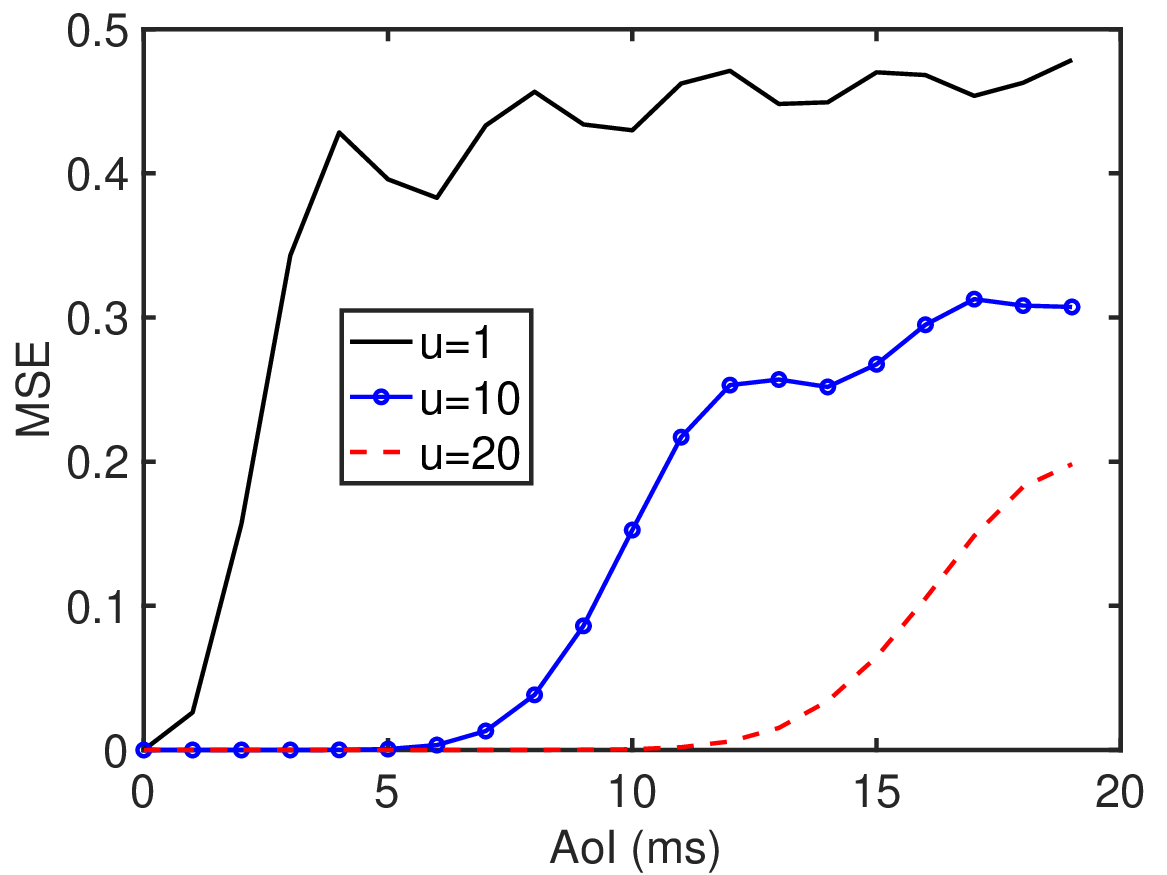}
  \subcaption{Training Error vs. AoI}
\end{subfigure}
\begin{subfigure}[b]{0.20\textwidth}
\includegraphics[width=1\linewidth]{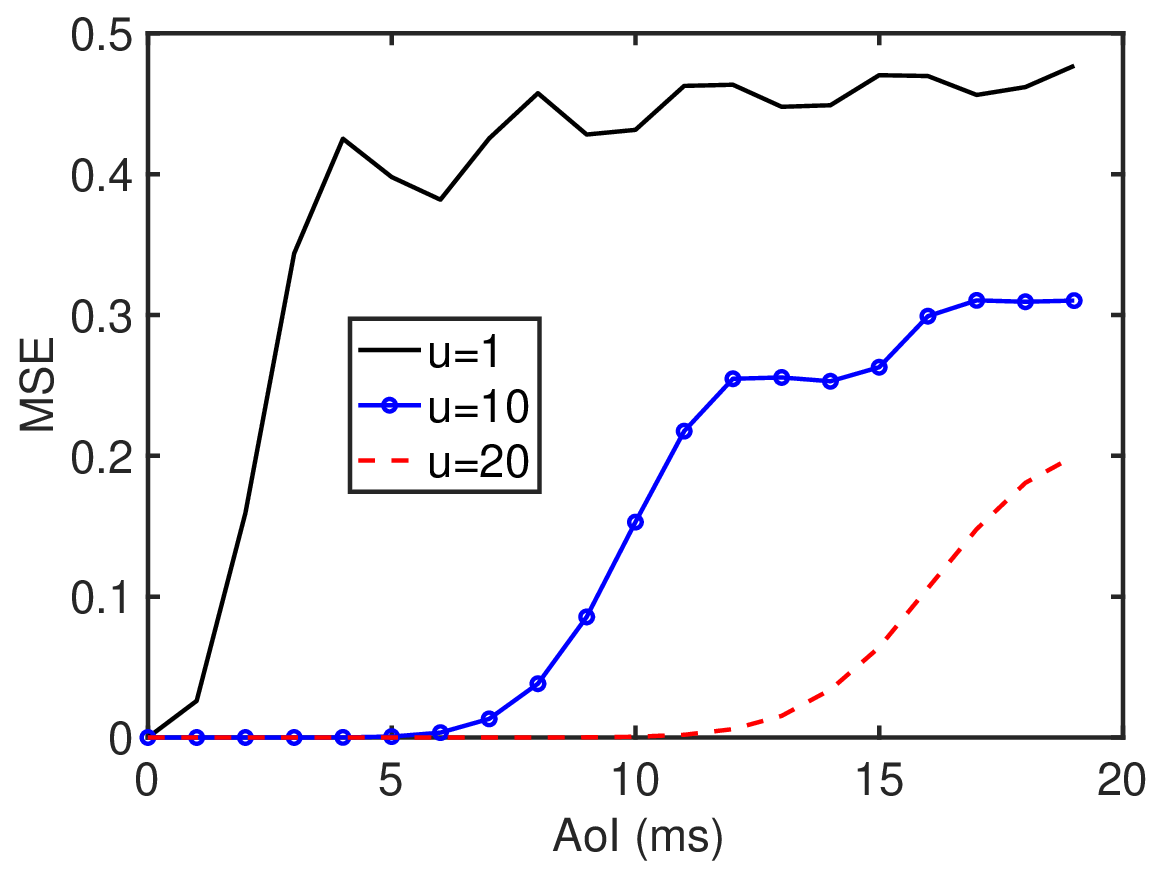}
  \subcaption{Inference Error vs. AoI}
\end{subfigure}
%
\caption{Performance of channel state information prediction. 
The training error and inference error are non-monotonic in AoI. As the feature sequence length $u$ increases, the errors tend closer to non-decreasing functions of the AoI.
}
\label{fig:Trainingcsi}
\end{figure}

By Lemma \ref{Symmetric}, the $\epsilon$-Markov chain can be denoted as $Y \overset{\epsilon} \leftrightarrow X \overset{\epsilon} \leftrightarrow Z$.
In the following lemma, we provide a relaxation of the data processing inequality, which is called an \emph{$\epsilon$-data processing inequality}.
\begin{lemma}[\textbf{$\epsilon$-data processing inequality}] \label{Lemma_CMI}
The following assertions are true:
\begin{itemize}
\item[(a)] If $Y \overset{\epsilon}\leftrightarrow X \overset{\epsilon}\leftrightarrow Z$ is an $\epsilon$-Markov chain, then 
\begin{align}
H_L(Y|X) \leq H_L(Y|Z)+O(\epsilon).
\end{align}
\item[(b)] If, in addition, $H_L(Y)$ is twice differentiable in $\mathcal{P^Y}$, then
\begin{align}
H_L(Y|X) \leq H_L(Y|Z)+O(\epsilon^2).
\end{align}
\end{itemize}
\end{lemma}
\ifreport
\begin{proof}
Lemma \ref{Lemma_CMI} is proven by using a local information geometric analysis. See Appendix \ref{PLemma_CMI} for the details.
\end{proof}
\else
\fi
Lemma \ref{Lemma_CMI}(b) was mentioned in \cite{shisher2021age} without proof. Lemma \ref{Lemma_CMI}(a) is a new result. Now, we are ready to characterize how $H_L(\tilde Y_0 | \tilde X_{-\delta})$ varies with the AoI $\delta$.

\begin{theorem}\label{theorem1}
The $L$-conditional entropy
\begin{align}\label{eMarkov}
H_L(\tilde Y_0|\tilde X_{-\delta})= g_1(\delta)-g_2(\delta)
\end{align}
is a function of $\delta$, where $g_1(\delta)$ and $g_2(\delta)$ are two non-decreasing functions of $\delta$, given by
\begin{align}\label{g12function}
\!\!g_1(\delta)=&H_L(\tilde Y_0 | \tilde X_0) + \sum_{k=0}^{\delta-1}~I_L(\tilde Y_0; \tilde X_{-k}  | \tilde X_{-k-1}),~\nonumber\\
g_2(\delta)=&\sum_{k=0}^{\delta-1} I_L(\tilde Y_0; \tilde X_{-k-1} | \tilde X_{-k}),\!\!\!
\end{align}
where the $L$-conditional mutual information $I_L(Y; X|Z)$ between two random variables $Y$ and $X$ given $Z$ is  
\begin{align}\label{CMI}
I_L(Y; X|Z)=H_L(Y | Z)-H_L(Y|X, Z).
\end{align}
If $\tilde Y_0 \overset{\epsilon}\leftrightarrow \tilde X_{-\mu} \overset{\epsilon}\leftrightarrow \tilde X_{-\mu-\nu}$ is an $\epsilon$-Markov chain for every $\mu, \nu \geq 0$, then $g_2(\delta) = O(\epsilon)$ and 
\begin{align}\label{eMarkov1}
H_L(\tilde Y_{0}|\tilde X_{-\delta})= g_1(\delta)+O(\epsilon).
\end{align}
\ignore{{\blue If, in addition,} $H_L(Y_0)$ is twice differentiable in $P_{Y_0}$, then $g_2(\delta) = O(\epsilon^2)$ and
\begin{align}\label{eMarkov2}
H_L(Y_0|X_{-\delta})= g_1(\delta)+O(\epsilon^2).
\end{align}}
\end{theorem}
\ifreport
\begin{proof}
See Appendix \ref{Ptheorem1}.
\end{proof}
\else
\fi

According to Theorem \ref{theorem1}, the monotonicity of $H_L(\tilde Y_0|\tilde X_{-\delta})$ in $\delta$ is characterized by the parameter $\epsilon \geq 0$ in the $\epsilon$-Markov chain model. If $\epsilon$ is small, then $\tilde Y_0 \overset{\epsilon}\leftrightarrow \tilde X_{-\mu} \overset{\epsilon}\leftrightarrow \tilde X_{-\mu-\nu}$ is close to a Markov chain, and $H_L(\tilde Y_0|\tilde X_{-\delta})$ is nearly non-decreasing in $\delta$. If $\epsilon$ is large, then $\tilde Y_0 \overset{\epsilon}\leftrightarrow \tilde X_{-\mu} \overset{\epsilon}\leftrightarrow \tilde X_{-\mu-\nu}$ is far from a Markov chain, and $H_L(\tilde Y_0|\tilde X_{-\delta})$ could be non-monotonic in $\delta$. 
Theorem  \ref{theorem1} can be readily extended to the training error with random AoI $\Theta$ by using stochastic orders \cite{stochasticOrder}.

\ignore{Theorem \ref{theorem1} {\blue tells us that} $H_L(Y_{t}|X_{t-\delta})$ is a function of the training AoI $\delta$, which is not necessarily monotonic.~The monotonicity of  $H_L(Y_{t}|X_{t-\delta})$ {\blue versus $\delta$} is characterized by the parameter $\epsilon \geq 0$ in the $\epsilon$-Markov chain model.~If $\epsilon$ is small, then the time-series data is close to a Markov chain and  $H_L(Y_{t}|X_{t-\delta})$ is nearly non-decreasing in the training AoI $\delta$. {\blue If $\epsilon$ is large, then the time-series data is far from a Markov chain and $H_L(Y_{t}|X_{t-\delta})$ is non-monotonic in $\delta$. ~Theorem \ref{theorem1} provides an interpretation of the} experimental results in Figure \ref{fig:Training}:~Recall that $u$ is the length of the observation sequence $X_{t-\delta}=(s_{t-\delta}, s_{t-\delta-1},\ldots, s_{t-\delta-u+1})$.~According to Shannon's interpretation of Markov sources in his seminal work \cite{Shannon1948}, the larger $u$, {\blue the closer $(Y_t, X_{t-\mu}, X_{t-\mu-\nu})$ tends to a Markov chain. The training error may be non-monotonic in the AoI $\delta$ for small $u$, but it will progressively become a growing function of the AoI as $u$ increases, which agrees with the results in Figure \ref{fig:Training}.}

\subsubsection{Training Error under Random Training AoI}
{\blue Theorem \ref{theorem1} can be extended to cover random training AoI $\Theta$ by using the stochastic ordering techniques \cite{stochasticOrder}. }}
\begin{definition}[\textbf{Univariate Stochastic Ordering}]\cite{stochasticOrder} 
A random variable $X$ is said to be stochastically smaller than another random variable $Z$, denoted as $X \leq_{st} Z$, if
\begin{align}
    P(X>x) \leq P(Z>x), \ \ \forall x \in \mathbb R.
\end{align}
\end{definition}
\begin{theorem}\label{theorem2}
If $\tilde Y_0 \overset{\epsilon}\leftrightarrow \tilde X_{-\mu} \overset{\epsilon}\leftrightarrow \tilde X_{-\mu-\nu}$ is an $\epsilon$-Markov chain for all $\mu, \nu \geq 0$, and the training AoIs in two experiments $1$ and $2$ satisfy $\Theta_{1} \leq_{st} \Theta_{2}$, then  
\begin{align}\label{dynamicsoln}
    H_L(\tilde Y_0|\tilde X_{-\Theta_1}, \Theta_1) \leq H_L(\tilde Y_0|\tilde X_{-\Theta_2}, \Theta_2)+O(\epsilon).
\end{align}
\ignore{\violet provided that the $L$-conditional entropies in \eqref{dynamicsoln} are finite. 
\item[(b)] If, in addition, $H_L(Y_t)$ is twice differentiable in $P_{Y_t}$, then
\begin{align}\label{dynamic_eqn}
    H_L(Y_{t}|X_{t-\Theta_1}, \Theta_1) \leq H_L(Y_{t}|X_{t-\Theta_2}, \Theta_2)+O(\epsilon^2).
\end{align}}
\end{theorem}
\ifreport
\begin{proof}
See Appendix \ref{Ptheorem2}.
\end{proof}
\else
\fi
According to Theorem \ref{theorem2}, if $\Theta_{1}$ is stochastically smaller than $\Theta_{2}$, then the training error in Experiment 1 is approximately
smaller than that in Experiment 2. If, in addition to the conditions in Theorems \ref{theorem1}-\ref{theorem2}, $H_L(\tilde Y_0)$ is twice differentiable in $P_{\tilde Y_0}$, then the last term $O(\epsilon)$ in \eqref{eMarkov1} and \eqref{dynamicsoln} becomes $O(\epsilon^2)$. 

\subsection{Inference Error vs. Inference AoI}\label{SecInferenceError}
\ignore{{\blue Next, we analyze the relationship between inference error and inference AoI.}
\subsubsection{Inference Error under Deterministic Inference AoI}} 
Using \eqref{given_L_condentropy}, \eqref{eq_cond_entropy1}, and \eqref{cond-cross-entropy}, it is easy to show that the $L$-conditional cross entropy $H_L(P_{Y_{t}|X_{t-\delta}}; P_{\tilde Y_0|\tilde X_{-\delta}} | P_{X_{t-\delta}})$ is lower bounded by the $L$-conditional entropy $H_L(Y_{t} | X_{t-\delta})$. 
In addition, $H_L(P_{Y_{t}|X_{t-\delta}}; P_{\tilde Y_0|\tilde X_{-\delta}} | P_{X_{t-\delta}})$ is close to its lower bound $H_L(Y_{t} | X_{t-\delta})$, if the conditional distributions $P_{Y_t|X_{t-\delta}}$ and $P_{\tilde Y_0|\tilde X_{-\delta}}$ are close to each other, as stated in Lemma \ref{lemma_inference}. 


\begin{lemma}\label{lemma_inference}
Given $\beta \geq 0$, if 
    \begin{align}\label{T3condition2}
     \sum_{x \in \mathcal X} P_{X_{t-\delta}}(x)\sum_{y \in \mathcal Y} (P_{Y_{t}|X_{t-\delta}=x}(y)\!-\!P_{\tilde Y_0| \tilde X_{-\delta}=x}(y))^2 \!\!~\leq\! \beta^2,
    \end{align}
    then we have 
    \begin{align}\label{Eq_Theorem3a}
       \!\!\! H_L(P_{Y_{t}|X_{t-\delta}}; P_{\tilde Y_0|\tilde X_{-\delta}} | P_{X_{t-\delta}})\!=\!H_L(Y_{t} | X_{t-\delta})\!+O(\beta).
    \end{align}
\end{lemma}
\ifreport
\begin{proof}
See Appendix \ref{Plemma_inference}.
\end{proof}
\ignore{If \eqref{T3condition2} is replaced by the condition
 \begin{align}\label{T3condition2_diff}
    \sum_{x \in \mathcal X} P_{X_{t-\delta}}(x)(\sum_{y \in \mathcal Y} \big|P_{Y_{t}|X_{t-\delta}}(y|x)\!-\!P_{\tilde Y_0| \tilde X_{-\delta}=x}(y|x)\big|)^2 \!\!~\leq\! \beta^2,
    \end{align}
then Lemma \ref{lemma_inference} still holds.}
\else
\fi 
Combining Theorem \ref{theorem1} and Lemma \ref{lemma_inference}, the monotonicity of $H_L(P_{Y_{t}|X_{t-\delta}}; P_{\tilde Y_0|\tilde X_{-\delta}} | P_{X_{t-\delta}})$ versus $\delta$ is characterized in the next theorem.
\begin{theorem}\label{theorem3}
\ignore{\item[(a)] If $\{(Y_t, X_t),t \in \mathbb Z\}$ is a stationary process, then $H_L(Y_{t}; \tilde Y_0 | X_{t-\delta})$ is a function of the inference AoI $\delta$.} 
If $Y_t \overset{\epsilon}\leftrightarrow X_{t-\mu} \overset{\epsilon}\leftrightarrow X_{t-\mu-\nu}$ is an $\epsilon$-Markov chain for all $\mu, \nu \geq 0$ and \eqref{T3condition2} holds for all $\delta \in \mathcal Z^{+}$, then for all $0 \leq \delta_1\leq \delta_2$
\begin{align}\label{eq_theorem3}
 &H_L(P_{Y_{t}|X_{t-\delta_1}}; P_{\tilde Y_0|\tilde X_{-\delta_1}} | P_{X_{t-\delta_1}}) \nonumber\\ \leq& H_L(P_{Y_{t}|X_{t-\delta_2}}; P_{\tilde Y_0|\tilde X_{-\delta_2}} | P_{X_{t-\delta_2}})\!+\!O\big(\max\{\epsilon, \beta\}\big).
\end{align}
\end{theorem}
 \ifreport
\begin{proof}
See Appendix \ref{Ptheorem3}.
\end{proof}
\else
\fi
According to Theorem \ref{theorem3}, if $\epsilon$ and $\beta$ are close to zero, $H_L(P_{Y_{t}|X_{t-\delta}}; P_{\tilde Y_0|\tilde X_{-\delta}} | P_{X_{t-\delta}})$ is nearly a non-decreasing function of $\delta$; otherwise, $H_L(P_{Y_{t}|X_{t-\delta}}; P_{\tilde Y_0|\tilde X_{-\delta}} | P_{X_{t-\delta}})$ as a function of $\delta$ can be non-monotonic. 


\subsection{Interpretation of the Experimental Results}\label{ExperimentalObs}

We use Theorems \ref{theorem1}-\ref{theorem3} to interpret the experimental results in Figs. \ref{fig:learning}-\ref{fig:Trainingcsi}.
In Fig. \ref{fig:learning}, the training and inference errors for video prediction are increasing functions of the AoI. This observation suggests that the target and feature time-series data $(Y_t, X_{t-\mu}, X_{t-\mu-\nu})$ for video prediction is close to a Markov chain.   
{\red As discussed in Appendix \ref{Experiments}, in the robot state prediction experiment, communication delays and periodic updates from leader robot introduce $20$ to $40$ time slots of AoI in the follower robot's perception of the leader robot's state. Consequently, the follower robot receives delayed state information from the leader robot to determine its movement. This implies that slightly outdated state of the leader robot $X_{t-d}$ are more relevant that the fresh observation $X_t$ for predicting the current state of the follower robot $Y_t$, as illustrated in Fig. \ref{fig:DelayedNetworkedControlled} and the target and feature data sequence $(Y_t, X_{t-\mu}, X_{t-\mu-\nu})$ may deviate significantly from a Markov chain.} In the experiment of actuator state prediction under mechanical response delay, pole angle at time $t$ is strongly correlated with the cart velocity generated $25$ ms ago, as observed from data traces in Fig. \ref{fig:TrainingCartVelocity}(b). Moreover, temperature and CSI signals have temporal dependence. For example, temperature at time $t$ depends on the temperature of $24$ hours ago. 
These observations imply that the target and feature data sequence $(Y_t, X_{t-\mu}, X_{t-\mu-\nu})$ for all $\mu, \nu \geq 0$ may not be close to a Markov chain in the experimental results depicted in Figs. \ref{fig:DelayedNetworkedControlled}-\ref{fig:Trainingcsi}. Because the target and feature time-series data involved is non-Markovian, Theorems \ref{theorem1}-\ref{theorem3} suggest that the training error and inference error could be non-monotonic with respect to AoI, as observed in Figs. \ref{fig:DelayedNetworkedControlled}-\ref{fig:Trainingcsi}.

Recall that $u$ is the sequence length of the feature $X_t = (V_t,V_{t-1},\ldots, V_{t-u+1})$. In Figs. \ref{fig:TrainingCartVelocity}-\ref{fig:Trainingcsi}, the training and inference errors tends to become non-decreasing functions of the AoI $\delta$ as the feature length $u$ grows. This phenomenon can be interpreted by Theorems \ref{theorem1}-\ref{theorem3}:
According to Shannon’s Markovian representation of discrete-time sources in his seminal work \cite{Shannon1948}, the larger $u$, the closer $(Y_t, X_{t-\mu}, X_{t-\mu-\nu})$ tends to a Markov chain. According to Theorems \ref{theorem1}-\ref{theorem3}, as $u$ increases, the training and inference errors tend to be non-decreasing with respect to the AoI $\delta$, which agrees with Figs. \ref{fig:TrainingCartVelocity}-\ref{fig:Trainingcsi}. One disadvantage of large feature length $u$ is that it increases the channel resources needed for transmitting the features. The optimal choice of the feature length $u$ is studied in \cite{shisher2023learning}. 
 
\section{A Model-based Interpretation of Information Freshness in Remote Inference}

{We construct two models to analyze the non-monotonocity of the $L$-conditional entropy $H_L(Y_t|X_{t-\delta})$ with respect to the AoI $\delta$ and interpret the reasons.

\subsection{Reaction Prediction with Delay}
{\red We present the following analytical example for understanding the experimental results illustrated in Figs. \ref{fig:DelayedNetworkedControlled}-\ref{fig:TrainingCartVelocity}.}

\begin{example}[Reaction Prediction]\label{Toy_Example2}
Consider a causal system represented by $Y_t = f(X_{t-d})\in \mathcal Y$, where $X_t \in \mathcal X$ and $Y_t \in \mathcal Y$ are the input and output of the system, respectively, $d\geq 0$ is the delay introduced by the system, and $f(\cdot)$ is a function. 


\begin{figure}[t]
\centering
\includegraphics[width=0.30\textwidth]{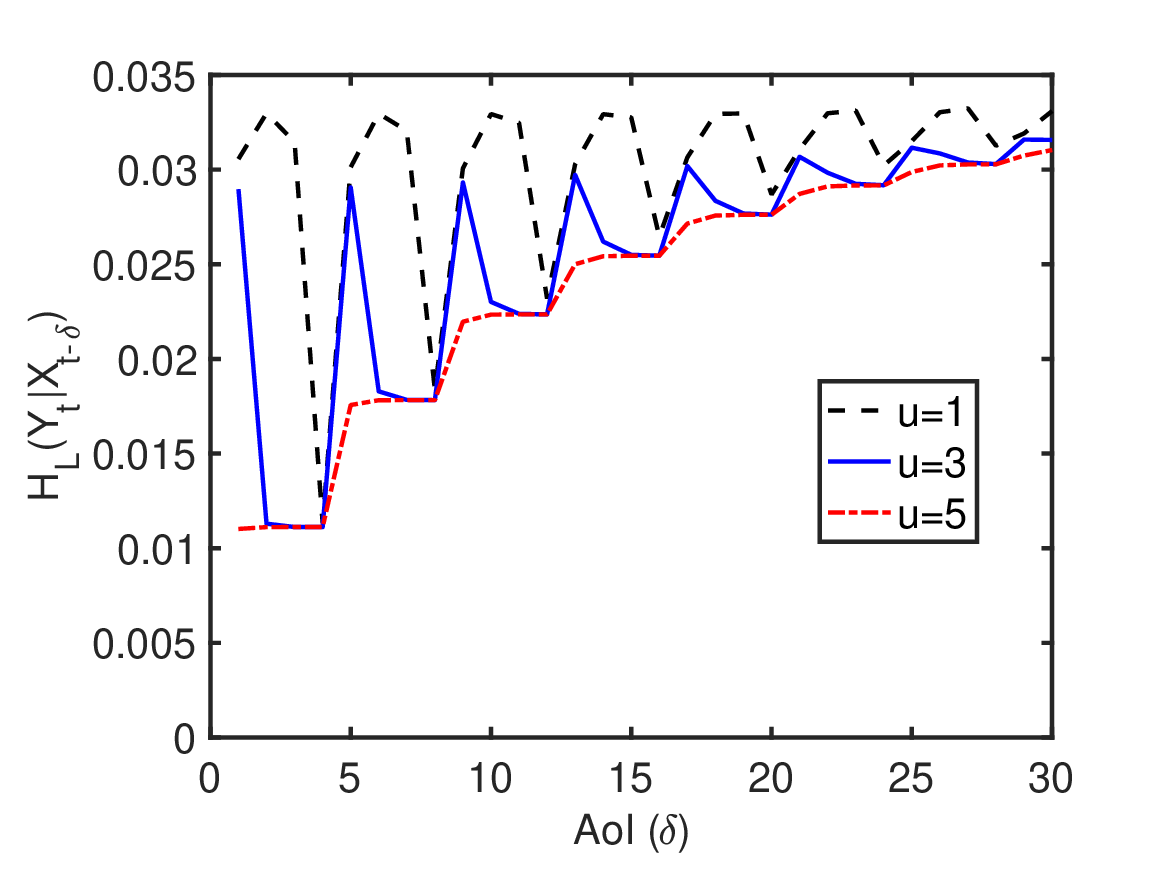}
\caption{\small $L$-conditional entropy vs. AoI for Example \ref{ARp}. \label{fig:ARp}
}
\end{figure}
 
\begin{lemma}\label{ToyExampleLemma1}
Let $X_t \in \mathcal X$ and $Y_t = f(X_{t-d}) \in \mathcal Y$, where $f(\cdot)$ is a function. If $X_t \leftrightarrow X_{t-\mu} \leftrightarrow X_{t-\mu-\nu}$ is a Markov chain for all $\mu, \nu \geq 0$, then $H_L(Y_t | X_{t-\delta})$ decreases with $\delta$ when $0 \leq \delta < d$ and increases with $\delta$ when $\delta \geq d$. In addition, for any random variable $Z \in \mathcal X$, 
\begin{align}
    H_L(Y_t|X_{t-d}) =  H_L(Y_t|Y_{t}) \leq \min_{Z\in \mathcal X} H_L(Y_t|Z).
\end{align}
\end{lemma}
\begin{proof}
    See Appendix \ref{ToyExample}.
\end{proof}

\end{example}
Lemma \ref{ToyExampleLemma1} implies that the feature $X_{t-d} \in \mathcal X$ achieves the minimum expected loss in predicting $Y_t$. Therefore, for predicting $Y_t$, $X_{t-d}$ is the optimal choice in $\mathcal X$, not the freshest feature $X_t \in \mathcal X$. Hence, fresh data is not always the best.

The robotic state prediction and the actuator state prediction experiments in Figs. \ref{fig:DelayedNetworkedControlled}-\ref{fig:TrainingCartVelocity} are also instances of reaction prediction. Similar to Example \ref{Toy_Example2}, the freshest feature with AoI$=$0 is not the best choice for predicting the reaction in Figs. \ref{fig:DelayedNetworkedControlled}-\ref{fig:TrainingCartVelocity}. However, the relationship between the leader and follower robots’ states in the robotic state prediction experiment and the relationship between cart velocity and pole angle in the actuator state prediction experiment are much more complicated than the input-output relationship in Example \ref{Toy_Example2}. 


\subsection{Autoregressive Model}

We evaluate the $L$-conditional entropy using an autoregressive linear system model.
\begin{example}[\textbf{Autoregressive Model}]\label{ARp}
Consider a discrete time autoregressive linear system of the order $4$, i.e., AR($4$):
\begin{align}
V_{t}=0.1 V_{t-1}+0.4 V_{t-4}+W_t,
\end{align}
where $W_t \in \mathbb R$ is a Gaussian noise with zero mean and variance $0.01$. Let $Y_t=V_t+N_t$ be the target variable, where $N_t$ is a Gaussian noise with zero mean and variance $0.01$. Both $W_t$ and $N_t$ are  {i.i.d.} over time. The goal is to predict $Y_t$ using a feature sequence $X_{t-\delta}=(V_{t-\delta}, V_{t-\delta-1}, \ldots, V_{t-\delta-u+1})$ with length $u$. The prediction error is characterized by a quadratic loss function $L(y, \hat y)=(y-\hat y)^2$. Because (i) $X_{t-\delta}$ and $Y_t$ are jointly Gaussian and (ii) the loss function is quadratic, the $L$-conditional entropy $H_L(Y_t|X_{t-\delta}) = \mathbb E [(Y_t - \mathbb E[Y_t|X_{t-\delta}])^2 ]$ is the linear minimum mean-square error in signal-agnostic remote estimation.
\end{example}

We can observe from Fig. \ref{fig:ARp} that
the $L$-conditional entropy is non-monotonic in the AoI $\delta$ when the feature length is $u = 1$ or  3. When $u$ is increased to $5$, $H_L(Y_t|X_{t-\delta})$ becomes a non-decreasing function of the AoI $\delta$. This is because $Y_t \overset{}\leftrightarrow X_{t-\mu} \overset{}\leftrightarrow X_{t-\mu-\nu}$ is a Markov chain when $u=5$. The model-based numerical evaluation results in Fig. \ref{fig:ARp} is similar to the experimental results in Figs. \ref{fig:TrainingCartVelocity}-\ref{fig:Trainingcsi}. In general, $H_L(Y_t|X_{t-\delta})$ for any AR($p$) model with $p\geq2$ can be non-monotonic in AoI $\delta$, particularly when the feature length $u$ is small. 
{\blue In a follow-up study \cite{shisher2024monotonicity}, we  provide analytical expressions of the $L$-conditional entropy and the parameter $\epsilon$ for the AR($p$) model. 
}}

%% file: Scheduling_SingleSource.tex

\section{Scheduling for Inference Error Minimization: The Single-source, Single-channel Case}\label{Scheduling}

\begin{figure}[t]
\centering
\includegraphics[width=0.35\textwidth]{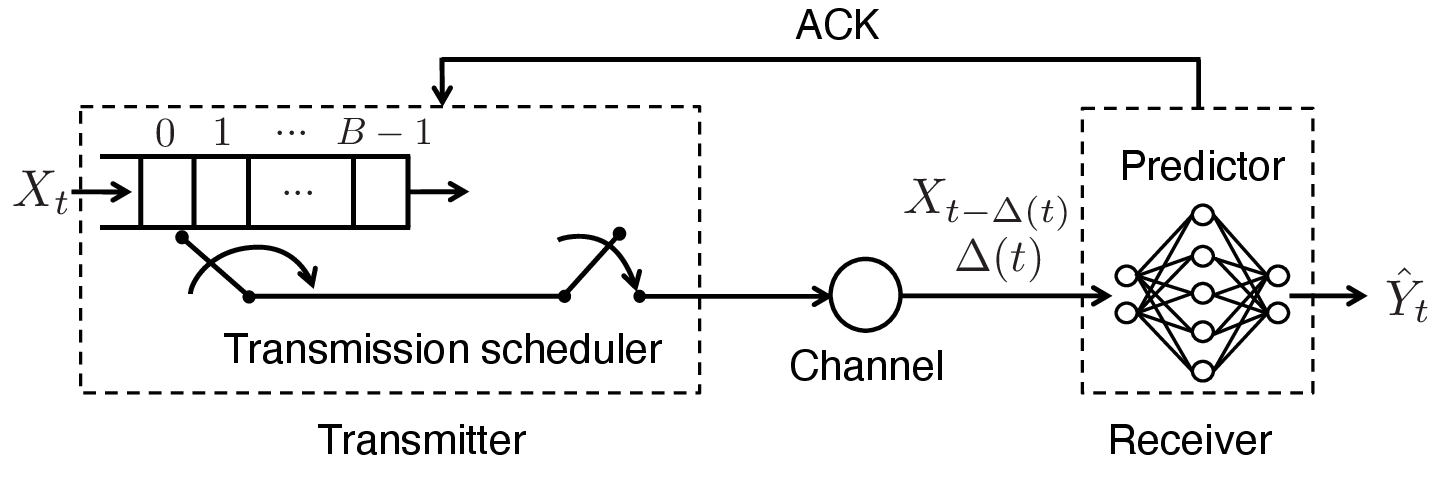}
\caption{\small  A remote inference system with ``selection-from-buffer.'' At each time slot $t$, the transmitter generates a feature $X_t$ and keeps it in a buffer that stores $B$ most recent features $(X_t, X_{t-1}, \ldots, X_{t-B+1}$). The scheduler decides (i) when to submit features to the channel and (ii) which feature in the buffer to submit. \label{fig:scheduling}
}
\vspace{-3mm}
\end{figure}

In this section, we will introduce a new ``selection-from-buffer" medium access model and develop a novel transmission scheduling policy to minimize the time-average inference error in single-source remote inference systems. This scheduling policy can effectively minimize general functions of the AoI, regardless of whether the function is monotonic or not.

\ignore{In this section, we will design a new scheduling algorithm for minimizing the long-term average inference error in supervised learning based real-time forecasting. Following the discussions in Section \ref{InformationAnalysis}, we can show that the inference error is a function of the inference AoI, which is not necessarily monotonic. Hence, existing scheduling algorithms for optimizing monotonic AoI metrics \cite{bedewy2020optimizing,SunTIT2020,bedewy2021optimal,SunNonlinear2019, OrneeTON2021,klugel2019aoi,  pan2020minimizing, jiang2018can, Maatouk2021, abd2020aoi,sun2017update, yates2021age, li2021age} may not be appropriate for inference error minimization. We will provide a semi-analytical optimal solution for minimizing the time-average of non-monotonic AoI functions, where the solution is expressed by using the index of an AoI bandit process. To the extent of our knowledge, the connection between AoI-based scheduling and index that we discover has not been reported before.}

\subsection{Selection-from-Buffer: A New Status Updating Model}
{Consider the remote inference system depicted in Fig. \ref{fig:scheduling}, where a source progressively generates features and sends them through a channel to a receiver. The system operates in discrete time-slots. In each time slot $t$, a pre-trained neural network at the receiver employs the freshest received feature $X_{t-\Delta(t)}$ to infer the current label $Y_t$. 

As discussed in Sections \ref{RemoteInference}-\ref{InformationAnalysis}, the inference error $\mathrm{err}_{\mathrm{inference}}(\Delta(t))$ is a function of the AoI $\Delta(t)$, whereas the function is not necessarily monotonically increasing. In certain scenarios, a stale feature with $\Delta(t)>0$ can outperform a freshly generated feature with $\Delta(t)=0$. Inspired by these observations, we introduce a novel medium access model for status updating, which is termed the ``selection-from-buffer" model. In this model, the source maintains a buffer that stores the $B$ most recent features $(X_t, X_{t-1}, \ldots, X_{t-B+1})$ in each time slot $t$. Specifically, at the beginning of time slot $t$, the source appends a newly generated feature $X_t$ to the buffer, while concurrently evicting the oldest feature $X_{t-B}$. 
If the channel is available at time slot $t$, the transmitter can send one of the $B$ most recent features or remain silent, where the transmission may last for one or multiple time-slots. Notably, the ``selection-from-buffer'' model generalizes the ``generate-at-will'' model \cite{yates2015lazy,YinUpdateInfocom}, with the latter is a special case of the former with $B=1$.\footnote{In comparison to the ``generate-at-will'' model, our ``selection-from-buffer'' model presents a critical advantage: it enables the systematic investigation of \emph{optimal feature design for remote inference.} As an evidence, our subsequent study \cite{shisher2023learning} demonstrates that substantial performance improvements can be attained through the joint optimization of transmission scheduling and the feature sequence length $u$.}

%


The system starts to operate at time slot $t = 0$. We assume that the buffer is initially populated with $B$ features $(X_0, X_{-1}, \ldots, X_{-B+1})$ at time slot $t=0$. By this, the buffer is kept full at all time slot $t\geq 0$. The channel is modeled as a non-preemptive server with feature transmission times $T_i\geq 1$, which can be random due to factors like time-varying channel conditions, collisions, random packet sizes, etc. We assume that the $T_i$'s are
\emph{i.i.d.} with 
$1 \leq \mathbb E[T_i ] < \infty$. 
The $i$-th feature is generated in time slot $G_i$, submitted to the channel in time slot $S_i$, and delivered to the receiver in time slot $D_i = S_i +T_i$, where $G_i \leq S_i<D_i$, and $D_i \leq S_{i+1}< D_{i+1}$. Once a feature is delivered, an acknowledgment (ACK) is fed back to the transmitter in the same time slot. Thus, the idle/busy state of the channel is known at the transmitter.}

\subsection{Scheduling Policies and Problem Formulation}

A transmission scheduler determines (i) when to submit features to the channel and (ii) which feature in the buffer to submit. In time slot $S_i$, let $X_{G_i}=X_{S_i-b_i}$ be the feature submitted to the channel, which is the $(b_i + 1)$-th freshest feature in the buffer, with $b_i \in \{0, 1, \ldots, B-1\}$. By this, $G_i=S_i-b_i$. A scheduling policy is denoted by a 2-tuple $(f, g)$, where $g = (S_1, S_2, \ldots)$ determines when to submit the features and $f= (b_1, b_2,$ $ \ldots)$ specifies which feature in the buffer to submit.


Let $U(t)= \max_i\{G_i : D_i \leq t \}$ represent the generation time of the freshest feature delivered to the receiver up to time slot $t$. Because $G_i=S_i-b_i$, $U(t)= \max_i\{S_i-b_i : D_i \leq t \}$. The age of information (AoI) at time $t$ is  \cite{kaul2012real}
\begin{align}\label{age}
\Delta(t) = t-U(t)=t-\max_i \{S_i-b_i: D_i \leq t\}.
\end{align}
Because $D_i < D_{i+1}$, $\Delta(t)$ can be re-written as
\begin{align}
\Delta(t) = t-S_i+b_i,~~ \mathrm{if}~~D_i \leq t < D_{i+1}.
\end{align}
The initial state of the system is assumed to be $S_0 = 0, D_0 = T_0$, and $\Delta(0)$ is a finite constant.


{\blue We focus on the class of \emph{signal-agnostic} scheduling policies in which each decision is determined without using the knowledge of the signal value of the observed process. A scheduling policy $(f, g)$ is said to be signal-agnostic, if the policy is independent of $\{(Y_t, X_t), t =0,1,2,\ldots\}$.}
Let $\Pi$ represent the set of all causal and signal-agnostic scheduling policies that satisfy three conditions: (i) the transmission time schedule $S_i$ and the buffer position $b_i$ are determined based on the current and the historical information available at the scheduler; (ii) the scheduler does not have access to the realization of the process $\{(Y_t, X_t), t =0,1,2,\ldots\}$; and (iii) the scheduler can access the inference error function $\mathrm{err}_{\mathrm{inference}}(\cdot)$ and the distribution of $T_i$.

Our goal is to find an optimal scheduling policy that minimizes the time-average expected inference error among all causal scheduling policies in $\Pi$:
\begin{align}\label{scheduling_problem}
\bar p_{opt}=\inf_{(f, g) \in \Pi}  \limsup_{T\rightarrow \infty}\frac{1}{T} \mathbb{E}_{(f, g)} \left[ \sum_{t=0}^{T-1} p(\Delta(t))\right].
\end{align}
where we use a simpler notation $p(\Delta(t))=\mathrm{err}_{\mathrm{inference}}(\Delta(t))$ to represent the inference error in time-slot $t$, and $\bar p_{opt}$ is the optimum value of \eqref{scheduling_problem}. Because $p(\cdot)$ is not necessarily monotonic and the scheduler needs to determine which feature in the buffer to send, \eqref{scheduling_problem} is more challenging than the scheduling problems for minimizing non-decreasing age functions in 
\cite{yates2015lazy,sun2017update, YinUpdateInfocom, SunNonlinear2019,SunSPAWC2018, orneeTON2021, Tripathi2019, klugel2019aoi, bedewy2021optimal, kadota2018optimizing, hsu2018age, sun2019closed, Kadota2018}. {Note that $p(\Delta(t))=\mathrm{err}_{\mathrm{inference}}(\Delta(t))$ is the inference error in time-slot $t$, instead of its information-theoretic approximations.} 


\ifreport

\ignore{{\blue In order to evaluate the average inference error during time slots $0, 1, \ldots, T-1,$ let the random variable $\Delta$ in \eqref{eq_inferenceerror} follow the empirical distribution of the inference AoI process $\{\Delta(t), t = 0,1,\ldots, T-1\}$ during the first $T$ time slots. Then, the cumulative distribution function of $\Delta$ is given by \cite[Example 2.4.8]{durrett2019probability} 
\begin{align}\label{distribution_function_AoI}
F_{\Delta}(x) =  \frac{1}{T} \sum_{t=0}^{T-1} \mathbf 1(\Delta(t) \leq x), 
\end{align}
where $\mathbf 1(\cdot)$ is the indicator function. Because $\{(\tilde Y_t , \tilde X_t), t\in \mathbb Z \}$ is a stationary process that is independent of $\Delta(t)$, by substituting \eqref{distribution_function_AoI} into \eqref{eq_inferenceerror}, the average inference error during the first $T$ time slots can be expressed as 
\begin{align}\label{eq_inferenceerror1}
\mathrm{err}_{\mathrm{inference}}(T) = \frac{1}{T} \mathbb E_{(f,g)} \left [ \sum_{t=0}^{T-1} p(\Delta(t))\right],
\end{align}
where $E_{(f,g)}[\cdot]$ represents a conditional expectation for given scheduling policy $(f,g)$, $p(\Delta(t))$ is the inference error in time slot $t$, with the AoI function $p(\cdot)$ defined by 
\begin{align}\label{instantaneous_err} 
p(\delta)=\mathbb E_{Y, X \sim P_{\tilde Y_t, \tilde X_{t-\delta}}}\left[L\left(Y,\phi^*_{P_{Y_t, X_{t-\Theta},\Theta}}(X,\delta)\right)\right].
\end{align}
Let $T\rightarrow \infty$, we obtain the long-term average inference error over an infinite time horizon, which is determined by 
\begin{align}\label{eq_inferenceerror2}
\mathrm{err}_{\mathrm{inference}} =  \limsup_{T\rightarrow \infty}\frac{1}{T} \mathbb{E}_{(f, g)} \left[ \sum_{t=0}^{T-1} p(\Delta(t))\right].
\end{align}
}}


\ignore{Our goal is to find an optimal scheduling policy that minimizes the long-term average inference error among all causal scheduling policies in $\mathcal F \times \mathcal G$:
\begin{align}\label{scheduling_problem}
\bar p_{opt}=\inf_{f \in \mathcal F, g \in \mathcal G}  \limsup_{T\rightarrow \infty}\frac{1}{T} \mathbb{E}_{(f, g)} \left[ \sum_{t=0}^{T-1} p(\Delta(t))\right].
\end{align}
where $p(\Delta(t))$ is the inference error at time slot $t$ and $\bar p_{opt}$ is the optimum value of \eqref{scheduling_problem}, i.e., the minimum time-average expected inference error. 

{\violet The new transmission model in Fig. \ref{fig:scheduling} is different from the widely used ``generate-at-will'' model \cite{yates2015lazy}, due to the additional feature buffer. As we will see in the next subsection, if $p(\cdot)$ is a non-decreasing function, then it is optimal to choose $b_i=0$ for all $i$ and the feature buffer is not needed. However, if $p(\cdot)$ is highly non-monotonic or even has a periodic pattern, it is better to send an old feature that was generated some time ago than sending a new feature that has just been created.}}

\subsection{An Optimal Scheduling Solution}
To elucidate the optimal solution to the scheduling problem \eqref{scheduling_problem}, we first fix the buffer position at $b_i = b$ for each feature $i$ submitted to the channel and focus on the optimization of the transmission time schedule $g=(S_1, S_2,\ldots)$. This simplified transmission scheduling problem is expressed as
\begin{align}\label{sub_scheduling_problem}
\bar p_{b, opt}=\inf_{(f_b, g)\in\Pi}  \limsup_{T\rightarrow \infty}\frac{1}{T} \mathbb{E}_{(f_b, g)}\left [ \sum_{t=0}^{T-1} p(\Delta(t))\right],
\end{align}
where $f_b = (b, b, \ldots)$ represents an invariant buffer position assignment policy and $\bar p_{b, opt}$ is the optimal objective value in \eqref{sub_scheduling_problem}. The insights gained from solving this simplified problem \eqref{sub_scheduling_problem} will subsequently guide us in deriving the optimal solution to the original scheduling problem \eqref{scheduling_problem}.



\begin{theorem}\label{theorem5}
If $|p(\delta)| \leq M$ for all $\delta$ and the $T_i$'s are i.i.d. with $1\leq \mathbb E[T_i]<\infty$, then $g=(S_1(\beta_b), S_2(\beta_b), \ldots)$ is an optimal solution to \eqref{sub_scheduling_problem}, where 
\begin{align}\label{OptimalPolicy_Sub}
S_{i+1}(\beta_b) = \min_{t \in \mathbb Z}\big\{ t \geq D_i(\beta_b): \gamma(\Delta(t)) \geq \beta_b\big\},
\end{align}
$D_i(\beta_b)=S_{i}(\beta_b)+T_i$ is the delivery time of the $i$-th feature submitted to the channel, $\Delta(t)=t-S_i(\beta_b)+b$ is the AoI at time $t$,  $\gamma(\delta)$ is an index function, defined by
\begin{align}\label{gittins}
\gamma(\delta)=\inf_{\tau \in \{1, 2, \ldots\}} \frac{1}{\tau} \sum_{k=0}^{\tau-1} \mathbb E \left [p(\delta+k+T_{1}) \right],
\end{align}
and the threshold $\beta_b$ is the unique root of 
\begin{align}\label{bisection}
\mathbb{E}\left[\sum_{t=D_i(\beta_b)}^{D_{i+1}(\beta_b)-1}  p\big(\Delta(t)\big)\right] - \beta_b~ \mathbb{E}\left[D_{i+1}(\beta_b)-D_i(\beta_b)\right]=0.
\end{align}
The optimal objective value to \eqref{sub_scheduling_problem} is given by 
\begin{align}\label{optimum}
\bar p_{b, opt}=\frac{\mathbb{E}\left[\sum_{t=D_i(\beta_b)}^{D_{i+1}(\beta_b)-1}  p\big(\Delta(t)\big)\right]}{\mathbb{E}\left[D_{i+1}(\beta_b)-D_i(\beta_b)\right]}.
\end{align}
Furthermore, $\beta_b$ is equal to the optimal objective value to \eqref{sub_scheduling_problem}, i.e., $\beta_b= \bar p_{b, opt}$.
\end{theorem}
\ifreport
\begin{proof}[Proof sketch]
{The scheduling problem \eqref{sub_scheduling_problem} is an infinite-horizon average-cost semi-Markov decision process (SMDP) \cite[Chapter 5.6]{bertsekasdynamic1}.  Define $\tau=S_{i+1}-D_i$ as the waiting time for sending the $(i+1)$-th feature after the $i$-th feature is delivered. The Bellman optimality equation of the  SMDP \eqref{sub_scheduling_problem} is 
\begin{align}\label{relativeValuederived}
h_b(\delta)=\inf_{\tau \in \{0, 1, 2, \ldots\}}~&\mathbb E \left [ \sum_{k=0}^{\tau+T_{i}-1} (p(\delta+k) -\bar p_{b, opt}) \right]\nonumber\\
&+\mathbb E[h_b(T_{i}+b)], ~\delta = 1, 2,\ldots,
\end{align}
where $h_b(\cdot)$ is the relative value function of the SMDP \eqref{sub_scheduling_problem}.
Theorem \ref{theorem5} is proven by directly solving \eqref{relativeValuederived}. The details are provided in Appendix \ref{MainResult}.}
\end{proof}
\else
\fi

In supervised learning algorithms, features are shifted, rescaled, and clipped during data pre-processing. Because of these pre-processing techniques, the inference error $p(\delta)$ is bounded, as illustrated in Figs. \ifreport \ref{fig:learning}-\ref{fig:Trainingcsi}\else \ref{fig:learning}-\ref{fig:DelayedNetworkedControlled}\fi. Therefore, the assumption $|p(\delta)| \leq M$ for all $\delta$ in Theorem \ref{theorem5} is not restrictive in practice. 

The optimal scheduling policy  in Theorem \ref{theorem5} is a threshold policy described by the index function $\gamma(\delta)$: According to \eqref{OptimalPolicy_Sub}, a feature is transmitted in time-slot $t$ if and only if two conditions are satisfied: (i) The channel is available for scheduling in time-slot $t$ and (ii) the index $\gamma(\Delta(t))$ exceeds a threshold $\beta_b$, which is precisely equal to the optimal objective value $\bar p_{b, opt}$ of \eqref{sub_scheduling_problem}. {The expression of $\gamma(\delta)$ in \eqref{gittins} is obtained by solving the Bellman optimality equation \eqref{relativeValuederived}, as explained in Appendix \ref{MainResult}.} The threshold $\beta_b$ is calculated by solving the unique root of \eqref{bisection}. Three low-complexity algorithms for this purpose were given by \cite[Algorithms 1-3]{orneeTON2021}. 

It is crucial to note that a non-monotonic AoI function $p(\delta)$ often leads to a non-monotonic index function $\gamma(\delta)$. Consequently, the inverse function of $\gamma(\delta)$ may not exist and the inequality $\gamma(\Delta(t))\geq \beta_b$ in the threshold policy \eqref{OptimalPolicy_Sub}
cannot be equivalently rewritten as an inequality of the form $\Delta(t) \geq \alpha$. This distinction represents a significant departure from previous studies for minimizing either the AoI $\Delta(t)$ or its non-decreasing functions, e.g., \cite{yates2015lazy,sun2017update, YinUpdateInfocom, SunNonlinear2019,SunSPAWC2018, orneeTON2021, Tripathi2019, klugel2019aoi, bedewy2021optimal, kadota2018optimizing, hsu2018age, sun2019closed, Kadota2018}. In these earlier works, the solutions were usually expressed as threshold policies in the form $\Delta(t) \geq \alpha$. Our pursuit of a simple threshold policy for minimizing general and potentially non-monotonic AoI functions was 
inspired by the restart-in-state formulation of the Gittins index \cite[Chapter 2.6.4]{gittins2011multi}, \cite{katehakis1987multi}. 




{\blue 

}

Now we present an optimal solution to \eqref{scheduling_problem}.

\begin{theorem}\label{theorem6}
If the conditions of Theorem \ref{theorem5} hold, then an optimal solution $(f^*, g^*)$ to \eqref{scheduling_problem} is determined by
\begin{itemize}
\item[(a)] $f^*=(b^*, b^*, \ldots)$, where 
\begin{align}\label{feature_Selects}
b^*= \arg \min_{b \in \{0, 1, \ldots, B-1\}} \bar p_{b, opt},
\end{align}
and $\bar p_{b, opt}$ is the optimal objective value to \eqref{sub_scheduling_problem}. 
\item[(b)] $g^* = (S_1^*,S_2^*,\ldots)$ , where 
\begin{align}\label{Optimal_Scheduler}
S_{i+1}^*=\min_{t \in \mathbb Z} \big\{ t \geq D_i^* : \gamma(\Delta(t)) \geq \bar p_{opt}\big\},
\end{align}
$D_i^* = S_i^* + T_i$,  $\gamma(\delta)$ is defined in \eqref{gittins}, and $\bar p_{opt}$ is the optimal objective value of \eqref{scheduling_problem}, given by 
\begin{align}\label{optimammain}
\bar p_{opt}= \min_{b \in \{0, 1, \ldots, B-1\}} \bar p_{b, opt}.
\end{align}
\end{itemize}
\end{theorem}
\ifreport
\begin{proof}
See Appendix \ref{MainResult}.
\end{proof}
\else
\fi

Theorem \ref{theorem6} suggests that, in the optimal solution to \eqref{scheduling_problem}, one should select features from a fixed buffer position $b_i=b^*$. In addition, a feature is  transmitted in time-slot $t$ if and only if two conditions are satisfied: (i) The channel is available for transmission in time-slot $t$, (ii) the index $\gamma(\Delta(t))$ exceeds a threshold $\bar p_{opt}$ (i.e., $\gamma(\Delta(t)) \geq \bar p_{opt}$), where the threshold $\bar p_{opt}$ is exactly the optimal objective value of \eqref{scheduling_problem}.

In the special case of a non-decreasing AoI function $p(\delta)$, it can be shown that the index function $\gamma(\delta) = \mathbb E[p(\delta+T_{1})]$ is non-decreasing and $b^*=0$ is the optimal buffer position in \eqref{feature_Selects}. The optimal strategy in such cases is to consistently select the freshest feature from the buffer such that $b_i = 0$. Hence, both the ``generate-at-will'' and ``selection-from-buffer'' models achieve the same minimum inference error. Furthermore, Theorem 3 in \cite{SunNonlinear2019} can be directly derived from Theorem \ref{theorem6}.

\ignore{{\violet According to Theorem \ref{theorem6}, in an optimal solution to \eqref{scheduling_problem}, the $(i+1)$-th feature should be sent in the earliest time slot $t$ that satisfies (i) the previously sent feature has already been delivered, and (ii) the index $\gamma(\Delta(t))$ is no less than the optimal objective value \emph{$\bar p_{opt}$}. Moreover, the scheduler should always send the $(b^*+1)$-th freshest feature in the buffer, where $b^*$ is the minimizer of \eqref{feature_Selects}.}}



\ignore{\subsubsection{Special Case: Non-decreasing $p(\cdot)$} 
{\violet In the special case of non-decreasing age function $p(\cdot)$, the index $\gamma(\delta)$ in \eqref{gittins} reduces to 
\begin{align}\label{gittins_monotonic}
\gamma(\delta) = \mathbb E[p(\Delta(t+T_{1}))].
\end{align}
In this case, the following result follows immediately from Theorem \ref{theorem6} and \eqref{gittins_monotonic}, and it coincides with some earlier results in \cite{SunNonlinear2019, SunTIT2020, OrneeTON2021}.}

\begin{corollary}\label{corollary1}
{\violet If $p(\cdot)$ is a non-decreasing function and the random transmission times $T_i$ is i.i.d. with a finite mean $\mathbb E[T_i]$, then there exists an optimal policy $\pi^*=(f^*, g^*)$ to \eqref{scheduling_problem} that satisfies
\begin{itemize}
\item[(a)] $f^*=(0, 0, \ldots)$, i.e., $b_i=0$ for all $i$.
\item[(b)]  $g^* = (S_1^*,S_2^*,\ldots)$ , where 
\begin{align}
S_{i+1}^*=\min_{t \in \mathbb Z} \big\{ t \geq S_i^* + T_i: \mathbb E[p(\Delta(t+T_{1}))] \geq \bar p_{opt}\big\},
\end{align}
$S_i^*+T_i$ is the delivery time of the $i$-th feature, and the optimal objective value $\bar p_{opt}$ is equal to $\beta_0$, i.e., the unique root to \eqref{bisection} with $b= 0$.   
\end{itemize}}
\end{corollary}
{\violet In \cite{SunNonlinear2019, SunTIT2020, OrneeTON2021}, the term $\mathbb E[p(\Delta(t+T_{1}))]$ occurred in the optimal scheduling policy, but it did not have a good interpretation. In this section, we have shown that $\mathbb E[p(\Delta(t+T_{1}))]$ is the index of an AoI bandit process $\Delta(t)$ with cost $p(\Delta(t))$. In addition, the results in \cite{SunNonlinear2019, SunTIT2020, OrneeTON2021} were obtained under an additional condition that inter-sending times $\{S_{i+1}-S_i, i=1,2, \ldots\}$ form a regenerative process. This condition is no longer required in our new proof techniques for Theorems \ref{theorem5} and \ref{theorem6}, which directly solve the Bellman optimality equations of \eqref{scheduling_problem} and \eqref{sub_scheduling_problem}. Finally, Corollary \ref{corollary1} implies that, if $p(\cdot)$ is non-decreasing, it is optimal to keep $b_i=0$ for all $i$ and hence no feature buffer is needed (as in the "generate-at-will" model). However, if $p(\cdot)$ is non-monotonic, it could be beneficial to add a feature buffer and send an old feature that has been stored in the buffer for some time, as indicated by Theorem \ref{theorem6}.}}

%% file: Scheduling_MultiSource.tex
\section{Scheduling for Inference Error Minimization: The Multi-source, Multi-channel Case}\label{Multi-scheduling}

In this section, we develop a novel transmission
scheduling policy to minimize 
the weighted summation inference error in multi-source, multi-channel remote inference systems.  
\subsection{Multi-source, Multi-channel Status Updating Model}
\begin{figure}[t]
\centering
\includegraphics[width=0.40\textwidth]{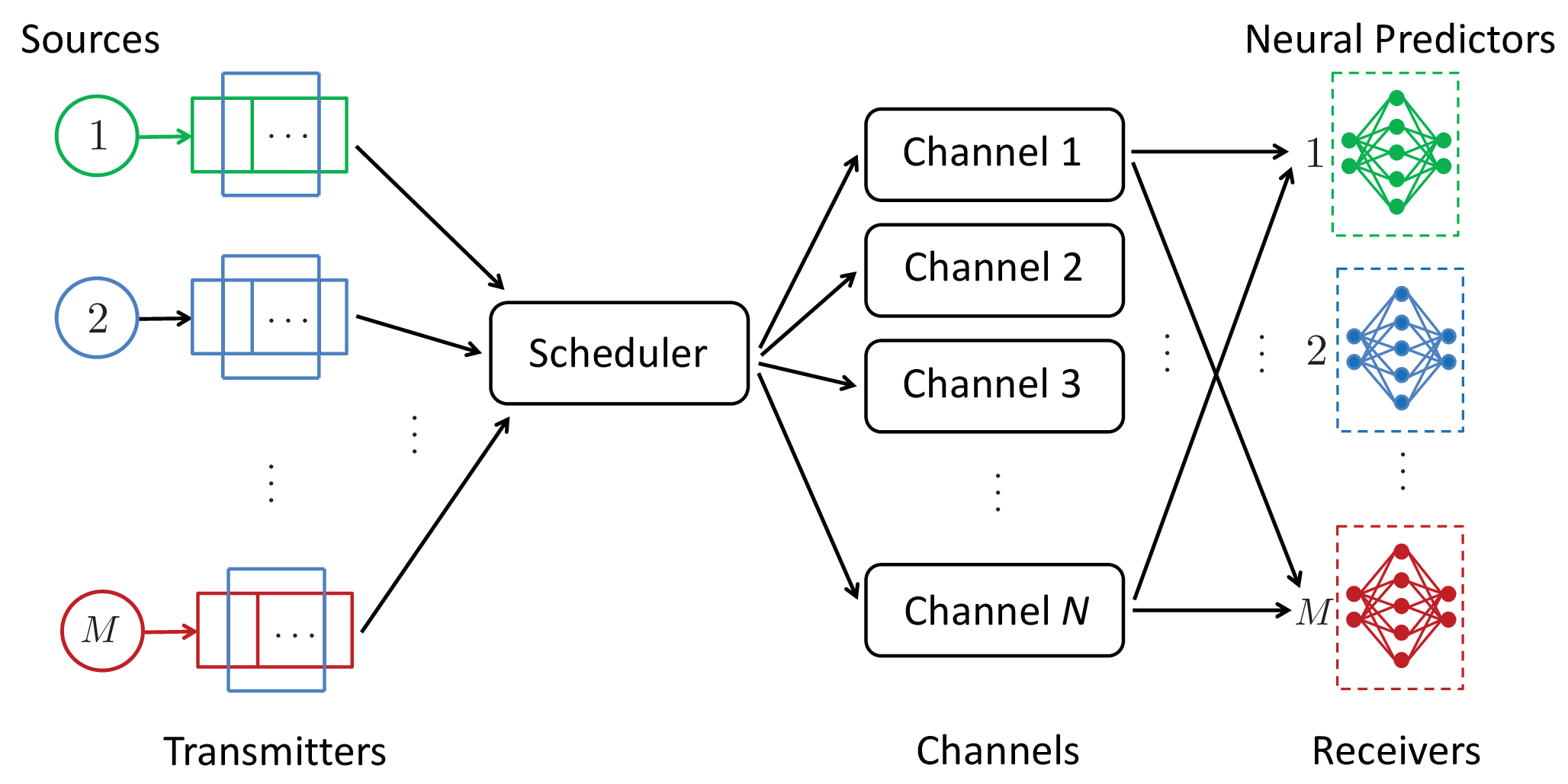}
\caption{\small A multi-source, multi-channel remote inference system. 
\label{fig:multi-scheduling}
}
\vspace{-3mm}
\end{figure}
Consider the remote inference system depicted in Fig. \ref{fig:multi-scheduling}, which consists of $M$ source-predictor pairs and $N$ channels. Each source adopts a ``selection-from-buffer'' model:  
At the beginning of time slot $t$, each source $m$ generates a feature $X_{m, t}$ and adds it into the buffer that stores $B_m$ most recent features $(X_{m,t},\ldots,$ $X_{m,t-B_m+1})$, meanwhile the oldest feature $X_{m, t-B_m}$ is removed from the buffer. 
At each time slot $t$, a central scheduler decides: (i) which
sources to select and (ii) which features from the buffer of selected sources to send. 
Each feature transmission lasts for one or multiple time slots. We consider non-preemptive transmissions, i.e., once a channel starts to send a feature, the channel must finish serving that feature before switching to serve another feature. At any given time slot, each source can be served by no more than one channel. 
We use an indicator variable $c_m(t) \in \{0, 1\}$ to represent whether a feature from source $m$ occupies a channel at time slot $t$, where $c_m(t)=1$ if source $m$ is being served by a channel at time slot $t$; otherwise, $c_m(t)=0$. Once a feature is delivered, an acknowledgment is fed back to the scheduler within the same time-slot.
By this, the channel occupation status $c_m(t)$ is known to the scheduler at every time slot $t$. Due to limited channel resources, the system must satisfy the constraint $\sum_{m=1}^M c_m(t) \leq N$ for all time slot $t=0, 1, \ldots$. 

The system starts to operate at time slot $t=0$. The $i$-th feature sent by source $m$ is generated in time slot $G_{m,i}$, submitted to a channel in time slot $S_{m, i}$, and delivered to the receiver in time slot 
$D_{m, i}=S_{m, i}+T_{m, i}$, where $G_{m, i} \leq S_{m,i} < D_{m,i}, D_{m,i} \leq S_{m,i+1} < D_{m,i+1}$, and  $T_{m, i}\geq 1$ is the feature transmission time. We assume that the $T_{m, i}$'s are independent across the sources and are i.i.d. for features originating from the same source with $1 \leq \mathbb E[T_{m,i}]<\infty$. 



\subsection{Scheduling Policies and Problem Formulation}
In time slot $S_{m, i}$, let $X_{G_{m,i}}=X_{S_{m, i}-b_{m,i}}$ be the feature submitted to a channel from source $m$, which is the $(b_{m,i} + 1)$-th freshest feature in source $m$'s buffer, with $b_{m, i} \in \{0, 1, \ldots, B_m-1\}$. By this,
a scheduling policy for source $m$ is denoted by $(f_m, g_m)$, where $g_m=(S_{m,1}, S_{m, 2}, \ldots)$ determines when to schedule source $m$, and $f_m=(b_{m,1}, b_{m, 2}, \ldots)$ specifies which feature to send from source $m$'s buffer. 

Let $U_m(t)= \max_i\{G_{m,i} : D_{m, i} \leq t \}$ represent the generation time of the freshest feature delivered from source $m$ to the receiver up to time slot $t$. Because $G_{m,i}=S_{m,i}-b_{m,i}$, $U_m(t)= \max_i\{S_{m, i}-b_{m,i} : D_{m,i} \leq t \}$. The age of information (AoI) of source $m$ at time slot $t$ is 
\begin{align}\label{multi-source-Age}
\!\!\!\!\Delta_m(t)\!=\!t\!-\!U_m(t)=\!t\!-\!\max_i\{S_{m, i}-b_{m, i}: D_{m, i} \leq t\}.
\end{align}
The initial state of the system is assumed to be $S_{m,0} = 0, D_{m,0} = T_{m,0}$, and $\Delta_m(0)$ is a finite constant.

Let $\Pi_m$ denote the set of all causal and signal-agnostic scheduling policies $(f_m, g_m)$ that satisfy the following conditions: (i) the transmission time schedule $S_{m, i}$ and the buffer position $b_{m, i}$ are determined based on the current and the historical information available at the scheduler; (ii) source $m$ can be served by at most one channel at a time and feature transmissions are non-preemptive; (iii) the scheduler does not
have access to the realization  of the feature and the target processes; and 
(iv) the scheduler can access the inference error function $\mathrm{err}_{\mathrm{inference},m}(\Delta_m(t))$ and the distribution of $T_{m, i}$ for source $m$. We define $\Pi$ as the set of all causal and signal-agnostic scheduling policies $\pi=(f_m, g_m)_{m=1}^M$. 

Our goal is to find a scheduling policy that minimizes the weighted summation of the time-averaged expected inference errors of the $M$ sources:
\begin{align}\label{Multi-scheduling_problem}
&\!\!\inf_{ \pi \in \Pi}  \limsup_{T\rightarrow \infty}\sum_{m=1}^M w_m~\mathbb{E} \left[ \frac{1}{T} \sum_{t=0}^{T-1} p_m(\Delta_m(t))\right], \\\label{Sceduling_constraint}
&\text{s.t.} \ \sum_{m=1}^M c_{m}(t) \leq N, t =0, 1, \ldots,\\\label{Sconstarint1}
&~~~~~c_m(t)\in \{0, 1\}, m=1,2,\ldots,M, t=0, 1, \ldots,
\end{align}
where $p_m(\Delta_m(t))=\mathrm{err}_{\mathrm{inference},m}(\Delta_m(t))$ is the inference error of source $m$ at time slot $t$ and $w_m>0$ is the weight of source $m$. 


Let $d_m(t) \in \{0, 1, \ldots\}$ denote the amount of time that has been spent to send the current feature of source $m$ by time slot $t$. Hence, $d_m(t) = 0$ if no feature of source $m$ is in service at time slot $t$, and $d_m(t) > 0$ if a feature of source $m$ is currently in service at time slot $t$. Problem \eqref{Multi-scheduling_problem}-\eqref{Sconstarint1} is a multi-action Restless Multi-armed Bandit (RMAB) problem, in which $(\Delta_m(t), d_m(t))$ is the state of the $m$-th bandit. 
At time slot $t$, if a feature from source $m$ is submitted to a channel, bandit $m$ is said to be \emph{active}; otherwise, if source $m$ is not under service or if one feature of source $m$ is under service whereas the service started before time slot $t$, then bandit $m$ is said to be \emph{passive}. 
The bandits are ``restless" because the state $(\Delta_m(t), d_m(t))$ undergoes changes even when the $m$-th bandit is passive \cite{whittle1988restless, Tripathi2019}.
When a bandit $m$ is activated, the scheduler can select any of the $B_m$ features from the buffer of source $m$. Thus, this problem is a multi-action RMAB.

{It is well-known that RMAB with binary actions is PSPACE-hard \cite{papadimitriou1994complexity}. RMABs with multiple actions, like \eqref{Multi-scheduling_problem}-\eqref{Sconstarint1}, would be even more challenging to solve. In the sequence, we will generalize the conventional Whittle index theoretical framework \cite{whittle1988restless} for binary-action RMABs, by developing a new index-based scheduling policy and proving this policy is asymptotically optimal for solving the multi-action RMAB problem \eqref{Multi-scheduling_problem}-\eqref{Sconstarint1}. This new theoretical framework contains four steps: (a) We first reformulate \eqref{Multi-scheduling_problem}-\eqref{Sconstarint1} as an equivalent multi-action RMAB problem with an equality constraint by using dummy bandits \cite{verloop2016asymptotically, ornee2023whittle}. The usage of dummy bandits is necessary for establishing the asymptotic optimality result in subsequent steps. (b) We relax the per-time-slot channel constraint as a time-average expected channel constraint, solve the relaxed problem by using Lagrangian dual optimization, and compute the optimal dual variable $\lambda^*$. 
(c) Problem \eqref{Multi-scheduling_problem}-\eqref{Sconstarint1} requires to determine (i) which source to schedule and (ii) which feature from the buffer of the scheduled source to send. In the proposed scheduling policy, the former is decided by a Whittle index policy, for which we establish indexability and derive an analytical expression of the Whittle index. The latter is determined by a $\lambda^*$-based selection-from-buffer policy. (d) Finally, we employ LP priority-based sufficient condition \cite{verloop2016asymptotically, gast2021lp} to prove that the proposed policy is asymptotically optimal as the numbers of users and channels increase to infinite with a fixed ratio.}

\subsection{Dummy Bandits and Constraint Relaxation}
Besides the original $M$ bandits, we introduce $N$ additional 
\emph{dummy bandits} that satisfy two conditions: (i) each dummy bandit has a zero age penalty function $p_0(\Delta_0(t)) = 0$; (ii) when activated, each dummy bandit occupies a channel. Let $c_0(t) \in \{0, 1, \ldots, N\}$ be the number of dummy bandits that are activated in time slot $t$. Let $\pi_0=\{c_{0}(t), t =0, 1, \ldots\}$ be a scheduling policy for the dummy bandits and $\Pi_0$ be the set of all policies $\pi_0$. Using these dummy bandits, \eqref{Multi-scheduling_problem}-\eqref{Sconstarint1} is reformulated as an  RMAB with  equality constraints \eqref{Sceduling_constraint1}, i.e.,
\begin{align}\label{Multi-scheduling_problem1}
&\inf_{ \pi \in \Pi, \pi_0 \in \Pi_0}  \limsup_{T\rightarrow \infty}\sum_{m=1}^M w_m~\mathbb{E} \left[\frac{1}{T} \sum_{t=0}^{T-1} p_m(\Delta_m(t))\right], \\ \label{Sceduling_constraint1}
&~~~~~\text{s.t.} \ \sum_{m=0}^M c_{m}(t) = N, t=0,1, \ldots, \\
&~~~~~~~\quad c_0(t) \in \{0, 1, \ldots, N\}, t=0,1,\ldots, \\\label{Sconstraint2}
&~~~~~~~\quad c_m(t)\in \{0, 1\}, m=1,2,\ldots,M, t=0,1, \ldots. 
\end{align}


Now, we replace the per-slot channel constraints \eqref{Sceduling_constraint1} by a time-average expected channel constraint \eqref{Changed_constraint}, which produces the following relaxed problem:
\begin{align}\label{relaxed_multiple}
&\inf_{ \pi \in \Pi, \pi_0 \in \Pi_0}  \limsup_{T\rightarrow \infty} \sum_{m=1}^M w_m~\mathbb{E} \left[  \sum_{t=0}^{T-1} p_m(\Delta_m(t))\right], \\
\label{Changed_constraint}
&~~~~~\text{s.t.}~ \limsup_{T \to \infty}   \sum_{m=0}^M \mathbb E\left[\frac{1}{T}\sum_{t=0}^{T-1}c_{m}(t)\right] = N, \\
&~~~~~~~ \quad c_0(t) \in \{0, 1, \ldots, N\}, t=0,1,\ldots,
\\\label{Sconstraint3}
&~~~~~~~ \quad c_m(t)\in \{0, 1\}, m=1,2,\ldots,M, t=0,1,\ldots. 
\end{align} 
The optimal objective value of \eqref{relaxed_multiple}-\eqref{Sconstraint3} provides a lower bound of the optimal objective value of \eqref{Multi-scheduling_problem}-\eqref{Sconstarint1}.

\subsection{Lagrangian Dual Optimization for Solving \eqref{relaxed_multiple}-\eqref{Sconstraint3}}
We solve the relaxed problem \eqref{relaxed_multiple}-\eqref{Sconstraint3} by Lagrangian dual optimization \cite{whittle1988restless, palomar2006tutorial}. To that end, we associate a Lagrangian multiplier $\lambda \in \mathbb R$ to the constraint \eqref{Changed_constraint} and get the following dual function  
\begin{align}\label{lagrangianproblem}
q(\lambda)=\inf_{\pi \in \Pi, \pi_0 \in \Pi_0} \limsup_{T\rightarrow \infty}~&\mathbb{E}\bigg[ \frac{1}{T} \sum_{t=0}^{T-1} \bigg(\sum_{m=1}^M w_m p_m(\Delta_m(t)) \nonumber\\ 
&+ \lambda \bigg(\sum_{m=0}^M c_{m}(t)-N\bigg)\bigg) \bigg],
\end{align}
where $\lambda \in \mathbb{R}$ is also referred to as the transmission cost. The dual optimization problem is given by 
\begin{align}\label{dualProblem}
\lambda^* = \arg\max_{\lambda \in \mathbb R} q(\lambda),
\end{align}
where $\lambda^*$ is the optimal dual solution.

\subsubsection{Solution to \eqref{lagrangianproblem}} The problem \eqref{lagrangianproblem} can be decomposed into $(M+1)$ sub-problems. For $m=0$, the sub-problem associated to the dummy bandits is given by
\begin{align}\label{dummyproblem}
\inf_{\pi_0 \in \Pi_0} \limsup_{T\rightarrow \infty} \mathbb E_{\pi_0}\left[\frac{1}{T}\sum_{t=0}^{T-1} \lambda c_0(t)\right].
\end{align} 
\begin{theorem}
If $\lambda > 0$, the optimal solution to \eqref{dummyproblem} is $c_0^*(t) = 0$ for all $t$; if $\lambda \leq 0$, the optimal solution to \eqref{dummyproblem} is $c_0^*(t) = N$ for all $t$.  
\end{theorem}
For each $m=1,2, \ldots, M$, the sub-problem associated with bandit $m$ is given by
\begin{align}\label{decoupled_problem}
\!\!\!\!\!&\bar p_{m, opt}(\lambda) \nonumber\\
\!\!\!\!\!=&\!\!\!\inf_{(f_{m}, g_m) \in \Pi_{m}} \!\!\!\limsup_{T\rightarrow \infty} \mathbb{E}\!\left[ \frac{1}{T}\!\!\sum_{t=0}^{T-1} \!w_m p_m(\Delta_m(t))\!+\! \lambda c_{m}(t) \right],\!\!\!\!
\end{align}
where $\bar p_{m, opt}(\lambda)$ is the optimal objective value to \eqref{decoupled_problem}. 

To explain the optimal solution to \eqref{decoupled_problem}, we first fix the buffer position at $b_{m,i} = b$ for all $i$ and optimize the transmission time schedule $g_m=(S_{m,1}, S_{m,2},\ldots)$. This simplified problem is formulated as
\begin{align}\label{sub_decoupled_problem}
&\bar p_{m, b, opt}(\lambda)\nonumber\\
&=\inf_{(f_{m, b}, g_m)\in\Pi_m}  \!\!\!\limsup_{T\rightarrow \infty} \mathbb{E}\!\left[ \frac{1}{T} \sum_{t=0}^{T-1} \!w_m p_m(\Delta_m(t))\!+\! \lambda c_{m}(t) \right],
\end{align}
where $f_{m,b} = (b, b, \ldots)$ and $\bar p_{m, b, opt}(\lambda)$ is the optimal objective value in \eqref{sub_decoupled_problem}.
\begin{theorem}\label{decoupled}
If $T_{m,i}$'s are i.i.d. with $1\leq \mathbb E[T_{m,i}]<\infty$, then $g_m(\lambda)=(S_{m,1}(\beta_{m,b}(\lambda)), S_{m,2}(\beta_{m,b}(\lambda)), \ldots)$ is an optimal solution to \eqref{sub_decoupled_problem}, where 
\begin{align}\label{OptimalPolicy_fixed_length_fixed_buffer}
&S_{m, i+1}(\beta_{m,b}(\lambda))\nonumber\\ =& \min_{t \in \mathbb Z}\big\{ t \geq D_{m,i}(\beta_{m,b}(\lambda)): \gamma_m(\Delta_m(t)) \geq \beta_{m,b}(\lambda)\big\},
\end{align}
$D_{m,i}(\beta_{m,b}(\lambda))=S_{m,i}(\beta_{m,b}(\lambda))+T_{m,i}$ is the delivery time of the $i$-th feature submitted to the channel, $\Delta_m(t)=t-S_{m,i}(\beta_{m,b}(\lambda))+b$ is the AoI at time $t$,  $\gamma_m(\delta)$ is an index function, defined by
\begin{align}\label{gittinsM}
\gamma_m(\delta)=\inf_{\tau \in \{1, 2, \ldots\}} \frac{1}{\tau} \sum_{k=0}^{\tau-1} \mathbb E \left [w_m p_m(\delta+k+T_{m, 1}) \right],
\end{align}
and the threshold $\beta_{m,b}(\lambda)$ is the unique root of 
\begin{align}\label{bisectionM}
&\!\!\!\mathbb{E}\left[\sum_{t=D_{m,i}(\beta_{m,b}(\lambda))}^{D_{m, i+1}(\beta_{m,b}(\lambda))-1}  w_m p_m\big(\Delta_m(t)\big)\right]+\lambda~\mathbb E[T_{m, i}]\nonumber\\ 
&\!\!\!- \beta_{m,b}(\lambda)~ \mathbb{E}\left[D_{m, i+1}(\beta_{m,b}(\lambda))-D_{m,i}(\beta_{m,b}(\lambda))\right]=0.
\end{align}
Furthermore, $\beta_{m,b}(\lambda)$ is equal to the optimal objective value to \eqref{sub_decoupled_problem}, i.e., $\beta_{m,b}(\lambda)= \bar p_{m,b, opt}(\lambda)$.
\end{theorem}

\begin{proof}
See Appendix \ref{MainResult}.
\end{proof}

Now we present an optimal solution to \eqref{decoupled_problem}.

\begin{theorem}\label{decoupled1}
If the conditions of Theorem \ref{decoupled} hold, then an optimal solution $(f^*_m(\lambda), g^*_m(\lambda))$ to \eqref{decoupled_problem} is determined by
\begin{itemize}
\item[(a)] $f^*_m(\lambda)=(b^*_m(\lambda), b^*_m(\lambda), \ldots)$, where 
\begin{align}\label{optimal_buffer_length_1}
b^*_m(\lambda)= \arg \min_{b \in \{0, 1, \ldots, B_m-1\}} \bar p_{m, b, opt}(\lambda),
\end{align}
and $\bar p_{m, b, opt}(\lambda)$ is the optimal objective value to \eqref{sub_decoupled_problem}. 
\item[(b)] $g^*_m(\lambda) = (S_{m,1}^*(\lambda),S_{m,2}^*(\lambda),\ldots)$, where 
\begin{align}\label{Optimal_Scheduler}
&S_{m, i+1}^*(\lambda)\nonumber\\
=&\min_{t \in \mathbb Z} \big\{ t \geq D_{m,i}^* : \gamma_m(\Delta_m(t)) \geq \bar p_{m, opt}(\lambda)\big\},
\end{align}
$D_{m,i}^*(\lambda) = S_{m,i}^*(\lambda) + T_{m,i}$,  $\gamma_m(\delta)$ is defined in \eqref{gittinsM}, and $\bar p_{m, opt}(\lambda)$ is the optimal objective value of \eqref{decoupled_problem}, given by 
\begin{align}\label{optimammain}
\bar p_{m,opt}(\lambda)= \min_{b \in \{0, 1, \ldots, B_m-1\}} \bar p_{m, b, opt}(\lambda).
\end{align}
\end{itemize}
\end{theorem}
\ifreport
\begin{proof}
See Appendix \ref{MainResult}.
\end{proof}
\else
\fi

\subsubsection{Solution to \eqref{dualProblem}} Next, we solve the dual problem \eqref{dualProblem}. Let $c_{0,\lambda}^*(t)\in \{0,1,\ldots, N\}$ be the number of dummy bandits activated in time slot $t$ in the optimal solution to \eqref{dummyproblem} and let $c_{m,\lambda}^*(t) \in \{0,1\}$ denote whether source $m$ is under service in time slot $t$ in the optimal solution to \eqref{decoupled_problem}. The dual problem \eqref{dualProblem} is solved by the following stochastic sub-gradient algorithm: 
\begin{align}\label{eq_dual}
\lambda_{k+1} =  \lambda_{k} + \frac{\alpha}{k}
\left\{\frac{1}{T}\sum_{m=0}^M \sum_{t=0}^{T-1} c_{m,\lambda_k}^*(t) - N\right\},
\end{align}
where $\alpha/k > 0$ is the step size and $T > 1$ is a sufficient large integer. In the $k$-th iteration, let $\lambda=\lambda_k$ and run the optimal solution to \eqref{lagrangianproblem} for $T$ time slots, then execute the dual update \eqref{eq_dual}. 

\subsection{A Scheduling Policy for the Original Problem \eqref{Multi-scheduling_problem}-\eqref{Sconstarint1}}
Now, we develop a scheduling policy for the original multi-action RMAB problem \eqref{Multi-scheduling_problem}-\eqref{Sconstarint1}. The proposed policy contains two parts: (i) a Whittle index policy is used to determine which sources to schedule, and (ii) a $\lambda^*$-based selection-from-buffer policy is employed to determine which features to choose from the buffers of the scheduled sources.  

\subsubsection{Whittle Index-based Source Scheduling Policy}
The Whittle index theory only applies to RMAB problems that are indexable \cite{whittle1988restless}.
Hence, we first establish the indexibility of problem \eqref{Multi-scheduling_problem1}-\eqref{Sconstraint2}. Define $\Omega_m(\lambda)$ as the set of all states $(\delta, d)$ such that if $\Delta_m(t)=\delta$ and $d_m(t)=d$, then the optimal solution for \eqref{decoupled_problem} is to take the passive action at time $t$.
 
\begin{definition}[\textbf{Indexability}]\label{indexability}\cite{verloop2016asymptotically}
Bandit $m$ is said to be
indexable if, as the cost $\lambda$ increases from $-\infty$ to $\infty$, the set $\Omega_m(\lambda)$ increases monotonically, i.e., $\lambda_1 \leq \lambda_2$
implies $\Omega_m(\lambda_1) \subseteq \Omega_m(\lambda_2)$.
The RMAB problem \eqref{Multi-scheduling_problem1}-\eqref{Sconstraint2} is said to be indexable if all $(M+1)$ bandits are indexable.
\end{definition}


\begin{theorem}\label{theorem7}
The RMAB problem \eqref{Multi-scheduling_problem1}-\eqref{Sconstraint2} is indexable.
\end{theorem}
\ifreport
\begin{proof}
See Appendix \ref{ptheorem7}.
\end{proof}
\else
\fi
\begin{definition}[\textbf{Whittle index}]\cite{verloop2016asymptotically}
Given indexability, the Whittle index $W_{m}(\delta, d)$ of bandit $m$ at state $(\delta, d)$ is 
\begin{align}\label{defWhittle}
W_{m}(\delta, d) = \inf \{\lambda \in\mathbb R: (\delta, d) \in \Omega_{m}(\lambda)\}. 
\end{align}
\end{definition}

\begin{lemma}\label{dummyWhittle}
The Whittle index of the dummy bandits is 0, i.e., $W_0(\delta, 0)=0$ for any $\delta$.
\end{lemma}

\begin{theorem}\label{theorem8}
The following assertions are true for the Whittle index $W_m(\delta, d)$ of bandit $m$ for $m=1,2,\ldots,M$: 

(a) If $d=0$, then for $m=1, \ldots, M$,
\begin{align}\label{WhittleIndex}
W_m(\delta, 0)=\max_{b \in \mathbb Z: 0 \leq b \leq B_m-1} W_{m, b}(\delta, 0),
\end{align}
where 
\begin{align}\label{Whittle_Index}
&W_{m, b}(\delta, 0)\nonumber\\
=&\frac{1}{\mathbb E[T_{l,1}]}~\mathbb{E}\bigg[D_{m, i+1}(\gamma_m(\delta))-D_{m, i}(\gamma_m(\delta))\bigg]~\gamma_m(\delta) \nonumber\\
&-\frac{1}{\mathbb{E}[T_{l,1}]}\mathbb E\left[\sum_{t=D_{m,i}(\gamma_m(\delta))}^{D_{m, i+1}(\gamma_m(\delta))-1}  w_m p_m(\Delta_{m}(t))\right],
\end{align}
$\Delta_m(t)=t-S_{m, i}(\gamma_m(\delta))+b$, $\gamma_m(\delta)$ is defined in \eqref{gittinsM}, $D_{m,i+1}(\gamma_m(\delta))= S_{m, i+1}(\gamma_m(\delta))+T_{m,i}$, and $S_{m, i+1}(\gamma_m(\delta))$ is given by
\begin{align}\label{OptimalPolicyWhittle}
&S_{m, i+1}(\gamma_m(\delta))\nonumber\\ =& \min_{t \in \mathbb Z}\big\{ t \geq D_{m,i}(\gamma_m(\delta)): \gamma_m(\Delta_m(t)) \geq \gamma_m(\delta)\big\}.
\end{align}
(b) If $d>0$, then $W_{m}(\delta, d)=-\infty$ for $m=1, \ldots, M$.
\end{theorem}
\ifreport
\begin{proof}
See Appendix \ref{ptheorem8}.
\end{proof}
\else
\fi

\begin{algorithm}[t]
\caption{Scheduling Policy for Multi-source, Multi-channel Inference Error Minimization \eqref{Multi-scheduling_problem}-\eqref{Sconstarint1}}\label{alg:Whittle}
\begin{algorithmic}[1]
{\State Initialize $t=0$
\State Input the optimal dual variable $\lambda^*$ of the problem \eqref{dualProblem}.
\For{$t=0, 1, \ldots$}
\State Update $\Delta_m(t)$ and $d_m(t)$ for all $m$.
\State  $W_m \gets W_{m}(\Delta_m(t), d_m(t))$ for all $m$.
      \For{all channel $n = 1, 2, \ldots, N$}
      \If{channel $n$ is idle and $\max_l W_l>0$} 
      \State Schedule source $m = \arg\max_l W_l$.
       \State Send the feature from the position   $b^*_{m}(\lambda^*)$ in \State source $m$'s buffer.
       \State $W_m \gets - \infty$. 
         
      \EndIf
      \EndFor
      \State $t \gets t+1$. 
      \EndFor}
\end{algorithmic}
\end{algorithm}

Theorem \ref{theorem8} presents an analytical expression of the Whittle index of bandit $m$ for $m=1, 2, \ldots, M$. If no feature of source $m$ is being served by a channel at time slot $t$ such that $d_m(t)=0$, then the Whittle index of bandit $m$ at time slot $t$ is determined by \eqref{WhittleIndex}. Otherwise, if source $m$ is being served by a channel at time slot $t$ such that $d_m(t) > 0$, then the Whittle index of bandit $m$ at time $t$ is $-\infty$.

In the special case that (i) the AoI function $p_m(\cdot)$ is non-decreasing and (ii) $T_{m,i} = 1$, it holds that $\gamma_m(\delta) = w_mp_m(\delta+1)$ and for  $\delta=1, 2, \ldots,$
\begin{align}
W_{m}(\delta, 0) = w_m\left[\delta~p_m(\delta+1)-\sum_{k=1}^{\delta} p_m(k)\right].
\end{align} By this, the Whittle index in Section IV of \cite[Equation (7)]{Tripathi2019} is recovered from Theorem \ref{theorem8}. 

Let $A(t)$ denote the number of available channel at the beginning of time slot $t$, where $A(t) \leq N$. Then, $A(t)$ bandits with the highest Whittle index are activated at any time slot $t$. As stated in Lemma \ref{dummyWhittle}, all $N$ dummy bandits have Whittle index of $W_0(\Delta_0(t), d_0(t))=0$. Consequently, if a bandit $m$ (for $m = 1,2,...,M$) possesses a negative Whittle index, denoted as $W_m(\Delta_m(t), d_m(t))<0$, it will remain inactive. 

Now, we return to the original RMAB \eqref{Multi-scheduling_problem}-\eqref{Sconstarint1} with no dummy bandits. As illustrated in Algorithm \ref{alg:Whittle}, 
if channel $n$ is idle, then source $m$ with highest non-negative Whittle index is activated.

\subsubsection{$\lambda^*$-based Selection-from-Buffer Policy}\label{DualityAction}
{When the $m$-th bandit is activated, our policy in Algorithm \ref{alg:Whittle} sends the feature from the buffer position $b^*_m(\lambda^*)$, determined by  
\begin{align}\label{optimal_buffer_length_2}
b^*_m(\lambda^*)= \arg \min_{b \in \{0, 1, \ldots, B_m-1\}} \bar p_{m, b, opt}(\lambda^*),
\end{align}
where $\lambda^*$ is the optimal solution to \eqref{dualProblem} and $\bar p_{m, b, opt}(\lambda)$ is the optimal objective value of \eqref{sub_decoupled_problem}.} 


{\subsection{Asymptotic Optimality of the Proposed Scheduling Policy}


Let $\pi_{\mathrm{our}}$ denote the scheduling policy outlined in Algorithm \ref{alg:Whittle}. Now, we demonstrate that $\pi_{\mathrm{our}}$ is asymptotically optimal. 


\begin{definition}[\textbf{Asymptotically optimal}]\label{defAsymptotically optimal}\cite{verloop2016asymptotically, gast2021lp}
Initially, we have $N$ channels and $M$ bandits. Let $\bar{p}_{\pi}^{r}$ represent the expected long-term average inference error under policy $\pi$, where both the number of channels $rN$ and the number of bandits $rM$ are scaled by $r$.
The policy $\pi_{\mathrm{our}}$ will be asymptotically optimal if $\bar p_{\pi_{\mathrm{our}}}^{r} \leq \bar p_{\pi}^{r}$ for all $\pi \in \Pi$ as $r$ approaches $\infty$, while maintaining a constant ratio $\alpha=\frac{rM}{rN}$. 
\end{definition}

We denote by $V^m_{\delta, d}(t)$ the fraction of class $m$ bandits with $\Delta_m(t)=\delta$ and $d_m(t)=d$. Also, define variables $v^m_{\delta, d}$ for all $\delta, d$, and $m$ as follows:
\begin{align}
 v^m_{\delta, d}=\limsup_{T\rightarrow \infty}\sum_{t=0}^{T-1} \frac{1}{T} \mathbb E [ V^m_{\delta, d}(t)].
\end{align}

If the transmission times $T_{m, i}$ are bounded for all $m$ and $i$, then we can find a $d_{\text{bound}} \in \mathbb N$ such that $0< T_{m, i} \leq d_{\text{bound}}$ for all $m$ and $i$. Then, the amount of time $d_m(t)$ that has been spent to send the current feature of source $m$ by time slot $t$ is also bounded, i.e., $d_m(t)\in \{0, 1, \ldots, d_{\text{bound}}\}.$ Also, we observe from Figs. \ref{fig:learning}-\ref{fig:Trainingcsi} that the inference error function $p_m(\delta)$ converges after a large AoI value. We can find an AoI $\delta_{\text{bound}}$ such that $p_m(\delta)=p_m(\delta_{\text{bound}})$ for all $\delta \geq \delta_{\text{bound}}$ and for all $m$. Therefore, we can analyze using a truncated state space $(\delta, d) \in \{1, 2, \ldots, \delta_{\text{bound}}\} \times \{0, 1, \ldots, d_{\text{bound}}\}$.

We further define the vectors $\mathbf{V}^m(t)$ and  $\mathbf{v}^m$ to contain $V^m_{\delta, d}$
and $v^m_{\delta, d}$, respectively, for all $\delta=1, 2, \ldots, \delta_{\mathrm{bound}}$, and $d=0, 1, \ldots, d_{\mathrm{bound}}$. Here, we consider a truncated state space. 

Now, we provide a uniform global attractor condition. For a policy $\pi$, we can have a mapping $\Psi_{\pi}$ of the state transitions, given by
\begin{align}
&\!\!\!\!\Psi_{\pi}((\mathbf{v}^m)_{m=1}^M)=\nonumber\\
&\!\!\!\!\mathbb{E}_{\pi}[(\mathbf{V}^m(t))_{m=1}^M(t+1)|(\mathbf{V}^m(t))_{m=1}^M(t)=(\mathbf{v}^m)_{m=1}^M].
\end{align}
We define the $t$-th iteration of the maps $\Psi_{\pi, t\geq 0}(\cdot)$ as $\Psi_{\pi, 0}((\mathbf{v}^m)_{m=1}^M)=(\mathbf{v}^m)_{m=1}^M$ and $\Psi_{\pi, t+1}((\mathbf{v}^m)_{m=1}^M)=\Psi_{\pi}(\Psi_{\pi, t}((\mathbf{v}^m)_{m=1}^M))$.

\begin{definition}\cite{gast2021lp} \label{def2}
\textbf{Uniform Global attractor.} An equilibrium point $(\mathbf v^{m*})_{m=1}^M$ associated with the optimal solution of \eqref{relaxed_multiple}-\eqref{Sconstraint3} is a uniform global attractor of $\Psi_{\pi, t\geq 0}(\cdot)$, i.e., for all $\epsilon>0$, there exists $T(\epsilon)>0$ such that for all $t \geq T(\epsilon)$ and for all $(\mathbf v^{m})_{m=1}^M$, one has $\|\Psi_{\pi,t}((\mathbf v^{m})_{m=1}^M)-(\mathbf v^{m*})_{m=1}^M)\|_{1}\leq \epsilon$.
\end{definition}

\begin{theorem}\label{asymptotic_optimal}
If the feature transmission times $T_{m, i}$ are bounded for all $m$ and $i$, then under the uniform global attractor condition provided in Definition \ref{def2}, $\pi_{\mathrm{our}}$ is asymptotically optimal.
\end{theorem} 
\begin{proof}[Proof sketch] {\blue We first establish that if $m$-th source is selected, then there exists an optimal feature selection policy that always selects features from the buffer position $b^*_m(\lambda^*)$. Hence, the multiple action RMAB problem \eqref{Multi-scheduling_problem}-\eqref{Sconstarint1} reduces to a binary action RMAB problem. Then, we use \cite[Theorem 13]{gast2021lp} to prove Theorem \ref{asymptotic_optimal}. See Appendix \ref{pasymptotic_optimal} for details.} 
\end{proof}

%% file: Simulations.tex

\section{Data Driven Evaluations}
{\blue In this section, we illustrate the performance of our scheduling policies, where we plug in the inference
error versus AoI cost functions from the data driven experiments in Section \ref{Experimentation}. Then, we simulate the performance of different scheduling
policies.}
\subsection{Single-source Scheduling Policies}
We evaluate the following four  scheduling policies:
\begin{itemize}
\item[1.] Generate-at-will, zero wait: The $(i+1)$-th feature sending time $S_{i+1}$ is given by 
$S_{i+1}=D_i = S_i+T_i$ and the feature selection policy is $f=(0, 0, \ldots)$, i.e., $b_i=0$ for all $i$. 

\item[2.] Generate-at-will, optimal scheduling: The policy is given by Theorem \ref{theorem5} with $b_i=0$ for all $i$.

\item[3.] Selection-from-buffer, optimal scheduling: The policy is given by Theorem \ref{theorem6}.

\item[4.] Periodic feature updating: Features are generated periodically with a period $T_p$ and appended to a queue with buffer size $B$. When the buffer is full, no new feature is admitted to the buffer. Features in the buffer are sent over the channel in a first-come, first-served order.  
\end{itemize}

\begin{figure}[t]
\centering
\includegraphics[width=0.35\textwidth]{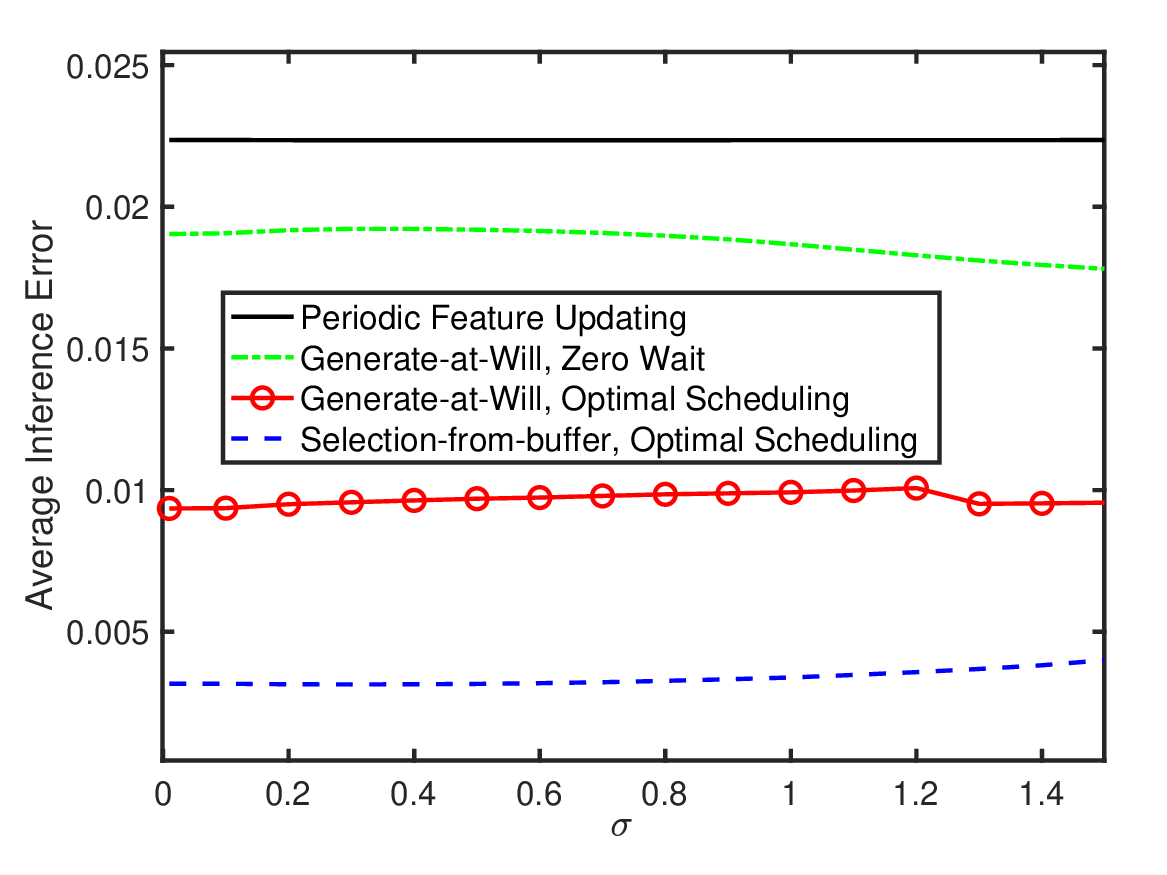}
\caption{\small Time average inference error vs. the scale parameter $\sigma$ of discretized i.i.d. log-normal transmission time distribution for single-source scheduling (in robot state prediction task). \label{fig:singlesourcedifferentsigma}
}
\end{figure}

\ifreport

\begin{figure}[t]
\centering
\includegraphics[width=0.35\textwidth]{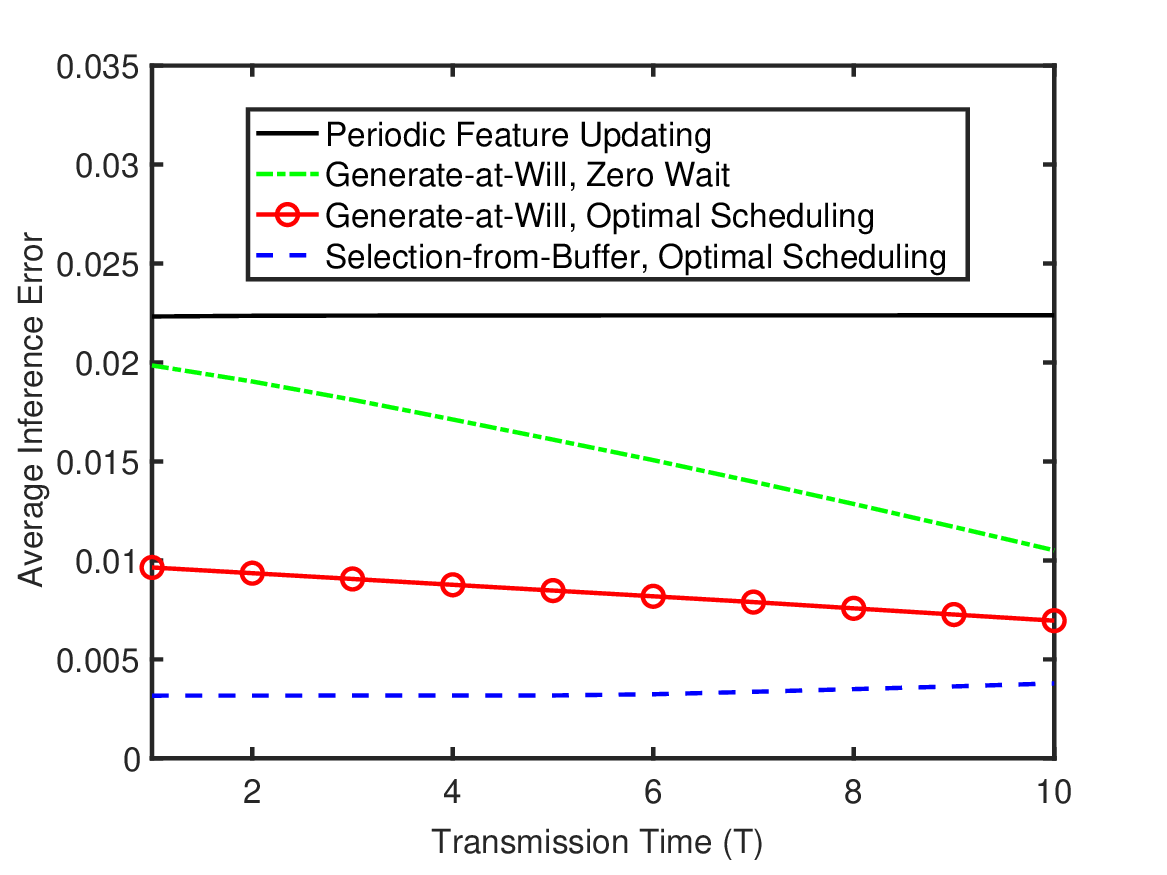}
\caption{\small Time-average inference error vs. constant transmission time (T) (in robot state prediction task). \label{fig:singlesourcedifferenttime}
}
\vspace{-3mm}
\end{figure}
\else
\fi

\begin{figure}[t]
\centering
\includegraphics[width=0.35\textwidth]{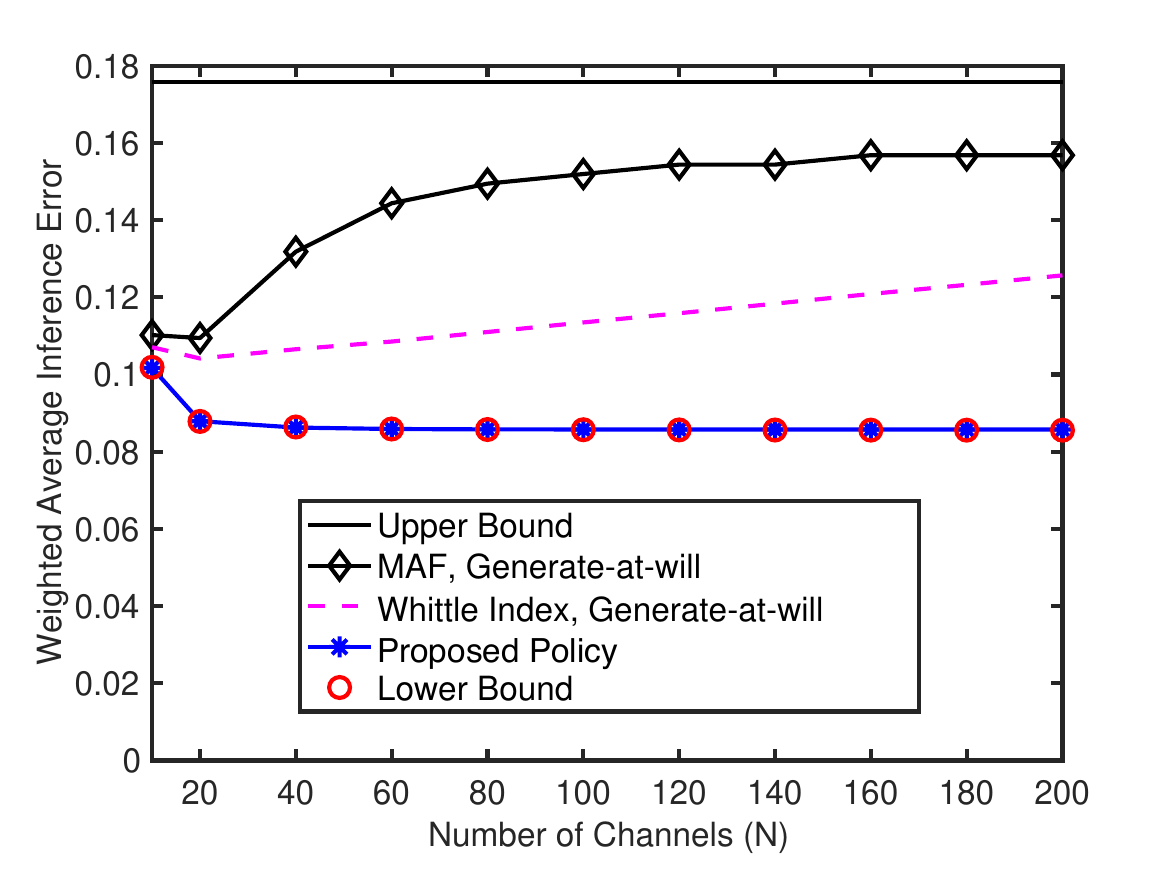}
\caption{\small Time-average weighted inference error vs. number of channels (N).  \label{fig:multisourceN}
}
\vspace{-3mm}
\end{figure}
\ifreport
\begin{figure}[t]
\centering
\includegraphics[width=0.35\textwidth]{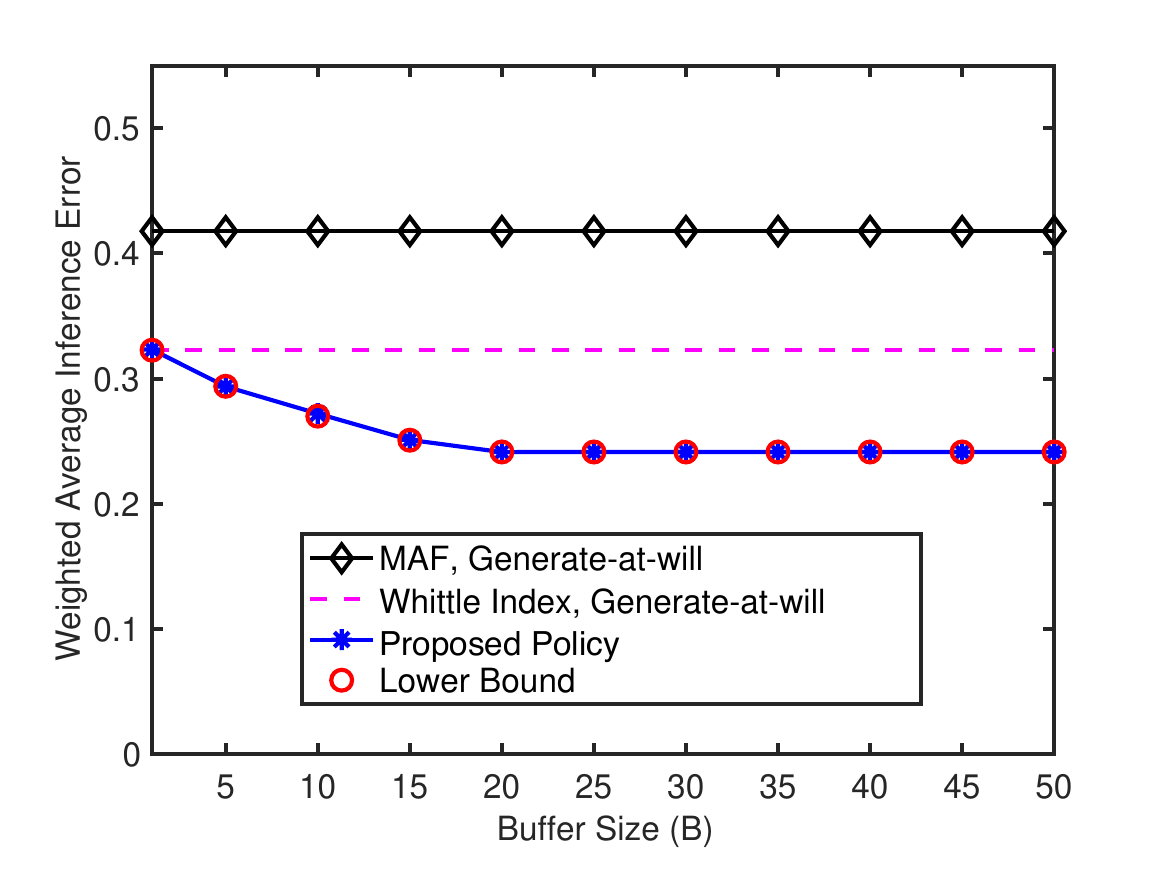}
\caption{\small Weighted time-average inference error vs. buffer size (B). \label{fig:multisourceweight1}
}
\vspace{-3mm}
\end{figure}


\else
\fi

{\blue Figs. \ref{fig:singlesourcedifferentsigma}-\ref{fig:singlesourcedifferenttime} compare the time-averaged inference error of the four single-source scheduling policies defined earlier. These policies are evaluated using the inference error function, $p(\delta)$, obtained from the robot state prediction experiment of the leader-follower robotic system presented in Section \ref{Experimentation} and illustrated in Fig. \ref{fig:DelayedNetworkedControlled}(c). The feature sequence length for this experiment is $u = 1$.} 

In Fig. \ref{fig:singlesourcedifferentsigma}, the $i$-th feature transmission time $T_i$ is assumed to follow a discrete valued i.i.d. log-normal distribution. In particular, $T_i$ can be expressed as $T_i=\lceil \alpha e^{\sigma Z_i}/ \mathbb E[e^{\sigma Z_i}] \rceil$, where $Z_i$'s are i.i.d. Gaussian random variables with zero mean and unit variance. In Fig. \ref{fig:singlesourcedifferentsigma}, we plot the time average inference error versus the scale parameter $\sigma$ of discretized i.i.d. log-normal distribution, where $\alpha=1.2$, the buffer size is $B=30$, and the period of uniform sampling is $T_p=3$. The randomness of the transmission time increases with the growth of $\sigma$. Data-driven evaluations in Fig. \ref{fig:singlesourcedifferentsigma} show that ``selection-from-buffer” with optimal scheduler achieves $3$ times performance gain compared to ``generate-at-will,” and $8$ times performance gain compared to periodic feature updating.

\ifreport

In Fig. \ref{fig:singlesourcedifferenttime}, the time average inference error versus constant transmission time $T$. This figure also shows that ``selection-from-buffer” with optimal scheduler can achieve $8$ time performance gain compared to periodic feature updating.

\else
\fi

\subsection{Multiple-source Scheduling Policies}

Now, we evaluate the following three multiple-source scheduling policies:
\begin{itemize}
 \item[1.] Maximum age first (MAF), Generate-at-will: At time slot $t$, if a channel is free, this policy schedules the freshest generated feature from source $\arg\max_{l \in A(t)} \Delta_l(t)$, where $A(t)$ is the set of available sources in time slot $t$.
  \item[2.] Whittle index, Generate-at-will: Denote 
  \begin{align}
  l_0^*&=\argmax_{l \in A(t)} \!W_{l, 0}(\Delta_l(t)). 
  \end{align}
If a channel is free and $\max_{l \in A(t)} \!W_{l, 0}(\Delta_l(t)) \geq 0$, 
the freshest feature from the source $l_0^*$ is scheduled; otherwise, no source is scheduled.  
 \item[3.] Proposed Policy: The policy is described in Algorithm \ref{alg:Whittle}. 
 \item[4.] Lower bound: Given the optimal dual variable $\lambda=\lambda^*$, the lower bound is obtained by implementing policy $(f^*_m(\lambda^*), g^*_m(\lambda^*))_{m=0}^M$, which is defined in Theorem \ref{decoupled1}.
 \item[5.] Upper bound: The upper bound is obtained if none of the sources is scheduled at every time slot $t$.
\end{itemize}

{\blue Figs. \ref{fig:multisourceN}-\ref{fig:multisourceweight1} compare our proposed policy with others under a multi-source scenario with $m=500$ sources. The inference error for half the sources (weight $w_1 = 1$, feature sequence length $u = 5$) originates from the pole angle prediction of the CartPole-v1 task in Sec. \ref{Experimentation} (Fig. \ref{fig:TrainingCartVelocity}(c)). The remaining sources (weight $w_2 = 5$, feature sequence length $u = 1$) use inference error from the robot state prediction experiment of the leader-follower robotic system presented in Sec.~\ref{Experimentation} (Fig.~\ref{fig:DelayedNetworkedControlled}(c)). Notably, the transmission time for all features is considered constant at 1.}

In Fig. \ref{fig:multisourceN}, we plot the weighted time-average inference error versus the number of channels, where the buffer size of all sources is set to $40$ (i.e., $B_l=B=40$ for all $l$). From Fig. \ref{fig:multisourceN}, it is evident that our proposed policy outperforms the ``Whittle index, Generate-at-will" and ``MAF, Generate-at-will" policies. Specifically, our policy achieves a weighted average inference error that is twice as low as that of the ``MAF, Generate-at-will" policy. Furthermore, as shown in Fig. \ref{fig:multisourceN}, the performance of our policy matches the lower bound of the multi-source, multi-action scheduling problem, thereby validating its asymptotic optimality.

Fig. \ref{fig:multisourceweight1} illustrates the weighted time-average inference error versus the buffer size $B$, with the number of channels set to $N=50$. The results presented in Fig. \ref{fig:multisourceweight1} underscore the effectiveness of the ``selection-from-buffer" model. The weighted time-average inference error achieved by our policy decreases as the buffer size $B$ increases, eventually reaching a plateau at a buffer size of $20$.


%% file: Conclusions.tex
\section{Conclusions}
In this paper, we explored the impact of data freshness on the performance of remote inference systems. Our analysis revealed that the inference error in a remote inference system is a function of AoI, but not necessarily a monotonic function of AoI. If the target and feature data sequence satisfy a Markov chain, then the inference error is a monotonic function of AoI. Otherwise, if the target and feature data sequence is far from a Markov chain, then the inference error can be a non-monotonic function of AoI.  
To reduce the time-average inference error, we introduced a new feature transmission model called ``selection-from-buffer'' and designed an optimal single-source scheduling policy. The optimal single-source scheduling policy is found to be a threshold policy.  Additionally, we developed a new asymptotically optimal policy for multi-source scheduling. Our numerical results validated the efficacy of the proposed scheduling policies. 

\section*{Acknowledgement}
The authors are grateful to Vijay Subramanian for one suggestion, to John Hung for useful discussions on this work, and to Shaoyi Li for his help on Fig. \ref{fig:learning}(b)-(c).

%% file: Appendix_Paper.tex
\ignore{
\begin{figure}
\centering
\includegraphics[width=0.30\textwidth]{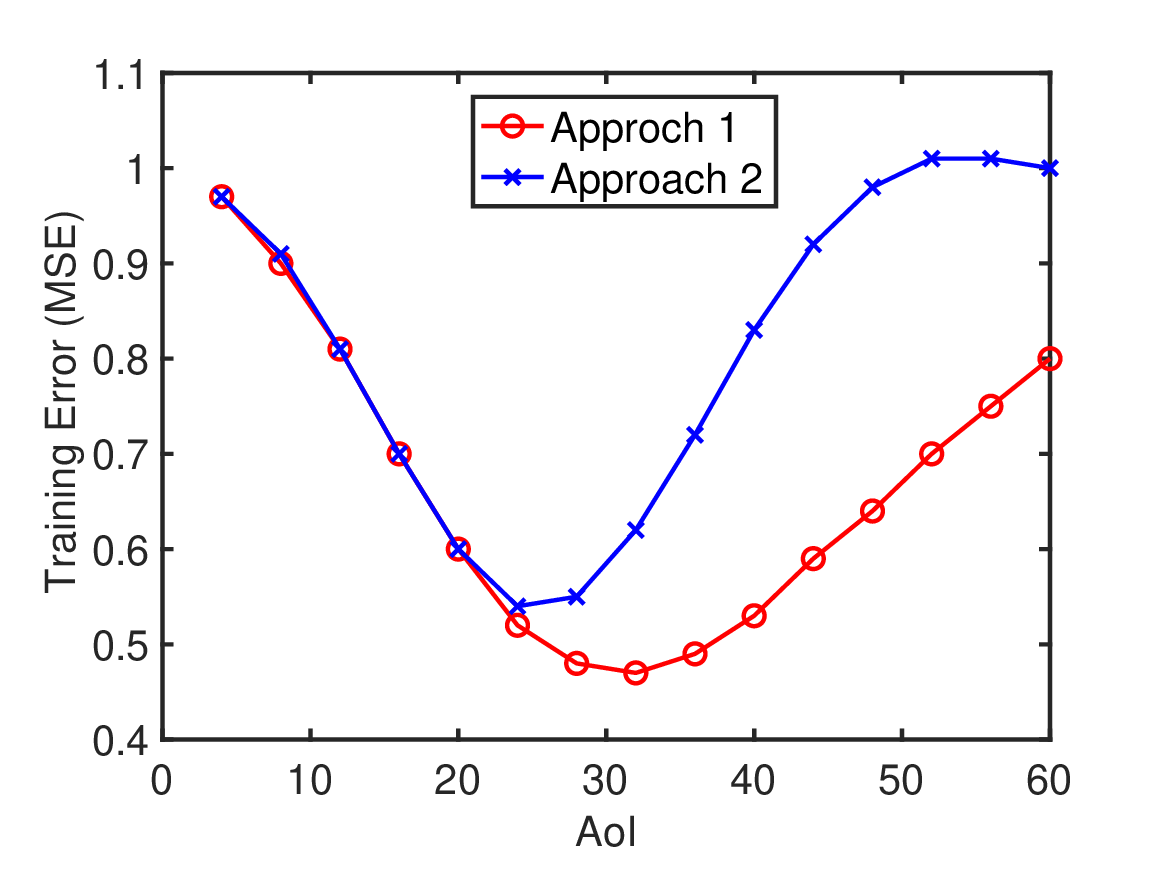}
\caption{\small Training Error vs. AoI. In the first approach, multiple neural networks are trained independently and in parallel, using distinct training datasets with different AoI values. The second approach involves training a single neural network using a larger dataset that encompasses a variety of AoI values.\label{fig:twoapproach}
}
\vspace{-3mm}
\end{figure}}
\section{Detailed Settings of Experiments in Sec. \ref{Experimentation}}\label{Experiments}
{\blue In all five experiments, we employed the first training method described in Section \ref{learningerror}. This approach involves training multiple neural networks independently and in parallel, each using a distinct dataset with a different AoI value. In contrast, the second approach trains a single neural network on a larger, combined dataset encompassing various AoI values. Due to the smaller dataset sizes for each network, the first approach can potentially have a shorter training time than the second approach. The experimental settings of the five experiments are provided below:} 
\subsection{Video Prediction} 

In video prediction experiment, a pre-trained neural network model called ``SAVP" \cite{lee2018stochastic} is used to evaluate on $256$ samples of ``BAIR" dataset \cite{ebert17sna}, which contains video frames of a randomly moving robotic arm. The pre-trained neural network model can be downloaded from the GitHub repository of \cite{lee2018stochastic}. 
\subsection{Robot State Prediction} 
In this experiment, we consider a leader-follower robotic system illustrated in a YouTube video \footnote{\url{https://youtu.be/_z4FHuu3-ag}}, {\blue where we used two Kinova JACO robotic arms with $7$ degrees of freedom and $3$ fingers to accomplish a pick and place task.} 
The leader robot sends its state (7 joint angles and positions of 3 fingers) $X_t$ to the follower robot through a channel. One packet for updating the leader robot's state is sent periodically to the follower robot every $20$ time-slots. The transmission time of each updating packet is $20$ time-slots. The follower robot moves towards the leader's most recent state and locally controls its robotic fingers to grab an object. We constructed a robot simulation environment using the Robotics System Toolbox in MATLAB. In each episode, a can is randomly generated on a table in front of the follower robot. The leader robot observes the position of the can and illustrates to the follower robot how to grab the can and place it on another table, without colliding with other objects in the environment. {\blue The rapidly-exploring random tree (RRT) algorithm is used to control the leader robot \cite{berenson2009manipulation}. For the local control of the follower robot, an interpolation method is used to generate a trajectory between two points sent from the leader robot while also avoiding collisions with other obstacles.}
The leader robot uses a neural network to predict the follower robot's state $Y_t$. The neural network consists of one input layer, one hidden layer with $256$ ReLU activation nodes, and one fully connected (dense) output layer. The dataset contains the leader and follower robots' states in 300 episodes of continue operation. The first $80\%$ of the dataset is used for the training and the other $20\%$ of the dataset is used for the inference. 


\subsection{Actuator State Prediction}\label{actuatorexp} 
We consider the OpenAI CartPole-v1 task \cite{brockman2016openai}, where a DQN reinforcement learning algorithm \cite{mnih2015human} is used to control the force on a cart and keep the pole attached to the cart from falling over. By simulating $10^4$ episodes of the OpenAI CartPole-v1 environment, a time-series dataset is collected that contains the pole angle $\psi_t$ and the velocity $v_{t}$ of the cart. The pole angle $\psi_t$ at time $t$ is predicted based on a feature $X_{t-\delta}=(v_{t-\delta}, \ldots, v_{t-\delta-u+1})$, i.e., a vector of cart velocity with length $u$, where $v_t$ is the cart velocity at time $t$ and $\Delta(t)=\delta$ is the AoI. The predictor in this experiment is an LSTM neural network that consists of one input layer, one hidden layer with 64 LSTM cells, and a fully connected output layer. First $72\%$ of the dataset is used for training and the rest of the dataset is used for inference. 

\subsection{Temperature prediction} 
{\blue The temperature $Y_t$ at time $t$ is predicted based on a feature $X_{t-\delta}=\{s_{t-\delta}, \ldots, s_{t-\delta-u+1}\}$, where $s_t$ is a $7$-dimensional vector consisting of temperature, pressure, saturation vapor pressure, vapor pressure deficit, specific humidity, airtightness, and wind speed at time $t$. We used the Jena climate dataset recorded by the Max Planck Institute for Biogeochemistry \cite{kerasexample}. The dataset comprises $14$ features, including temperature, pressure, humidity, etc., recorded once every $10$ minutes from $10$ January $2009$ to $31$ December $2016$. The first $75\%$ of the dataset is used for training and the later $25\%$ is used for inference. Temperature is predicted every hour using an LSTM neural network composed of one input layer, one hidden layer with $32$ LSTM units, and one output layer.}


\subsection{CSI Prediction} The CSI $h_t$ at time $t$ is predicted based on a feature $X_{t-\delta}=\{h_{t-\delta}, \ldots, h_{t-\delta-u+1}\}$. The dataset for CSI is generated by using Jakes model \cite{baddour2005autoregressive}.

\section{Examples of Loss function $L$, $L$-entropy, and $L$-cross entropy}\label{InformationTheory1}
Several examples of loss function $L$, $L$-entropy, and $L$-cross entropy are listed below. Additional examples can be found in \cite{Dawid2004, Dawid1998,farnia2016minimax}.
\subsubsection{Logarithmic~Loss (log-loss)} The log-loss function is given by $L_{\mathrm{log}}(y, Q_Y) = -\log Q_Y(y)$, where the action $a=Q_Y$ is a distribution  in $\mathcal{P}^{\mathcal{Y}}$. The corresponding $L$-entropy is the well-known Shannon's entropy \cite{cover1999elements}, defined as
\begin{align}
H_{\mathrm{log}}(Y) =& -\sum_{y \in \mathcal{Y}} P_Y(y)\ \mathrm{log}\ P_Y(y), 
\end{align}
where $P_Y$ is the distribution of $Y$. The corresponding $L$-cross entropy is given by
\begin{align}
    H_{\mathrm{log}}(Y; \tilde Y)=&-\sum_{y \in \mathcal{Y}} P_{Y}(y)\ \mathrm{log}\ P_{\tilde Y}(y).
\end{align}
The $L$-mutual information and $L$-divergence associated with the log-loss are Shannon's mutual information and the K-L divergence defined in \eqref{chi-divergence-def}, respectively. 
\subsubsection{Brier~Loss} The Brier loss function is defined as $L_B(y, Q_Y)=\sum_{y' \in \mathcal{Y}} Q_Y(y')^2-2 \ Q_Y(y)+1$ \cite{Dawid2004}. The associated $L$-entropy is given by
\begin{align}
H_B(Y)=&1-\sum_{y \in \mathcal{Y}} P_Y(y)^2,
\end{align}
and the associated $L$-cross entropy is 
\begin{align}
    H_B(Y; \tilde Y)=&\sum_{y \in \mathcal Y} P_{\tilde Y}(y)^2-2 \sum_{y \in \mathcal Y}P_{\tilde Y}(y)P_{Y}(y)+1.
\end{align}
\subsubsection{0-1~Loss} The 0-1 loss function is given by $L_{\text{0-1}}(y, \hat{y}) = \mathbf{1}(y \neq \hat{y})$, where $\mathbf{1}(A)$ is the indicator function of event $A$. For this case, we have
\begin{align}
H_{\text{0-1}}(Y)~=&~ 1- \max_{y \in \mathcal{Y}} P_Y(y), \\
 H_{\text{0-1}}(Y; \tilde Y)~=&~ 1-P_{Y}\left(\arg\max_{y \in \mathcal{Y}} P_{\tilde Y}(y)\right).
\end{align}

\subsubsection{$\alpha$-Loss} The $\alpha$-loss function is defined by $L_{\alpha}(y, Q_Y)=\frac{\alpha}{\alpha-1} \left[1 - Q_Y(y)^{\frac{\alpha-1}{\alpha}}\right]$ for $\alpha>0$ and $\alpha \neq 1$  \cite[Eq. 14]{alpha}. It becomes the log-loss function in the limit $\alpha \rightarrow 1$ and the 0-1 loss function in the limit $\alpha \rightarrow \infty$. The $L$-entropy and $L$-cross entropy associated with the $\alpha$-loss function are given by 
\begin{align}
H_{\alpha \text{-loss}}(Y)= \frac{\alpha}{\alpha-1} \left[1 - \left(\sum_{y \in \mathcal{Y}} P_Y(y)^{\alpha}\right)^{\frac{1}{\alpha}}\right],
\end{align}
\begin{align}
   H_{\alpha \text{-loss}}(Y; \tilde Y)=\frac{\alpha}{\alpha-1} \left[1 - \left(\sum_{y \in \mathcal{Y}} P_{\tilde Y}(y)^{\alpha}\right)^{\frac{1}{\alpha}}\lambda\right],
\end{align}
where 
\begin{align}
    \lambda=\frac{\sum_{y \in \mathcal Y} \frac{P_{Y}(y)}{P_{\tilde Y}(y)}P_{\tilde Y}(y)^{\alpha}}{\sum_{y \in \mathcal Y} P_{\tilde Y}(y)^{\alpha}}.
\end{align}

\subsubsection{Quadratic~Loss} The quadratic loss function is $L_2(y, \hat{y})= (y - \hat{y})^2$. The $L$-entropy function associated with the quadratic loss is the variance of $Y$, given by 
\begin{align}
H_2(Y)= \mathbb{E}[Y^2] - \mathbb{E} [Y]^2.
\end{align}
The corresponding $L$-cross entropy is  
\begin{align}
     H_{2}(Y;\tilde Y)=\mathbb E[ Y^2]-2\mathbb E[\tilde Y] \mathbb E[Y]+\mathbb E[\tilde Y]^2.
\end{align}

\section{Definitions of $L$-divergence, $L$-mutual information, and $L$-conditional mutual information}\label{otherLmetrics}
The $L$-\emph{divergence} $D_L(P_{Y} || P_{\tilde Y})$ of $P_{Y}$ from $P_{\tilde Y}$ can be expressed as \cite{Dawid2004, farnia2016minimax}
\begin{align}\label{divergence}
&D_L(P_{Y} || P_{\tilde Y}) \nonumber\\
 &=\!\mathbb E_{Y \sim P_{Y}}\left[L\left(Y, a_{P_{Y}}\right)\right]-\mathbb E_{Y \sim P_{Y}}\left[L\left(Y, a_{P_{\tilde Y}}\right)\right].
\end{align}
Because $a_{P_{Y}}$ is an optimal to \eqref{eq_Lentropy}, from \eqref{divergence}, we have 
\begin{align}
D_L(P_{Y} || P_{\tilde Y}) \geq 0.
\end{align}
The \emph{$L$-mutual information} $I_L(Y;X)$ is defined as \cite{Dawid2004, farnia2016minimax}
\begin{align}\label{MI}
I_L(Y; X)=& \mathbb E_{X \sim P_{X}}\left[D_L\left(P_{Y|X}||P_{Y}\right)\right]\nonumber\\
=&H_L(Y)-H_L(Y|X) \geq 0,
\end{align}
which measures the performance gain in predicting $Y$ by observing $X$. In general, $I_L( X ; Y)  \neq I_L(Y ; X)$.  The $L$-conditional mutual information $I_L(Y; X | Z)$ is given by 
\begin{align}\label{CMI}
I_L(Y; X|Z)=& \mathbb E_{X, Z \sim P_{X, Z}}\left[D_L\left(P_{Y|X, Z}||P_{Y | Z}\right)\right]\nonumber\\
=&H_L(Y | Z)-H_L(Y|X, Z) \geq 0.
\end{align}

\section{Relationship among $L$-divergence, Bregman divergence, and $f$-divergence} \label{InformationTheory2}
We explain the relationship among the $L$-divergence defined in \eqref{divergence}, the Bregman divergence \cite{zhang2017matrix}, and the $f$-divergence \cite{csiszar2004information}. All these three classes of divergence have been widely used in the machine learning literature. Their differences are explained below. 

Let $\mathcal{P^Y}$ denote the set of all probability distributions on the discrete set $\mathcal{Y}$. Define $\mathcal Z \subset \mathbb R^{|\mathcal Y|}$ as the set of all probability vectors $\mathbf z=(z_1, \ldots, z_{|\mathcal Y|})^T$ that satisfy $\sum_{i=1}^{|\mathcal Y|}z_i = 1$ and $z_i \geq 0$ for all $i=1, 2, \ldots, |\mathcal Y|$. Any distribution $P_Y \in \mathcal{P^Y}$ can be represented by a probability vector $\mathbf{p}_Y=(P_Y(y_1), \ldots, P_Y(y_{|\mathcal Y|}))^T \in \mathcal Z$. 
\begin{definition}\cite{zhang2017matrix}
Let $F : \mathcal Z \mapsto \mathbb R$  be a continuously differentiable and strictly convex function defined on the convex set $\mathcal Z$. The Bregman divergence associated with $F$ between two distributions $P_Y, Q_Y \in \mathcal P^{\mathcal Y}$ is defined as
\begin{align}\label{Bregaman}
\!\!\! \!\!B_{F}(P_Y||Q_Y)\!=\! F(\mathbf p_Y)\!-F(\mathbf q_Y)\!-\nabla F(\mathbf q_Y)^T \!(\mathbf p_Y\!-\mathbf q_Y),
\end{align}
where $\mathbf p_Y \in \mathcal Z$ and $\mathbf q_Y \in \mathcal Z$ are two probability vectors associated with the distributions $P_Y$ and $Q_Y$, respectively, and $\nabla F=\left(\frac{\partial F}{\partial z_1}, \ldots, \frac{\partial F}{\partial z_{|\mathcal Y|}}\right)^T$. 
\end{definition}
We establish the following lemma: 
\begin{lemma}\label{B-L}
For any continuously differentiable and strictly convex function $F : \mathcal Z \mapsto \mathbb R$, the Bregman divergence $B_{F}(P_Y||Q_Y)$ associated with F is an $L_F$-divergence $D_{L_F}(P_Y || Q_Y)$ associated with the 
loss function
\begin{align}\label{eq-L_F}
L_F(y, Q_Y)= - F(\mathbf q_Y)-\frac{\partial F(\mathbf q_Y)}{\partial z_y}+\nabla F(\mathbf q_Y)^T \mathbf q_Y.
\end{align}  
\end{lemma}
\begin{proof}
According to \eqref{eq_Lentropy}, the $L_F$-entropy associated with the loss function $L_F(y, Q_Y)$ in \eqref{eq-L_F} is defined as
\begin{align}\label{entropy_F}
H_{L_F}(Y)= \min_{Q_Y \in \mathcal{P^Y}} E_{Y \sim P_Y} \left [ L_F\left (Y, Q_Y\right)\right],
\end{align}
where $P_Y$ is the distribution of $Y$. Using \eqref{eq-L_F}, we can get
\begin{align}\label{expected_loss_F}
&E_{Y \sim P_Y} \left [ L_F\left (Y, Q_Y\right)\right] \nonumber\\
&=-\sum_{y \in \mathcal Y} P_Y(y) F(\mathbf q_Y)-\sum_{y \in \mathcal Y} P_Y(y) \frac{\partial F(\mathbf q_Y)}{\partial z_y} \nonumber\\
&~~+\sum_{y \in \mathcal Y} P_Y(y) \nabla F(\mathbf q_Y)^T \mathbf q_Y \nonumber\\
&=-F(\mathbf q_Y)-\nabla F(\mathbf q_Y)^T (\mathbf p_Y-\mathbf q_Y),
\end{align}
where the last equality holds because $F(\mathbf q_Y)$ and $\nabla F(\mathbf q_Y)^T \mathbf q_Y$ are constants that remain unchanged regardless of the variable $y$.
Because the function $F$ is continuously differentiable and strictly convex on the convex set $\mathcal Z$, we have for all $\mathbf p_Y, \mathbf q_Y \in \mathcal Z $
\begin{align}\label{lowerboundachieved}
-F(\mathbf q_Y)-\nabla F(\mathbf q_Y)^T (\mathbf p_Y-\mathbf q_Y) \geq -F(\mathbf p_Y).
\end{align}
Equality holds in \eqref{lowerboundachieved} if and only if $\mathbf q_Y=\mathbf p_Y$. 
From \eqref{entropy_F}, \eqref{expected_loss_F}, and \eqref{lowerboundachieved}, we obtain
\begin{align}\label{entropyLf}
H_{L_F}(Y)&=\min_{\mathbf q_Y \in \mathcal Z} -F(\mathbf q_Y)-\nabla F(\mathbf q_Y)^T (\mathbf p_Y-\mathbf q_Y) \nonumber\\
&=-F(\mathbf p_Y).
\end{align}
Substituting \eqref{expected_loss_F} and \eqref{entropyLf} into \eqref{Bregaman}, yields 
\begin{align}
&B_{F}(P_Y || Q_Y) \nonumber\\
=&F(\mathbf p_Y)-F(\mathbf q_Y)-\nabla F(\mathbf q_Y)^T (\mathbf p_Y-\mathbf q_Y)\nonumber\\
=&E_{Y \sim P_Y} \left [ L_F\left (Y, Q_Y\right)\right]-H_{L_F}(Y) \nonumber\\
=&D_{L_F}(P_Y || Q_Y).
\end{align}
This completes the proof.
\end{proof}
 
Let us rewrite the $L$-entropy $H_L(Y )$ as $H_L(\mathbf p_Y )$ to emphasize
that it is a function of the probability vector $\mathbf p_Y$. If $H_{L}(\mathbf p_Y)$ is continuously differentiable and strictly concave in $\mathbf p_Y$, then the $L$-divergence $D_L(P_Y||Q_Y)$ can be expressed as \cite[Section 3.5.4]{Dawid2004}
\begin{align}\label{L-divergencetoBregman}
    &D_L(P_Y||Q_Y) \nonumber\\
    =&H_L(\mathbf q_Y)+\nabla H_L(\mathbf q_Y)^T (\mathbf p_Y-\mathbf q_Y)-H_L(\mathbf p_Y) \nonumber\\
    =& B_{-H_L}(P_Y||Q_Y),
\end{align}
which is the Bregman divergence $B_{-H_L}(P_Y||Q_Y)$ associated with the continuously differentiable and strictly convex $-H_L$. However, if $H_{L}(\mathbf p_Y)$ is not continuously differentiable or not strictly concave in $\mathbf p_Y$, $D_L(P_Y||Q_Y)$ is not necessarily a Bregman divergence.

\begin{definition} \cite{csiszar2004information}
Let $f: (0, \infty) \mapsto \mathbb R$ be a convex function with
$f(1) = 0$. The $f$-divergence between two probability distributions $P_Y, Q_Y \in \mathcal{P^Y}$ is defined as 
\begin{align}
   D_f(P_Y||Q_Y)=\sum_{y \in \mathcal Y} Q_Y(y) f\left(\frac{P_Y(y)}{Q_Y(y)}\right).
\end{align}
\end{definition}
An $f$-divergence may not be $L$-divergence, and vice versa. In fact, the KL divergence $D_{\mathrm{log}}(P_{Y}||Q_Y)$ defined in \eqref{chi-divergence-def} and its dual $D_{\mathrm{log}}(Q_{Y}||P_Y)$ are the unique divergences belonging to both the classes of $f$-divergence and Bregman divergence \cite{Amari}. Because KL divergence is also an $L$-divergence, $D_{\mathrm{log}}(P_{Y}||Q_Y)$ and $D_{\mathrm{log}}(Q_{Y}||P_Y)$ are the only divergences belonging to all the three classes of divergences. 

The $f$-mutual information can be expressed using the $f$-divergence as
\begin{align}
    I_f(Y; X)=\mathbb E_{X \sim P_X}[D_f(P_{Y|X}||P_Y)].
\end{align}
{\blue The $f$-mutual information is symmetric, i.e., $I_f(Y; X)=I_f(X;Y),$ which can be shown as follows:
\begin{align}
I_f(Y; X)=&\sum_{x \in \mathcal X} P_X(x) \sum_{y \in \mathcal Y} P_Y(y) f\left(\frac{P_{Y|X}(y|x)}{P_Y(y)}\right)\nonumber\\
=&\sum_{\substack{x \in \mathcal X \\ y \in \mathcal Y}} P_X(x) P_Y(y) f\left(\frac{P_{Y|X}(y|x)P_X(x)}{P_Y(y)P_X(x)}\right)\nonumber\\
=&\sum_{\substack{x \in \mathcal X \\ y \in \mathcal Y}}P_X(x) P_Y(y) f\left(\frac{P_{X|Y}(x|y)P_Y(y)}{P_Y(y)P_X(x)}\right)\nonumber\\
=&\sum_{y \in \mathcal Y} P_Y(y) \sum_{x \in \mathcal X} P_X(x) f\left(\frac{P_{X|Y}(x|y)}{P_X(x)}\right)\nonumber\\
=&I_f(X; Y).
\end{align}}
On the other hand, the $L$-mutual information is generally non-symmetric, i.e., $I_L(Y ;X)\neq I_L(X; Y )$, except for some special cases.  For example, Shannon's mutual information is defined by 
\begin{align}
 I_{\mathrm{log}}(Y; X)=\mathbb E_{X \sim P_X}[D_{\mathrm{log}}(P_{Y|X}||P_Y)],
\end{align}
which is both an $L$-mutual information and a $f$-mutual information. It is well-known that $I_{\log} (Y;X) = I_{\log} (X;Y)$.

\section{Proof of Equation \eqref{instantaneous_err1}}\label{pinferenceerror}
Because we assume that $Y_t$ and $X_{t-\delta}$ are independent of $\Delta(t)$, for all $y \in \mathcal{Y}$, $x \in \mathcal{X}$, and $\delta \in \mathbb{Z}^{+}$, we have
\begin{align}\label{probability1}
P_{Y_t, X_{t-\delta} | \Delta(t)=\delta}(y, x) = P_{Y_t, X_{t-\delta}}(y, x).
\end{align}
By using $\phi^*$ on the inference dataset, the inference error given $\Delta(t)=\delta$ is determined by
\begin{align}\label{e11-2}
p(\delta)= \mathbb{E}_{Y, X \sim P_{Y_t, X_{t-\Delta(t)} | \Delta(t)=\delta}} \left[ L\left(Y, \phi^*(X, \delta)\right)\right],
\end{align}
where $P_{Y_t, X_{t-\Delta(t)} | \Delta(t)=\delta}$ is the distribution of target $Y_t$ and feature $X_{t-\Delta(t)}$ given $\Delta(t)=\delta$.

By substituting \eqref{probability1} into \eqref{e11-2}, we obtain
\begin{align}
p(\delta)\nonumber&= \mathbb{E}_{Y, X \sim P_{Y_t, X_{t-\Delta(t)} | \Delta(t)=\delta}} \left[ L\left(Y, \phi^*(X, \delta)\right)\right]\nonumber\\
&= \mathbb{E}_{Y, X \sim P_{Y_t, X_{t-\delta} | \Delta(t)=\delta}} \left[ L\left(Y, \phi^*(X, \delta)\right)\right]\nonumber\\
&= \mathbb{E}_{Y, X \sim P_{Y_t, X_{t-\delta}}} \left[ L\left(Y, \phi^*(X, \delta)\right)\right]\nonumber\\
&= \mathbb{E}_{Y, X \sim P_{Y_{\delta}, X_{0}}} \left[ L\left(Y, \phi^*(X, \delta)\right)\right].
\end{align}
The last equality holds due to the stationarity of $\{(Y_t, X_t), t =0, 1, 2, \ldots\}$. This completes the proof. $\qed$

\section{Proof of Equation \eqref{freshness_aware_cond}}\label{pfreshness_aware_cond}
Let $\mathcal D=\{\delta: P_{\Theta}(\delta)>0\}$ be support set of $P_{\Theta}$. From \eqref{eq_TrainingErrorLB1}, we have
\begin{align}\label{L_condEntropy1}
&H_L(\tilde Y_0| \tilde X_{-\Theta},\Theta)\nonumber\\
=&\sum_{\substack{x \in \mathcal X, \delta \in \mathcal D}} \!\!\!\! P_{\tilde X_{-\Theta}, \Theta}(x, \delta) H_L(\tilde Y_0|\tilde X_{-\Theta}=x, \Theta=\delta)\nonumber\\
=& \!\sum_{\delta \in \mathcal D} P_{\Theta}(\delta) \sum_{x \in \mathcal X} P_{\tilde X_{-\Theta}|\Theta=\delta}(x) H_L(\tilde Y_0|\tilde X_{-\Theta}=x, \Theta=\delta)\!\!\! \nonumber\\
=&\sum_{\delta \in \mathcal D} P_{\Theta}(\delta) \sum_{x \in \mathcal X} P_{\tilde X_{-\delta}|\Theta=\delta}(x) H_L(\tilde Y_0|\tilde X_{-\Theta}=x, \Theta=\delta).\!\!\!\!
\end{align}

Next, from \eqref{given_L_condentropy}, we obtain that for all $x \in \mathcal X$ and $\delta \in \mathcal D$,
\begin{align}\label{given_L_condentropy2}
&H_L(\tilde Y_0| \tilde X_{-\Theta}=x,\Theta=\delta) \nonumber\\
=& \min_{a\in\mathcal A}  \mathbb E_{Y\sim P_{\tilde Y_0|\tilde X_{-\Theta}=x, \Theta=\delta}} [L(Y, a)] \nonumber\\
=& \min_{a\in\mathcal A} \mathbb E_{Y\sim P_{\tilde Y_0|\tilde X_{-\delta}=x, \Theta=\delta}} [L(Y, a)].
\end{align}
Because we assume that $\tilde Y_0$ and $\tilde X_{-k}$ are independent of $\Theta$ for every $k \geq 0$, for all $x \in \mathcal X, y \in \mathcal Y$ and $\delta \in \mathcal D$
\begin{align}\label{ind_XY}
P_{\tilde Y_0|\tilde X_{-\delta}=x, \Theta=\delta}(y)&=P_{\tilde Y_0| \tilde X_{-\delta}=x}(y), \\\label{ind_X}
P_{\tilde X_{-\delta}|\Theta=\delta}(x)&=P_{\tilde X_{-\delta}}(x).
\end{align}
Substituting \eqref{ind_XY} into \eqref{given_L_condentropy2}, we get
\begin{align}\label{given_L_condentropy1}
&H_L(\tilde Y_0| \tilde X_{-\Theta}=x,\Theta=\delta)\nonumber\\
=& \min_{a\in\mathcal A} \mathbb E_{Y\sim P_{\tilde Y_0|\tilde X_{-\delta}=x}} [L(Y, a)]\nonumber\\
=&H_L(\tilde Y_0| \tilde X_{-\delta}=x).
\end{align}
Substituting \eqref{given_L_condentropy1} and \eqref{ind_X} into \eqref{L_condEntropy1}, we observe that 
\begin{align}
H_L(\tilde Y_0| \tilde X_{-\Theta},\Theta)=&\sum_{\delta \in \mathcal D} P_{\Theta}(\delta) \sum_{x \in \mathcal X} P_{\tilde X_{-\delta}}(x) H_L(\tilde Y_0|\tilde X_{-\delta}=x)\nonumber\\
=&\sum_{\delta \in \mathbb Z^{+}} P_{\Theta}(\delta)~H_L(\tilde Y_0| \tilde X_{-\delta}).
\end{align}

This completes the proof.


\ignore{\section{Proof of Equation \eqref{Decomposed_Cross_entropy}}\label{pDecomposed_Cross_entropy}
Because $Y_t$ and $X_{t-\delta}$ are independent of $\Theta$, from \eqref{L_condEntropy} we get
\begin{align}
\hat \phi_{P_{Y_t, X_{t-\Theta},\Theta}}(x,\delta)=&\arg\min_{\phi(x, \delta)\in\mathcal A}\!\! E_{Y\sim P_{Y_t|X_{t-\Theta}=x, \Theta=\delta}}[L(Y,\phi(x, \delta))] \nonumber\\
=&\arg\min_{a\in\mathcal A}\!\!~E_{Y\sim P_{Y_t|X_{t-\Theta}=x, \Theta=\delta}}[L(Y,a)] \nonumber\\
=&\arg\min_{a\in\mathcal A}\!\!~E_{Y\sim P_{Y_t|X_{t-\delta}=x}}[L(Y,a)].
\end{align}
From the above equation, we can say that $\hat \phi_{P_{Y_t, X_{t-\Theta},\Theta}}(x,\delta)$ is a Bayes action $a_{P_{Y_t|X_{t-\delta}=x}}$. 

Next, since $\tilde Y_t$ and $\tilde X_{t-\delta}$ are independent of $\Delta$, we obtain 
\begin{align}
&H_L(\tilde Y_{t}; Y_{t} | \tilde X_{t-\Delta}, \Delta)\nonumber\\
=&\sum_{y \in \mathcal Y, x \in \mathcal X, \delta \in \mathcal D} P_{\tilde Y_t, \tilde X_{t-\Delta}, \Delta}(y, x, \delta)~L\left(y, a_{P_{Y_t|X_{t-\delta}=x}}\right)\nonumber\\
=&\sum_{\delta \in \mathcal D} P_{\Delta}(\delta) \sum_{\substack{y \in \mathcal Y, x \in \mathcal X}}P_{\tilde Y_t, \tilde X_{t-\Delta}| \Delta=\delta}(y, x)~L\left(y, a_{P_{Y_t|X_{t-\delta}=x}}\right)\nonumber\\
=&\sum_{\delta \in \mathcal D} P_{\Delta}(\delta)\sum_{\substack{y \in \mathcal Y, x \in \mathcal X}}P_{\tilde Y_t, \tilde X_{t-\delta}}(y, x)~L\left(y, a_{P_{Y_t|X_{t-\delta}=x}}\right)\nonumber\\
=&\sum_{\delta \in \mathcal D} P_{\Delta}(\delta)~H_L(\tilde Y_{t}; Y_{t} | \tilde X_{t-\delta}).
\end{align}}
\section{Proof of Lemma \ref{Symmetric}}\label{PSymmetric}
This is due to the following symmetry property: 
\begin{align}
I_{\mathrm{log}}(Y;Z|X)=I_{\mathrm{log}}(Z;Y|X).
\end{align}

\section{Proof of Lemma \ref{Lemma_CMI}}\label{PLemma_CMI}
By using the definition of $L$-conditional mutual information in \eqref{CMI}, we obtain
\begin{align}\label{cond_entropy_two}
    H_L(Y|X, Z)=&H_L(Y|X)-I_L(Y;Z|X) \nonumber\\
    =&H_L(Y|Z)-I_L(Y;X|Z).
\end{align}
From \eqref{CMI} and \eqref{cond_entropy_two}, we get
\begin{align}
    H_L(Y|X)=&H_L(Y|Z)+I_L(Y;Z|X)-I_L(Y;X|Z) \nonumber\\
    \leq& H_L(Y|Z)+I_L(Y;Z|X),
\end{align}
where the last inequality is due to $I_L(Y;X|Z) \geq 0$. Now, we need to show that if $Y \overset{\epsilon}\leftrightarrow X \overset{\epsilon}\leftrightarrow  Z$, then
\begin{align}
  I_L(Y;Z|X)=O(\epsilon).
\end{align}
and in addition, if $H_L(Y)$ is twice differentiable, then
\begin{align}
  I_L(Y;Z|X)=O(\epsilon^2).
\end{align}

From \eqref{CMI}, we see that
\begin{align}\label{CMI1}
I_L(Y ; Z | X) = \mathbb{E}_{X, Z} [D_L(P_{Y|X,Z} || P_{Y|X})].
\end{align}
We know from Pinsker's inequality \cite[Lemma 11.6.1]{cover1999elements} that
\begin{align}\label{inequalityinformationtheory} 
\sum_{y \in \mathcal Y} (P_Y(y)-Q_Y(y))^2 \leq 2 \mathrm{ln}2~D_{\mathrm{log}}(P_Y||Q_Y).
\end{align}
If $Y \overset{\epsilon}\leftrightarrow X \overset{\epsilon}\leftrightarrow  Z$ is an $\epsilon$-Markov chain, then Definition \ref{epsilonMarkovChain} yields:
\begin{align}\label{eMarkovcondition1}
  \sum_{(x,z) \in \mathcal X \times \mathcal Z} P_{X,Z}(x,z) D_{\mathrm{log}}(P_{Y|X=x,Z=z} || P_{Y|X=x}) \leq \epsilon^2.
\end{align}
Let $\mathcal X' \times \mathcal Z'=\{ (x, z): P_{X, Z}(x, z)>0\}$ be the support set of $P_{X, Z}$. Then, \eqref{eMarkovcondition} reduces to 
\begin{align}\label{eMarkovcondition}
  \sum_{(x,z) \in \mathcal X' \times \mathcal Z'} P_{X,Z}(x,z) D_{\mathrm{log}}(P_{Y|X=x,Z=z} || P_{Y|X=x}) \leq \epsilon^2.
\end{align}
Because the left side of the above inequality is the summation of non-negative terms, the following holds:
\begin{align}\label{inequalityepsilon}
  P_{X,Z}(x,z) D_{\mathrm{log}}(P_{Y|X=x,Z=z} || P_{Y|X=x}) \leq \epsilon^2,
\end{align}
for all $(x, z) \in \mathcal X' \times \mathcal Z'$. 
Because $P_{X,Z}(x,z)>0$ for all $(x, z) \in \mathcal X' \times \mathcal Z'$, from \eqref{inequalityinformationtheory} and \eqref{inequalityepsilon}, we can write
\begin{align}\label{Oepsilon}
  \sum_{y \in \mathcal Y} (P_{Y|X=x, Z=z}(y)-P_{Y|X=x}(y))^2 \leq \dfrac{2 \text{ln}2 \epsilon^2}{P_{X,Z}(x,z)},
\end{align}
for all $(x, z) \in \mathcal X' \times \mathcal Z'$.
Next, we need the following lemma.
\begin{lemma}\label{divergenceL}
The following assertions are true:
\begin{itemize}
\item[(a)] If two distributions $Q_Y\in \mathcal{P^Y}$ and $P_Y \in \mathcal{P^Y}$ satisfy 
\begin{align}\label{condition_divergenceL}
\sum_{y \in \mathcal Y} (P_Y(y)-Q_Y(y))^2 \leq \beta^2,
\end{align}
then 
\begin{align}
D_L(P_Y || Q_Y)=O(\beta).
\end{align}
\item[(b)] If, in addition, $H_L(Y)$ is twice differentiable in $P_Y$, then
\begin{align}
D_L(P_Y || Q_Y)=O(\beta^2).
\end{align}
\end{itemize}
\end{lemma}
\begin{proof}
See in Appendix \ref{PdivergenceL}.
\end{proof}

Using \eqref{Oepsilon} and Lemma \ref{divergenceL}(a) in \eqref{CMI1}, we obtain
\begin{align}\label{lemma1result1}
   &I_L(Y;Z|X) \nonumber\\
   =&\sum_{(x,z) \in \mathcal X \times \mathcal Z} P_{X,Z}(x,z)~D_{L}(P_{Y|X=x,Z=z} || P_{Y|X=x}) \nonumber\\
   =&\sum_{(x,z) \in \mathcal X' \times \mathcal Z'} P_{X,Z}(x,z)~D_{L}(P_{Y|X=x,Z=z} || P_{Y|X=x}) \nonumber\\
   =&\sum_{(x,z) \in \mathcal X' \times \mathcal Z'} P_{X, Z}(x,z)~O\left(\frac{\sqrt{2\text{ln}2}\epsilon}{\sqrt{P_{X, Z}(x,z)}}\right) \nonumber\\
   =&~O(\epsilon).
 \end{align}
Similarly, when $H_L(Y)$ is differentiable in $P_Y$, by using Lemma \ref{divergenceL}(b) we obtain 
\begin{align}\label{lemma1result2}
    I_L(Y;Z|X)=O(\epsilon^2).
\end{align}
This completes the proof of Lemma \ref{Lemma_CMI}.

\section{Proof of Theorem \ref{theorem1}}\label{Ptheorem1}
By using the definition of $L$-conditional mutual information in \eqref{CMI}, we can show that
\begin{align}
&H_L(\tilde Y_0 | \tilde X_{-k}, \tilde X_{-k-1}) \nonumber\\
=&H_L(\tilde Y_0 | \tilde X_{-k-1})-I_L( \tilde Y_0; \tilde X_{-k} | \tilde X_{-k-1}) \nonumber\\
=&H_L(\tilde Y_0 | \tilde X_{-k})-I_L( \tilde Y_0; \tilde X_{-k-1} | \tilde X_{-k}).
\end{align}
Expanding $H_L(\tilde Y_0 | \tilde X_{-k})$, we have
\begin{align}
H_L(\tilde Y_0 | \tilde X_{-k})=& H_L(\tilde Y_0 | \tilde X_{-k-1})+ I_L(\tilde Y_0; \tilde X_{-k-1} | \tilde X_{-k})\nonumber\\
&-I_L(\tilde Y_0; \tilde X_{-k} | \tilde X_{-k-1}).
\end{align}
Since the above equation is valid for all values of $k \geq 0$, taking the summation of $H_L(\tilde Y_0 | \tilde X_{-k})$ from $k=0$ to $\delta-1$ yields:
\begin{align} \label{PMarkovC1}
H_L(\tilde Y_0 | \tilde X_{-\delta})=& H_L(\tilde Y_0 | \tilde X_0) + \sum_{k=0}^{\delta-1} I_L(\tilde Y_0; \tilde X_{-k} | \tilde X_{-k-1})\nonumber\\
&- \sum_{k=0}^{\delta-1} I_L(\tilde Y_0; \tilde X_{-k-1} | \tilde X_{-k}).
\end{align}
Thus, we can express $H_L(\tilde Y_0 | \tilde X_{-\delta})$ as a function of $\delta$ as in \eqref{eMarkov} and \eqref{g12function}.
Moreover, the functions $g_1(\delta)$ and $g_2(\delta)$ defined in \eqref{g12function} are non-decreasing in $\delta$ as $I_L(\tilde Y_0; \tilde X_{-k} | \tilde X_{-k-1}) \geq 0$ and $I_L(\tilde Y_0; \tilde X_{-k-1} | \tilde X_{-k}) \geq 0$ for all values of $k$. 

To prove the next part, we use Lemma \ref{Lemma_CMI}. Because for every $\mu, \nu \geq 0$, $\tilde Y_{0} \overset{\epsilon} \leftrightarrow \tilde X_{-\mu} \overset{\epsilon} \leftrightarrow \tilde X_{-\mu-\nu}$ is an $\epsilon$-Markov chain, we can write
\begin{align}
I_L(\tilde Y_0; \tilde X_{-k-1} | \tilde X_{-k})=O(\epsilon).
\end{align}
This implies
\begin{align}
g_2(\delta)=\sum_{k=0}^{\delta-1} O(\epsilon)= O(\epsilon).
\end{align}
The last equality follows from the summation property of big-O-notation. This completes the proof.

\section{Proof of Theorem \ref{theorem2}}\label{Ptheorem2}
Using \eqref{freshness_aware_cond} and Theorem \ref{theorem1}, we obtain
\begin{align}\label{Th2_1}
\!\!\!\!\!\!&H_L(\tilde Y_0 | \tilde X_{-\Theta}, \Theta) \nonumber\\
=& \sum_{\delta \in \mathbb Z^{+}} \!P_{\Theta}(\delta) \big(H_L(\tilde Y_0 | \tilde X_0)+\!\hat g_1(\delta)-g_2(\delta)\big)\nonumber\\
=&H_L(\tilde Y_0 | \tilde X_0)\!+\mathbb E_{\Theta \sim P_{\Theta}}\! [\hat g_1(\Theta)]\!-\mathbb E_{\Theta \sim P_{\Theta}}\! [g_2(\Theta)],
\end{align}
where 
\begin{align}
\hat g_1(\delta)&=g_1(\delta)-H_L(\tilde Y_0 | \tilde X_0)\nonumber\\
&=\sum_{k=0}^{\delta-1} I_L(\tilde Y_0; \tilde X_{-k} | \tilde X_{-k-1}). 
\end{align}
Because mutual information $I_L(\tilde Y_0; \tilde X_{-k} | \tilde X_{-k-1})$ is non-negative, we have
\begin{align} 
\hat g_1(\delta)=\sum_{k=0}^{\delta-1} I_L(\tilde Y_0; \tilde X_{-k} | \tilde X_{-k-1})\geq 0.
\end{align}
Because $\hat g_1(\delta)$ is non-negative for all $\delta$, the function $\hat g_1(\cdot)$ is Lebesgue integrable with respect to all probability measure $P_\Theta$ \cite{durrett2019probability}. Hence, the expectation $\mathbb E_{\Theta \sim P_{\Theta}} [\hat g_1(\Theta)]$ exists. Note that $E_{\Theta \sim P_{\Theta}} [\hat g_1(\Theta)]$ can be infinite ($+\infty$). By using the same argument, we obtain that $E_{\Theta \sim P_{\Theta}} [g_2(\Theta)]$ exists but can also be infinite. Moreover, the functions $\hat g_1(\delta)$ and $g_2(\delta)$ are non-decreasing in $\delta$.

Because (i) the function $\hat g_1(\delta)$ is non-decreasing in $\delta$, (ii) the expectation $\mathbb E_{\Theta \sim P_{\Theta}} [\hat g_1(\Theta)]$ exists, and (iii) $\Theta_1 \leq_{st} \Theta_2$, we have \cite{stochasticOrder}
\begin{align}\label{Sorder}
\mathbb E_{\Theta \sim P_{\Theta_1}} [\hat g_1(\Theta)] \leq \mathbb E_{\Theta \sim P_{\Theta_2}} [\hat g_1(\Theta)].
\end{align}

Next, we obtain:
\begin{align}
&H_L(\tilde Y_0 | \tilde X_{-\Theta_1}, \Theta_1) \nonumber\\
\overset{(a)}{=}& H_L(\tilde Y_0 | \tilde X_0)+\mathbb E_{\Theta \sim P_{\Theta_1}} [\hat g_1(\Theta)]-\mathbb E_{\Theta \sim P_{\Theta_1}} [g_2(\Theta)] \nonumber\\
\overset{(b)}{\leq}& H_L(\tilde Y_0 | \tilde X_0)+\mathbb E_{\Theta \sim P_{\Theta_2}} [\hat g_1(\Theta)]-\mathbb E_{\Theta \sim P_{\Theta_1}} [g_2(\Theta)]\nonumber\\
\overset{(c)}{=}&H_L(\tilde Y_0 | \tilde X_{-\Theta_2}, \Theta_2) \nonumber\\
&+\mathbb E_{\Theta \sim P_{\Theta_2}} [g_2(\Theta)] -\mathbb E_{\Theta \sim P_{\Theta_ 1}} [g_2(\Theta)] \nonumber\\
\overset{(d)}{=}&H_L(\tilde Y_0 | \tilde X_{-\Theta_2}, \Theta_2)+O(\epsilon),
\end{align} 
where (a) and (c) hold due to \eqref{Th2_1}, (b) is obtained using \eqref{Sorder}, and (d) follows from the fact that $\tilde Y_0 \overset{\epsilon}\leftrightarrow \tilde X_{-\mu} \overset{\epsilon}\leftrightarrow \tilde X_{-\mu-\nu}$ is an $\epsilon$-Markov chain for all $\mu, \nu \geq 0$ (see Theorem \ref{theorem1}(b)). This completes the proof. 

\section{Proof of Lemma \ref{lemma_inference}}\label{Plemma_inference}
By using condition \eqref{T3condition2} and Lemma \ref{divergenceL}(a), we obtain for all $x \in \mathcal X$:
\begin{align}\label{E1}
D_L\left(P_{Y_{t}| X_{t-\delta}=x} || P_{\tilde Y_{0}| \tilde X_{-\delta}=x}\right)=O(\beta).
\end{align}
Next, by using \eqref{L-CondCrossEntropy} and \eqref{E1}, we have
\begin{align}
&H_L(P_{Y_{t}|X_{t-\delta}}; P_{\tilde Y_{0}|X_{-\delta}} | P_{X_{t-\delta}})\nonumber\\
=& H_L(Y_{t} | X_{t-\delta}) \nonumber\\
&+\sum_{x \in \mathcal X} P_{X_{t-\delta}}(x)~D_L\left(P_{Y_{t}| X_{t-\delta}=x} || P_{\tilde Y_{0}| \tilde X_{-\delta}=x}\right)\nonumber\\
=& H_L(Y_{t} | X_{t-\delta})+O(\beta).
\end{align}  
This completes the proof.
\section{Proof of Theorem \ref{theorem3}}\label{Ptheorem3}
Part (a): 
By the definition of $L$-conditional cross entropy \eqref{L-CondCrossEntropy}, we get 
\begin{align}\label{T3_1}
&H_L(P_{Y_{t}|X_{t-\delta}}; P_{\tilde Y_{0}|X_{-\delta}} | P_{X_{t-\delta}}) \nonumber\\
=&\sum_{x \in \mathcal X} \!\!P_{X_{t-\delta}}(x)\mathbb{E}_{Y \sim P_{Y_t| X_{t- \delta}=x}} \left[ \!L\left(Y,a_{\tilde Y_0|\tilde X_{-\delta}=x}\right)\right] ,\!\!
\end{align} 
where the Bayes predictor $a_{\tilde Y_0|\tilde X_{-\delta}=x}$ is fixed in the inference phase for every time slot $t$. Because $\{(Y_t, X_t),t=0, 1, 2, \ldots\}$ is a stationary process, \eqref{T3_1} is a function of the AoI $\delta$.

Part (b): We can apply Lemma \ref{lemma_inference} since \eqref{T3condition2} holds for all $x \in \mathcal X$ and $\delta \in \mathcal Z^{+}$. This gives us:
\begin{align}
&H_L(P_{Y_{t}|X_{t-\delta_1}}; P_{\tilde Y_{0}|X_{-\delta_1}} | P_{X_{t-\delta_1}})\nonumber\\
=& H_L(Y_{t} | X_{t-\delta_1})+O(\beta) \nonumber\\
\leq& H_L( Y_{t} | X_{t-\delta_2})+O(\epsilon)+O(\beta) \nonumber\\
=&H_L(P_{Y_{t}|X_{t-\delta_2}}; P_{\tilde Y_{0}|X_{-\delta_2}} | P_{X_{t-\delta_2}})+O(\beta)+O(\epsilon)+O(\beta) \nonumber\\
=&H_L(P_{Y_{t}|X_{t-\delta_2}}; P_{\tilde Y_{0}|X_{-\delta_2}} | P_{X_{t-\delta_2}})+O(\max\{\epsilon, \beta\}),
\end{align}
where we use Lemma \ref{lemma_inference} to obtain the first and the third equality, and the second inequality holds due to the assumption that $Y_t \overset{\epsilon}\leftrightarrow X_{t-\mu} \overset{\epsilon}\leftrightarrow X_{t-\mu-\nu}$ is an $\epsilon$-Markov chain for all $\mu, \nu \geq 0$ (see Theorem \ref{theorem1}). This completes the proof.

\ignore{\section{Proof of Theorem \ref{theorem4}}\label{Ptheorem4}
By using \eqref{Decomposed_Cross_entropy}, Lemma \ref{lemma_inference}, and Theorem \ref{theorem2}, we obtain
\begin{align}
H_L(\tilde Y_{t}; Y_{t} | \tilde X_{t-\Delta_1}, \Delta_1)=&\sum_{\delta \in \mathcal D} P_{\Delta_1}(\delta)H_L(\tilde Y_{t}; Y_{t} | \tilde X_{t-\delta}) \nonumber\\
=& \sum_{\delta \in \mathcal D} P_{\Delta_1}(\delta)~H_L(\tilde Y_{t} | \tilde X_{t-\delta})+O(\beta) \nonumber\\
=& H_L(\tilde Y_{t} | \tilde X_{t-\Delta_1}, \Delta_1)+O(\beta) \nonumber\\
\leq& H_L(\tilde Y_{t} | \tilde X_{t-\Delta_2}, \Delta_2)+O(\epsilon)+O(\beta) \nonumber\\
=&H_L(\tilde Y_{t}; Y_{t} | \tilde X_{t-\Delta_2}, \Delta_2)+O(\beta)+O(\epsilon)+O(\beta) \nonumber\\
=&H_L(\tilde Y_{t}; Y_{t} | \tilde X_{t-\Delta_2}, \Delta_2)+O(\max\{\epsilon, \beta\}).
\end{align}}
\ignore{\section{Simplification of the Gittins Index in \eqref{gittins} }\label{GittinsDerivation}
For the bandit process $\{\Delta(t) : t \geq 0\}$ in \eqref{Bandit}, define the $\sigma$-field 
\begin{align}\label{sigma-field}
\mathcal F_s^t= \sigma(\Delta(t+k): k \in \{0, 1, \ldots, s\}),
\end{align} which is the set of events whose occurrence are determined by the realization of the process $\{\Delta(t+k) :  k \in \{0, 1, \ldots, s\}\}$ from time slot $t$ up to time slot $t+s$. Then, $\{ \mathcal F_k^t, k \in \{0, 1, \ldots\}\}$ is the filtration of the time shifted process $\{\Delta(t+k) : k \in \{0, 1, \ldots\} \}$. We define $\mathfrak M$ as the set of all stopping times by 
\begin{align}
\mathfrak M=\{ \nu \geq 0 : \{\nu=k\} \in \mathcal F_k^t, k \in \{0, 1, 2, \ldots\}\}.
\end{align}

The Gittins index $\gamma(\delta)$ \cite{gittins2011multi} is the value of reward $r$ for which the \textsc{stop} and \textsc{continue} actions are equally profitable at state $\Delta(t)=\delta$. Hence, $\gamma(\delta)$ is the root of the following equation of $r$:  
\begin{align}\label{gittins_subtract}
&\sup_{\nu \in \mathfrak M, \nu \neq 0} \mathbb E\left[ \sum_{k=0}^{\nu+T_1-1} [r-p(\Delta(t+k))]\bigg| \Delta(t)=\delta \right]\nonumber\\
&~~~~~~=\mathbb E\left[ \sum_{k=0}^{T_1-1} [r-p(\Delta(t+k))]\bigg| \Delta(t)=\delta \right].
\end{align}
where the left hand side of \eqref{gittins_subtract} is the maximum total expected profit under \textsc{continue} action and the right hand side of \eqref{gittins_subtract} is the total expected profit under \textsc{stop} action. By re-arranging \eqref{gittins_subtract}, it can be expressed as
\begin{align}\label{gittins_subtract1}
&\sup_{\nu \in \mathfrak M, \nu \neq 0} \mathbb E\left[ \sum_{k=0}^{\nu-1} [r-p(\Delta(t+k+T_1))]\bigg| \Delta(t)=\delta \right]=0.
\end{align}
Because the left hand side of \eqref{gittins_subtract1} is the supremum of strictly increasing and linear functions of $r$, it is convex, continuous, and strictly increasing in $r$. Thus, the fixed-point equation \eqref{gittins_subtract1} has a unique root. The root can also be expressed as
\begin{align}\label{gittins1}
&\gamma(\delta) \nonumber\\
=&\bigg\{r:\!\!\!\!\sup_{\nu \in \mathfrak M, \nu \neq 0}\!\!\!\! \mathbb E\left[ \sum_{k=0}^{\nu-1} [r-p(\Delta(t+k+T_1))]\bigg| \Delta(t)\!=\!\delta \right]\!\!=\!0\bigg\}.
\end{align}
Let $\nu^* \in \mathfrak M$ be the optimal stopping time that solves \eqref{gittins1}. Because of \eqref{Bandit} and $T_1 \geq 1$, for any $k \leq \nu^*$, $\Delta(t+k)=\Delta(t)+k$. Hence, $\{\Delta(t+k): 1 \leq k \leq \nu^*\}$ is completely determined by the initial value $\Delta(t)$ and for all $k \leq \nu^*$, the $\sigma$-field $\mathcal F_k^t$ can be simplified as $\mathcal F_k^t=\sigma(\Delta(t))$. Thus, any stopping time in $\mathfrak M$ is a deterministic time, given $\Delta(t)=\delta$. By this, \eqref{gittins1} can be simplified as
\begin{align}\label{gittins2}
\!\!&\gamma(\delta)\nonumber\\
=&\bigg\{r:\!\!\!\!\! \sup_{\tau \in \{1, 2, \ldots\}}\!\!\!\!\! \mathbb E\left[ \sum_{k=0}^{\tau-1} [r-p(\Delta(t+k+T_1))]\bigg| \Delta(t)\!=\!\delta \right]\!\!=\!0\bigg\}\nonumber\\
=&\bigg\{r:\!\!\!\!\! \inf_{\tau \in \{1, 2, \ldots\}} \mathbb E\left[ \sum_{k=0}^{\tau-1} [p(\Delta(t+k+T_1))-r]\bigg| \Delta(t)\!=\!\delta \right]\!\!=\!0\bigg\}\nonumber\\
=&\!\bigg\{r: \!\!\!\!\! \inf_{\tau \in \{1, 2, \ldots\}} \sum_{k=0}^{\tau-1}\!\mathbb E\left[p(\Delta(t+k+T_1))-r\bigg| \!\Delta(t)\!=\delta \right]\!\!=\!0\bigg\},
\end{align}
where $\tau$ is a deterministic positive integer.

Define 
\begin{align}
f(r)=\inf_{\tau \in \{1, 2, \ldots\}} \sum_{k=0}^{\tau-1}\mathbb E\left[p(\delta+k+T_1)-r\right].
\end{align}
Similar to Lemma 7 in \cite{orneeTON2021}, the following lemma holds.
\begin{lemma}\label{fraction_programming}
$f(r) \lesseqgtr 0$ if and only if
\begin{align}
\inf_{\tau \in \{1, 2, \ldots\}} \frac{1}{\tau} \sum_{k=0}^{\tau-1} \mathbb E \left [p(\delta+k+T_{1}) \right] \lesseqgtr r.
\end{align}
\end{lemma}
By \eqref{gittins1}, \eqref{gittins2}, and Lemma \ref{fraction_programming}, the root of equation $f(r) = 0$ is given by \eqref{gittins}. This completes the proof.}

\section{Proofs of Theorem \ref{theorem5}, Theorem \ref{theorem6}, Theorem \ref{decoupled}, and Theorem \ref{decoupled1}}\label{MainResult}
In this section, we prove Theorem \ref{theorem5}, Theorem \ref{theorem6}, Theorem \ref{decoupled}, and Theorem \ref{decoupled1}. These theorems provide optimal solutions for the scheduling problems \eqref{scheduling_problem}, \eqref{sub_scheduling_problem}, \eqref{decoupled_problem}, and \eqref{sub_decoupled_problem}.
We begin by deriving the optimal solution for \eqref{decoupled_problem}. Subsequently, the optimal solutions for \eqref{scheduling_problem}, \eqref{sub_scheduling_problem}, and \eqref{sub_decoupled_problem} follow directly, as these problems are special cases of \eqref{decoupled_problem}.

Since the problem \eqref{decoupled_problem} focuses solely on a single source, we simplify the notation by omitting the source index $m$ and rewrite the problem \eqref{decoupled_problem} as follows:
\begin{align}\label{decoupled_problem1}
\!\!\bar p_{opt}(\lambda)\!=\!\inf_{\pi \in \Pi} \limsup_{T\rightarrow \infty} \mathbb{E}\!\left[ \frac{1}{T} \sum_{t=0}^{T-1} w~p(\Delta(t)) + \lambda c(t) \right],
\end{align}
where $p(\Delta(t))$ is the penalty at time $t$, $\Delta(t)\in \mathbb Z^{+}$ is the AoI, $c(t) \in \{0, 1\}$ is the channel occupation status at time $t$, $\pi=((S_1, b_1), (S_2, b_2), \ldots)$ is a scheduling policy, $\Pi$ is the set of all causal and signal-agnostic scheduling policies, $w>0$ is a weight, and $\bar p_{opt}(\lambda)$ is the optimal objective value to \eqref{decoupled_problem1}.

The scheduling problem \eqref{decoupled_problem1} is an infinite-horizon average-cost semi-Markov decision process (SMDP) \cite[Chapter 5.6]{bertsekasdynamic1}. We provide a detailed description of the components of this problem:

\begin{itemize}
    \item{\bf Decision Time:} Each $i$-th feature delivery time $D_i=S_i+T_i$ is a decision time of the problem \eqref{decoupled_problem1}, where $S_i$ is the scheduling time of the $i$-th feature and the $i$-th feature takes $T_i\geq 1$ time slots to be delivered. 
     \item {\bf State:} At time slot $D_i$, the state of the system is represented by AoI $\Delta(D_i)$. 
     \item {\bf Action:} Let $\tau_{i+1}=S_{i+1}-D_i$ represent the waiting time for sending the $(i + 1)$-th feature after the $i$-th feature is delivered. As we consider $S_0=0$ and $S_i=\sum_{j=1}^{i} (T_{j-1}+\tau_{j})$ for each $i=1, 2, \ldots$. Given $(T_0, T_1, \ldots)$, the sequence $(S_1, S_2, \ldots)$ is uniquely determined by $(\tau_1, \tau_2, \ldots)$. Hence, one can also use $(\tau_1, \tau_2, \ldots)$ to represent a sequence of actions instead of $(S_1, S_2, \ldots)$. 

     At time $D_i$, the scheduler decides the waiting time $\tau_{i+1}$ and the buffer position $b_{i+1}$.

     \item {\bf State Transitions:} The AoI process $\Delta(t)$ evolves as 
\begin{align}\label{AoIProcess}
\!\! \Delta(t)=
  \begin{cases}
        T_{i}+b_i, &\text{if}~t=D_{i}, i=0, 1, \ldots, \\
        \Delta(t-1)+1, &\text{otherwise}.
    \end{cases}
    \end{align} 
\item {\bf Expected Transition Time:} The expected time difference between two decision times, $D_i$ and $D_{i+1}$, is given by
\begin{align}\label{interdeliverytime}
\mathbb E[D_{i+1}-D_i] &=\mathbb E[S_{i+1}+T_{i+1}-(S_i+T_i)]\nonumber\\&=\mathbb E[S_{i}+T_i+\tau_{i+1}+T_{i+1}-S_i-T_i]\nonumber\\&=\mathbb E[\tau_{i+1}+T_{i+1}].
\end{align}

    \item {\bf Expected Transition Cost:} The expected cumulative cost incurred during the interval between two decision times, $D_i$ and $D_{i+1}$, can be calculated as
\begin{align}\label{interdeliverycost}
&\mathbb E\left[\sum_{t=D_i}^{D_{i+1}-1}\bigg(w~p(\Delta(t))+\lambda c(t)\bigg)\right]\nonumber\\
=&\mathbb E\left[\sum_{k=0}^{\tau_{i+1}+T_{i+1}-1} w~ p(\Delta(D_i+k))\right]\!+\!\lambda \mathbb E[T_{i+1}].
\end{align}
\end{itemize}
The infinite-horizon average-cost SMDP \eqref{decoupled_problem1} can be solved by using dynamic programming \cite{bertsekasdynamic1, puterman2014markov}. Let $h: \mathbb Z^{+} \mapsto \mathbb R$ be the relative value function associated with the average-cost SMDP \eqref{decoupled_problem1}. At time $t=D_i$, the optimal action $(\tau_{i+1}, b_{i+1})$ can be determined by solving the following Bellman optimality equation \cite[P. 275]{bertsekasdynamic1}: 
\begin{align}\label{Bellman1}
&h(\Delta(D_i))\nonumber\\
&=\!\!\inf_{\overset{\tau_{i+1} \in \{0, 1, \ldots\}}{b_{i+1} \in \{0, \ldots, B-1\}}}\!\! \mathbb E \left [ \sum_{k=0}^{\tau_{i+1}+T_{i+1}-1} w p(\Delta(D_i+k)) \right]+\lambda \mathbb E[T_{i+1}]\nonumber\\
&\quad \quad \quad \quad \quad-\bar p_{opt}(\lambda)\mathbb E[\tau_{i+1}+T_{i+1}]+\mathbb E[h(\Delta(D_{i+1})]\nonumber\\
&=\!\!\inf_{\overset{\tau_{i+1} \in \{0, 1, \ldots\}}{b_{i+1} \in \{0, \ldots, B-1\}}}\!\! \mathbb E \left [ \sum_{k=0}^{\tau_{i+1}+T_{i+1}-1} \bigg(w p(\Delta(D_i+k))\!-\!\bar p_{opt}(\lambda)\bigg) \right]\nonumber\\
&\quad \quad \quad \quad \quad+\lambda \mathbb E[T_{i+1}]+\mathbb E[h(\Delta(D_{i+1})]\nonumber\\
&=\!\!\inf_{\overset{\tau_{i+1} \in \{0, 1, \ldots\}}{b_{i+1} \in \{0, \ldots, B-1\}}}\!\! \mathbb E \left [ \sum_{k=0}^{\tau_{i+1}+T_{i+1}-1} \bigg(w p(\Delta(D_i+k))\!-\!\bar p_{opt}(\lambda)\bigg) \right]\nonumber\\
&\quad \quad \quad \quad \quad+\lambda \mathbb E[T_{i+1}]+\mathbb E[h(T_{i+1}+b_{i+1})],
\end{align}
where the last equality holds because $\Delta(D_{i+1})=T_{i+1}+b_{i+1}$. 

From \eqref{Bellman1}, it is observed that the buffer position $b_{i+1}$ only depends on the term $\mathbb E[h(T_{i+1}+b_{i+1})]$, while the waiting time $\tau_{i+1}$ has no impact on $\mathbb E[h(T_{i+1}+b_{i+1})]$. Hence, the optimal buffer position $b^*_{i+1}$ is determined by
\begin{align}\label{optimization2}
b_{i+1}^*=&\arg\min_{b_{i+1} \in \{0, 1, \ldots, B-1\}} \mathbb E[h(T_{i+1}+b_{i+1})].
\end{align}
Since $T_i$'s are i.i.d., $\mathbb E[h(T_{i+1}+b)]=\mathbb E[h(T_{i}+b)]=\cdots=\mathbb E[h(T_{1}+b)]$ for all $i$ and $b$. Hence, from \eqref{optimization2}, it is evident that there exists a $b^*\in \{0, 1, \ldots, B-1\}$ such that $b^*_1=b^*_2=\cdots=b^*_{i+1}=b^*$ that satisfies
\begin{align}
b^*=&\arg\min_{b \in \{0, 1, \ldots, B-1\}} \mathbb E[h(T_{1}+b)].
\end{align}
Because the optimal buffer position is time-invariant, the problem \eqref{decoupled_problem1} can be expressed as
\begin{align}\label{rduced_scheduling_problem}
\bar p_{opt}(\lambda)
=\min_{b\in \{0, 1,\ldots, B-1\}} \bar p_{b, opt} (\lambda),
\end{align}
where $\bar p_{b, opt} (\lambda)$ is given by
\begin{align}\label{scheduling_problem_b}
\bar p_{b, opt}(\lambda)=\inf_{\pi_b\in\Pi_b}  \limsup_{T\rightarrow \infty}\frac{1}{T} \mathbb{E}_{\pi_b}\left [ \sum_{t=0}^{T-1} w~p(\Delta(t))+\lambda c(t)\right],
\end{align}
$\pi_b=((S_1, b), (S_2, b), \ldots)$, $\Pi_b$ is the set of all causal and signal-agnostic scheduling policies $\pi_b$ with fixed buffer position $b$, and $\bar p_{b, opt}(\lambda)$ is the optimal objective value to \eqref{scheduling_problem_b}.

At every $i$-th decision time $D_i$ of the average-cost SMDP \eqref{scheduling_problem_b}, the scheduler decides the waiting time $\tau_{i+1}$. The average-cost SMDP \eqref{scheduling_problem_b} can be solved by using dynamic programming \cite{bertsekasdynamic1, puterman2014markov}. Given AoI value $\delta$ at decision time $D_i$, the Bellman optimality equation of \eqref{scheduling_problem_b} is obtained by substituting $\Delta(D_i)=\delta$, $b_{i+1}=b$, and $\bar p_{opt}(\lambda)=\bar p_{b, opt}(\lambda)$ into \eqref{Bellman1}, given by
\begin{align}\label{Bellman2}
h_b(\delta)=&\inf_{\tau \in \{0, 1, 2, \ldots\}}\mathbb E \left [ \sum_{k=0}^{\tau+T_{i+1}-1} (w~p(\delta+k) -\bar p_{b, opt}(\lambda)) \right]\nonumber\\
&+\lambda\mathbb E[T_{i+1}]+\mathbb E[h_b(T_{i+1}+b)], ~\delta = 1, 2,\ldots \nonumber\\
=&\inf_{\tau \in \{0, 1, 2, \ldots\}}\mathbb E \left [ \sum_{k=0}^{\tau+T_{1}-1} (w~p(\delta+k) -\bar p_{b, opt}(\lambda)) \right]\nonumber\\
&+\lambda\mathbb E[T_{1}]+\mathbb E[h_b(T_{1}+b)], ~\delta = 1, 2,\ldots, 
\end{align}
where the last equality holds because $T_i$'s are identically distributed. Let $\tau(\delta, \bar p_{b, opt}(\lambda))$ be an optimal solution to \eqref{Bellman2}. If $\Delta(D_i)=\delta$, then an optimal waiting time $\tau_{i+1}$ of \eqref{scheduling_problem_b} for sending the $(i+1)$-th feature is $\tau(\delta, \bar p_{b, opt}(\lambda))$. 

From \eqref{Bellman2}, we can show that $\tau(\delta, \bar p_{b, opt}(\lambda))=0$ if 
\begin{align}\label{inequality1}
&\inf_{\tau \in \{1, 2, \ldots\}}\mathbb E \left [ \sum_{k=0}^{\tau+T_{1}-1} (w~p(\delta+k) -\bar p_{b, opt}(\lambda)) \right] \nonumber\\
&\geq \mathbb E \left [ \sum_{k=0}^{T_{1}-1} (w~p(\delta+k) -\bar p_{b, opt}(\lambda)) \right].
\end{align}
After some rearrangement, the inequality \eqref{inequality1} can also be expressed as 
\begin{align}\label{inequality2}
\inf_{\tau \in \{1, 2, \ldots\}}\mathbb E \left [ \sum_{k=0}^{\tau-1} (w~p(\delta+k+T_{1}) -\bar p_{b, opt}(\lambda)) \right] \geq 0.
\end{align}
Next, similar to \cite[Lemma 7]{orneeTON2021}, the following lemma holds.
\begin{lemma}\label{fraction_programming}
The inequality \eqref{inequality2} holds if and only if
\begin{align}\label{inequality3}
\inf_{\tau \in \{1, 2, \ldots\}} \frac{1}{\tau} \mathbb{E} \left[ \sum_{k=0}^{\tau-1} w~p(\delta+k+T_{1}) \right] \geq \bar p_{b, opt}(\lambda).
\end{align}
\end{lemma}
According to \eqref{gittins}, the left-hand side of \eqref{inequality3} is equal to $\gamma(\delta)$.
Similarly, $\tau(\delta, \bar p_{b, opt}(\lambda))=1$, if $\tau(\delta, \bar p_{b, opt}(\lambda))\neq 0$ and 
\begin{align}\label{inequality4}
\inf_{\tau \in \{2, 3, \ldots\}}\mathbb E \left [ \sum_{k=0}^{\tau-1} (w~p(\delta+k+T_{1}) -\bar p_{b, opt}(\lambda)) \right] \geq 0
\end{align}
By using Lemma \ref{fraction_programming} and \eqref{gittins}, we can show that the inequality \eqref{inequality4} can be rewritten as 
\begin{align}
\gamma(\delta+1)\geq \bar p_{b, opt}(\lambda).
\end{align}
By repeating this process, we get $\tau(\delta, \bar p_{b, opt}(\lambda))=k$ is optimal, if $\tau(\delta, \bar p_{b, opt}(\lambda))\neq 0, 1, \ldots, k-1$ and 
\begin{align}
\gamma(\delta+k)\geq\bar p_{b, opt}(\lambda).
\end{align}
Hence, the optimal waiting time $\tau_{i+1}=\tau(\delta, \bar p_b(\lambda))$ is determined by
\begin{align}\label{optimalwaitingb}
\tau(\delta, \bar p_b(\lambda))=\min_{k \in \mathbb Z}\{k \geq 0: \gamma(\delta+k)\geq  \bar p_{b, opt}(\lambda)\}.
\end{align}

Now, we are ready to compute the optimal objective value $\bar p_{b, opt}(\lambda)$. Using \eqref{Bellman2}, we can determine the value of $\mathbb{E}[h_b(T_1+b)]$, which is given by
\begin{align}\label{relativeValue4}
&\mathbb{E}[h_b(T_1+b)]\nonumber\\
&= \mathbb{E} \left[ \sum_{k=0}^{\tau(T_1+b, \bar p_{b, opt}(\lambda))+T_{1}-1} (wp(T_1+b+k) -  \bar p_{b, opt}(\lambda)) \right] \nonumber \\
&+\lambda \mathbb E[T_1]+ \mathbb{E}[h_b(T_{1}+b)],
\end{align}
which yields
\begin{align}\label{optimalfinding}
&\mathbb{E} \left[ \sum_{k=0}^{\tau(T_1+b, \bar p_{b, opt}(\lambda))+T_{1}-1} \bigg(w p(T_1+b+k) -  \bar p_{b, opt}(\lambda)\bigg) \right]\nonumber\\
&+\lambda \mathbb E[T_1] = 0.
\end{align}
Rearranging \eqref{optimalfinding}, we get
\begin{align}\label{pb}
 &\bar p_{b, opt}(\lambda) \nonumber\\
 &= \frac{\mathbb{E} \left[ \sum_{k=0}^{\tau(T_1+b,  \bar p_{b, opt}(\lambda))+T_{1}-1} wp(T_1+b+k) \right]+\lambda \mathbb E[T_1]}{\mathbb{E}[\tau(T_1+b,  \bar p_{b, opt}(\lambda))+T_{1}]}\nonumber\\
&=\frac{\mathbb{E}\left[\sum_{t=D_i( \bar p_{b, opt}(\lambda))}^{D_{i+1}( \bar p_{b, opt}(\lambda))-1}  wp\big(\Delta(t)\big)\right]+\lambda E[T_1]}{\mathbb{E}\left[D_{i+1}( \bar p_{b, opt}(\lambda))-D_i( \bar p_{b, opt}(\lambda))\right]},
\end{align}
where $D_{i+1}( \bar p_{b, opt}(\lambda))=S_{i+1}( \bar p_{b, opt}(\lambda))+T_{i+1}$ and \begin{align}
&S_{i+1}(\bar p_{b, opt}(\lambda))\nonumber\\
&=\min_{t\geq 0}\{t \geq D_i(\bar p_{b, opt}(\lambda)):\gamma(\Delta(t))\geq \bar p_{b, opt}(\lambda)\},
\end{align}and the last equality holds due to \eqref{interdeliverytime} and \eqref{interdeliverycost}.

Now, by combining \eqref{rduced_scheduling_problem}, \eqref{optimalwaitingb}, and \eqref{pb} and by substituting appropriate values of $\lambda$ and $w$, we obtain optimal solution to \eqref{sub_scheduling_problem}, \eqref{scheduling_problem}, \eqref{decoupled_problem}, and \eqref{sub_decoupled_problem}.

Finally, we need to prove that
\begin{align}\label{JfunctionW}
&\mathbb{E}\left[\sum_{t=D_i(\beta)}^{D_{i+1}(\beta)-1}  wp\big(\Delta(t)\big)\right]\nonumber\\
&-\beta \mathbb{E}\left[D_{i+1}(\beta)-D_i(\beta)\right]+\lambda \mathbb E[T_1]=0.
\end{align}
has a unique root. We define 
\begin{align}\label{Jfunction}
J(\beta)&=\mathbb{E}\left[\sum_{t=D_i(\beta)}^{D_{i+1}(\beta)-1}  wp\big(\Delta(t)\big)\right]\nonumber\\
&-\beta \mathbb{E}\left[D_{i+1}(\beta)-D_i(\beta)\right]+\lambda \mathbb E[T_1].
\end{align}
\begin{lemma}\label{jB}
The function $J(\beta)$ has the following properties:
\begin{itemize}
\item[(i)] The function $J(\beta)$ is concave, continuous, and strictly decreasing in $\beta$.
\item[(ii)] 
$\lim_{\beta \to \infty} j(\beta) = -\infty$ and 
$\lim_{\beta \to -\infty} j(\beta) = \infty.$ 
\end{itemize}
\end{lemma}
\begin{proof}
See Appendix \ref{pjB}
\end{proof}



Lemma \ref{jB} proves the uniqueness of \eqref{Jfunction1}.
Also, the uniqueness of the root of \eqref{bisection} and \eqref{bisectionM} follows immediately from Lemma \ref{jB}.

\ignore{ \section{Proof of Theorem \ref{theorem6}}\label{Ptheorem6}
We first show that our problem \eqref{scheduling_problem} exhibits an important property: there exists an optimal solution $(f^*, g^*)$ to \eqref{scheduling_problem}, where $f^* = (b^*, b^*, \ldots)$, i.e., under the optimal feature selection policy $f^*$, all features are transmitted from a fixed buffer position $b^* \in \{0, 1, \ldots, B-1\}$.

Let $Z_i = S_{i+1} - D_i \geq 0$ represent the waiting time between the delivery time $D_i$ of the $i$-th feature and the sending time $S_{i+1}$ of the $(i+1)$-th feature. Then, $S_i=\sum_{j=0}^{i-1} (T_j+Z_j)$ and $D_i=\sum_{j=0}^{i-1}(T_j+Z_j)$ for each $i=1, 2, \ldots$. Given $(T_0, T_1, \ldots)$, $(S_1, S_2, \ldots)$ is uniquely determined by $(Z_0, Z_1, \ldots)$. Hence, one can also use $g=(Z_0, Z_1, \ldots)$ to represent a sending time policy.

The infinite time-horizon average cost problem \eqref{scheduling_problem} can be represented as an average cost semi-Markov decision process (SMDP), where at time slot $D_i$, a scheduler decides the waiting the $Z_i$ and the buffer position $b_{i+1}$. At time $D_i$, the state is represented by the AoI value $\Delta(D_i)$. The AoI process $\Delta(t)$ evolves as:
\begin{align}\label{AoIProcess}
\!\! \Delta(t)=
  \begin{cases}
        T_{i}+b_i, &\text{if}~t=D_{i}, i=0, 1, \ldots, \\
        \Delta(t-1)+1, &\text{otherwise}.
    \end{cases}
    \end{align}
The expected total cost within the interval between two decision times, $D_i$ and $D_{i+1}$, is given by
\begin{align}
\mathbb E\left[\sum_{t=D_i}^{D_{i+1}-1}p(\Delta(t))\right]=\mathbb E\left[\sum_{k=0}^{Z_{i}+T_{i+1}-1} p(\Delta(D_i+k))\right].
\end{align}
The expected time different between two decision times, $D_i$ and $D_{i+1}$, is given by
\begin{align}\label{interdeliverytime}
\mathbb E[D_{i+1}-D_i] =\mathbb E[Z_i+T_{i+1}].
\end{align}
The SMDP \eqref{scheduling_problem} can be solved by using dynamic programming \cite{bertsekasdynamic1}. Let $h: \mathbb Z^{+} \mapsto \mathbb R$ be the relative value function of the average cost SMDP. If $\Delta(D_i)=\delta$, then we can find the optimal waiting time $Z_i$ and optimal buffer position $b_{i+1}$ using the following Bellman equation \cite[Section 5.6]{bertsekasdynamic1}:
\begin{align}\label{relativeValueT6}
h(\delta)=\inf_{\overset{Z_i \in \{0, 1, \ldots\}} {b_{i+1} \in \{0, 1, \ldots, B-1\}}}&\bigg\{ \mathbb E \left [ \sum_{k=0}^{Z_i+T_{i+1}-1} p(\Delta(D_i+k))\right]\nonumber\\
&-\mathbb E[Z_i+T_{i+1}]~\bar p_{opt}\nonumber\\
&+\mathbb E[h(T_{i+1}+b_{i+1})]\bigg\}.
\end{align}
From \eqref{relativeValueT6}, we can observe that $b_{i+1}$ depends on only $\mathbb E[h(T_{i+1}+b_{i+1})]$.
Therefore, from \eqref{relativeValueT6}, we can deduce that the optimal buffer position $b^*_{i+1}$ satisfies
\begin{align}\label{optimization2}
b^*_{i+1}=\argmin_{b \{0, 1, \ldots, B-1\}} \mathbb E[h(T_{i+1}+b)].
\end{align}
Since $T_i$'s are i.i.d., $\mathbb E[h(T_{1}+b)]=\mathbb E[h(T_{2}+b)]=\cdots=\mathbb E[h(T_{1}+b)]$ for all $i$. Hence, from \eqref{optimization2}, it is evident that there exists a $b^*\in \{0, 1, \ldots, B-1\}$ such that $b^*_1=b^*_2=\cdots=b^*_{i+1}=b^*$. Therefore, the problem \eqref{scheduling_problem} becomes
\begin{align}\label{rduced_scheduling_problem}
\bar p_{opt}&=\min_{b\in \{0, 1,\ldots, B-1\}}\inf_{(f_b, g)\in\Pi}  \limsup_{T\rightarrow \infty}\frac{1}{T} \mathbb{E}_{(f_b, g)}\left [ \sum_{t=0}^{T-1} p(\Delta(t))\right]\nonumber\\
&=\min_{b\in \{0, 1,\ldots, B-1\}} \bar p_b,
\end{align}
where $\bar p_b$ is the optimal objective value of \eqref{sub_scheduling_problem}. The optimal solution to \eqref{rduced_scheduling_problem} yields the buffer position $b^*$.

Given a fixed feature selection policy $f_{b^*} = (b^*, b^*, \ldots)$, the optimal sending times $g= (S_1, S_2, \ldots)$ can be determined from Theorem \ref{theorem5}. By substituting $\bar p_b$ with $\bar p_{opt}$ or $\bar p_{b^*}$ in \eqref{OptimalPolicy_Sub} of Theorem \ref{theorem5}, we obtain \eqref{Optimal_Scheduler}.}

\ignore{Theorem \ref{theorem6} consists of part (a) and part (b). If part (a) holds, then part (b) is directly obtained using Theorem \ref{theorem5}. Hence, It is enough to prove part (a).

The problem \eqref{scheduling_problem} can be cast as an MDP problem. The State and the penalty of the MDP are the same as the MDP discussed in Appendix \ref{Ptheorem5}. The action space is different: If the channel is idle, the scheduler sends $(b+1)$-th freshest feature or does not send any feature. Let $d(t)=\textsc{idle}$ means the scheduler does not send feature at time $t$ and $d(t)=b$ means the scheduler sends $(b+1)$-th freshest feature at time $t$. Then, $S_{i+1}$ and $b_{i+1}$ are determined by 
\begin{align}\label{sendingtimecal}
S_{i+1}&=\inf_{t \in \mathbb Z} \{t \geq S_i+T_i: d(t)\neq \textsc{idle}\}, \\
b_{i+1}&=d(S_{i+1}).
\end{align}

If the channel is idle and $\Delta(t)=\delta$, then the optimal decision $d(t)$ in time slot $t$ satisfies the following Bellman optimality equation 
\begin{align}\label{Bellman_optimality1}
V(\delta)=\min_{d(t) \in \{\textsc{idle}, 0, 1, \ldots, B-1\}} Q(\delta, d(t)),
\end{align}
where the function $V$ is the relative value function and $Q(\delta, d(t))$ is given by
\begin{align}
&Q(\delta, b) \nonumber\\
&= \mathbb E \left [ \sum_{k=0}^{T_{i+1}-1} [p(\delta+k)-\bar p_{\text{opt}}]\right]+\mathbb E[V(T_{i+1}+b)], \\
&Q(\delta, \textsc{idle}) \nonumber\\
&=\inf_{\tau \in \{1, 2, \ldots\}} \mathbb E \left [ \sum_{k=0}^{\tau-1} [p(\delta+k+T_{i+1})-\bar p_{\text{opt}}]\right] \nonumber\\
&+\mathbb E \left [ \sum_{k=0}^{T_{i+1}-1} [p(\delta+k)-\bar p_{\text{opt}}]\right]+\min_{b \in \mathcal B} \mathbb E[V(T_{i+1}+b)].
\end{align}
where $\mathcal B=\{0, 1, \ldots, B-1\}$ and $\bar p_{\text{opt}}$ is the optimal value of \eqref{scheduling_problem}.

By \eqref{Bellman_optimality1}, $d(t)= \textsc{idle}$ is not an optimal choice if
\begin{align}\label{proofinequality}
Q(\delta, \textsc{idle})-\min_{b \in \mathcal B} Q(\delta, b) \geq 0.
\end{align}
By using similar steps \eqref{Bellman-inequality}-\eqref{eq_fractional1}, we can get that the inequality \eqref{proofinequality} holds if and only if 
\begin{align}\label{proofinequality1}
\gamma(\delta) \geq \bar p_{\text{opt}}.
\end{align}
Then, by using \eqref{sendingtimecal}, \eqref{proofinequality}, and \eqref{proofinequality1}, we get the optimal sending time $S^*_{i+1}$ in \eqref{Optimal_Scheduler}. 

Next, we need to get the optimal $b^*_{i+1}$ and $\bar p_{\text{opt}}$. When $\Delta(S^*_{i+1})=\delta$, $b^*_{i+1}=b^*$ is optimal if
\begin{align}
b^*=&\arg\min_{b \in \mathcal B} Q(\delta, b) \nonumber\\
=&\arg\min_{b \in \mathcal B}\mathbb E \left [ \sum_{k=0}^{T_{i+1}-1} [p(\delta+k)-\bar p_{\text{opt}}]\right]+\mathbb E[V(T_{i+1}+b)] \nonumber\\
=&\arg\min_{b \in \mathcal B}\mathbb E[V(T_{i+1}+b)].
\end{align}
Observe that the optimal $b^*$ is independent of the state $\Delta(S_{i+1})=\delta$. Moreover, because $T_i$ is identically distributed, 
\begin{align}
\mathbb E[V(T_{0}+b)]=\mathbb E[V(T_{1}+b)]=\ldots= \mathbb E[V(T_{i}+b)], \forall i.
\end{align}
Thus, the optimal $b^*$ is constant for all $i$. 
If $b_{i+1}=b$ for all $i$, then $\bar p_b$ is the average inference error. Hence, the optimal choice $b^*$ satisfies 
\begin{align}
b^*=\arg\min_{b \in \mathcal B} \bar p_b,
\end{align}
where $\bar p_b$ is $\beta_b$, which is the root of \eqref{bisection}. The optimal value is 
\begin{align}
\bar p_{\text{opt}}=\min_{b \in \mathcal B} \bar p_b.
\end{align}}

\section{Proof of Theorem \ref{theorem7}}\label{ptheorem7}
\ignore{Similar to Section \ref{Ptheorem5}, the problem \eqref{decoupled_problem} can be cast as an MDP and we solve the problem by using dynamic programming. Then, we prove the indexability of the arm $(l, b_l)$.

If the channel is idle and $\Delta_l(t)=\delta$, the optimal decision $d_{l, b_l}(t)$ for the problem \eqref{decoupled_problem} in time slot $t$ satisfies the following Bellman optimality equation 
\begin{align}\label{Belmandecoupled}
V_{l, b_l}(\delta)=\min_{d_{l, b_l}(t) \in \{0, 1\}} Q_{l, b_l}(\delta, d_{l, b_l} (t)),
\end{align}
where the function $V_{l, b_l}$ is the relative value function and $ Q_{l, b_l}(\delta, d_{l, b_l} (t))$ is given by
\begin{align}
&Q_{l, b_l}(\delta, 1)\nonumber\\
&= \mathbb E \left [ \sum_{k=0}^{T_{l, i+1}-1} [w_l p_l(\delta+k)-\bar p_{l, b_l}(\lambda)]\right]\! \\
&+\!\mathbb E[V_{l, b_l}(T_{l, i+1}+b_l)]\!+\!\lambda, \nonumber\\
&Q_{l, b_l}(\delta, 0)\nonumber\\
&=\inf_{\tau \in \{1, 2, \ldots\}} \mathbb E \left [ \sum_{k=0}^{\tau-1} [w_l p_l(\delta+k+T_{l, i+1})-\bar p_{l, b_l}(\lambda)]\right] \nonumber\\
&+\mathbb E \left [ \sum_{k=0}^{T_{l, i+1}-1} [w_l p_l(\delta+k)-\bar p_{l, b_l}(\lambda)]\right]\!\nonumber\\
&+\! \mathbb E[V_{l, b_l}(T_{l, i+1}+b_l)]\!+\! \lambda.
\end{align}
where $\bar p_{l, b_l}(\lambda)$ is optimal objective value of the problem \eqref{decoupled_problem}. 

Similar to the proof of \eqref{eq_fractional1} and \eqref{getoptimalvalue} in Section \ref{Ptheorem5}, by solving \eqref{Belmandecoupled}, $d_{l, b_l}(t)=0$ is optimal if 
\begin{align}\label{solutiondecoupled}
w_l~\gamma_l(\delta) \leq \bar p_{l, b_l}(\lambda), 
\end{align} 
otherwise $d_{l, b_l}(t)=1$ is optimal, where $\bar p_{l, b_l}(\lambda)$ is given by
\begin{align}\label{getoptimalvaluearm}
\bar p_{l, b_l}(\lambda)=\frac{C^l_{i, i+1}(\lambda)}{N^l_{i, i+1}(\lambda)},
\end{align}
where $C^l_{i, i+1}(\lambda)$ is the expected penalty of source $l$ starting from $i$-th delivery time to $(i+1)$-th delivery time and $N^l_{i, i+1}(\lambda)$ is the expected number of time slots from $i$-th delivery time to $(i+1)$-th delivery time, given by
\begin{align}\label{C}
\!\!C^l_{i, i+1}(\lambda)\!\!=\!\mathbb{E}\left[\sum_{t=D_{l, i}(\bar p_{l, b_l}(\lambda))}^{D_{l, i+1}(\bar p_{l, b_l}(\lambda))-1}  \!\!\!\!\!w_l p_l\big(\Delta_l(t)\big)\right]+\lambda \mathbb E[T_{l, i+1}],
\end{align}
\begin{align}\label{N}
N^l_{i, i+1}(\lambda)=\mathbb{E}\left[D_{l, i+1}(\bar p_{l, b_l}(\lambda))-D_{l, i}(\bar p_{l, b_l}(\lambda))\right],
\end{align}
the $(i+1)$-th feature delivery time $D_{l, i+1}(\bar p_{l, b_l}(\lambda))$ from source $l$ is  
\begin{align}\label{deliverytimearm}
D_{l, i+1}(\bar p_{l, b_l}(\lambda))=S_{l, i+1}(\bar p_{l, b_l}(\lambda))+T_{l, i+1}
\end{align}
and the $(i+1)$-th sending time $S_{l, i+1}(\bar p_{l, b_l}(\lambda))$ is 
\begin{align}\label{sendingtimearm}
S_{l, i+1}(\bar p_{l, b_l}(\lambda))=\inf_{t \in \mathbb Z}\{t \geq D_{l, i}: w_l~\gamma_l(\delta) \leq \bar p_{l, b_l}(\lambda)\}.
\end{align}
The sending time $S_{l, i+1}(\bar p_{l, b_l}(\lambda))$ can also be expressed as 
\begin{align}\label{revisedsendingtime}
S_{l, i+1}(\bar p_{l, b_l}(\lambda))=D_{l, i}(\bar p_{l, b_l}(\lambda))+z(T_{l, i}, b_l, \bar p_{l, b_l}(\lambda)),
\end{align}
the waiting time $z(T_{l, i}, b_l, \bar p_{l, b_l}(\lambda))$ after the delivery time $D_{l, i}(\bar p_{l, b_l}(\lambda))$ is 
\begin{align}\label{waitingformulation}
&z(T_{l, i}, b_l, \bar p_{l, b_l}(\lambda))\nonumber\\
=&\inf_{z \in \mathbb Z}\{ z \geq 0: w_l~\gamma_l(\Delta_l(D_{l, i}(\bar p_{l, b_l}(\lambda)))+z) \geq \bar p_{l, b_l}(\lambda)\} \nonumber\\
=&\inf_{z \in \mathbb Z}\{ z \geq 0: w_l~\gamma_l(T_{l, i}+b_l+z) \geq \bar p_{l, b_l}(\lambda)\},
\end{align}
where the last equality holds because from \eqref{multi-source-Age}, we get
\begin{align}\label{MultiAoIProcess}
\Delta_l(t)=
  \begin{cases}
        T_{l, i}+b_l, &\text{if}~t=D_{l, i}(\bar p_{l, b_l}(\lambda)), i=0, 1, \ldots, \\
        \Delta_l(t-1)+1, &\text{otherwise}.
    \end{cases}
    \end{align}
By using \eqref{deliverytimearm}-\eqref{MultiAoIProcess}, \eqref{C} and \eqref{N} reduce to 
\begin{align}
C^l_{i, i+1}(\lambda)=&\mathbb{E}\left[\sum_{k=T_{l, i}}^{T_{l, i}+z(T_{l, i}, b_l, \bar p_{l, b_l}(\lambda))+T_{l, i+1}-1}  \!\!\!\!\!\!w_l~p_l(k+b_l)\right]\nonumber\\
&+\lambda \mathbb E[T_{l, i+1}],
\end{align}
\begin{align}
N^l_{i, i+1}(\lambda)=\mathbb{E}\left[ z(T_{l, i}, b_l, \bar p_{l, b_l}(\lambda))+T_{l, i+1}\right].
\end{align}}
According to \eqref{Optimal_Scheduler} and the definition of set $\Omega_m(\lambda)$, a point $(\Delta_m(t), d_m(t)) \in \Omega_m(\lambda)$ if either (i) $d_m(t)>0$ such that a feature from source $m$ is currently in service at time t, or (ii) $\gamma_m(\Delta_m(t)) < \bar p_{m, opt}(\lambda)$ such that the threshold condition in \eqref{Optimal_Scheduler} for sending a new feature is not satisfied. By this, an analytical expression of set $\Omega_m(\lambda)$ is derived as
\begin{align}\label{passive-set}
\Omega_m(\lambda)=\{(\delta, d): d>0 ~\mathrm{or}~ \gamma_m(\delta) < \bar p_{m, opt}(\lambda)\},
\end{align}
where according to Theorem \ref{decoupled} and Theorem \ref{decoupled1}, $\beta=\bar p_{m, opt}(\lambda)$ is the unique root of 
\begin{align}\label{J_mUnique}
    J_{m, 1}(\beta)+\lambda \mathbb E[T_{m,1}]=0,
\end{align} and
\begin{align}\label{Jmfunction}
J_{m, 1}(\beta)\nonumber=&\mathbb{E}\left[\sum_{t=D_{m,i}(\beta)}^{D_{m, i+1}(\beta)-1}  w_mp_m\big(\Delta_m(t)\big)\right]\nonumber\\
&-\beta \mathbb{E}\left[D_{m, i+1}(\beta)-D_{m,i}(\beta)\right].
\end{align}
Because $\lambda \mathbb E[T_{m,1}]$ does not change with $\beta$, from Lemma \ref{jB}, we can show that $J_{m, 1}(\beta)$ is a strictly decreasing function of $\beta$ with $\lim_{\beta \to \infty} j_{m,1}(\beta) = -\infty$ and 
$\lim_{\beta \to -\infty} j_{m, 1}(\beta) = \infty$. Hence, the inverse function $J_{m, 1}^{-1}(\cdot)$ exists and from \eqref{J_mUnique}, we get $J_{m, 1}^{-1}(\lambda \mathbb E[T_{m,1}])=\bar p_{m, opt}(\lambda)$. Because the inverse function of a strictly decreasing function is strictly increasing, $\bar p_{m, opt}(\lambda)$ is strictly increasing function of $\lambda$. Substituting this into \eqref{passive-set}, we get that if $\lambda_1\leq \lambda_2$, then $\Omega_m(\lambda_1) \subseteq \Omega_m(\lambda_2)$.

For dummy bandits, it is optimal in $\eqref{dummyproblem}$ to activate a dummy bandit when $\lambda \leq 0$. Hence, dummy bandits are always indexable.

\section{Proof of Theorem \ref{theorem8}}\label{ptheorem8}
By substituting \eqref{passive-set} into \eqref{defWhittle}, we obtain, if $d>0$, then $W_{m}(\delta, d)=-\infty$, and if $d=0$, then 
\begin{align}\label{eq_whittle_1}
W_{m}(\delta, 0)=\inf\{\lambda: \gamma_m(\delta) < \bar p_{m, opt}(\lambda)\}.
\end{align}
Using \eqref{optimammain} and \eqref{eq_whittle_1}, we get
\begin{align}
W_m(\delta, 0)=\max_{0 \leq b \leq B_m-1} W_{m, b}(\delta, 0),
\end{align}
where 
\begin{align}\label{eq_whittlle_new}
W_{m, b}(\delta, 0)=\inf\{\lambda: \gamma_m(\delta) < \bar p_{m,b,opt}(\lambda)\}.
\end{align}
Because $\bar p_{m,b,opt}(\lambda)$ is strictly increasing function of $\lambda$, \eqref{eq_whittlle_new} implies that $W_{m, b}(\delta, 0)$ is unique and satisfies 
\begin{align}\label{eq_whittle_2}
\bar p_{m,b,opt}(W_{m, b}(\delta, 0))=\gamma_m(\delta) .
\end{align}
By including source index $m$ into \eqref{pb}, we get
\begin{align}\label{pbm}
 &\bar p_{m, b, opt}(\lambda) \nonumber\\
&=\frac{\mathbb{E}\left[\sum_{t=D_{m,i}( \bar p_{m,b, opt}(\lambda))}^{D_{m, i+1}( \bar p_{m, b, opt}(\lambda))-1}  p_m\big(\Delta_m(t)\big)\right]+\lambda \mathbb E[T_{m,1}]}{\mathbb{E}\left[D_{i+1}( \bar p_{m, b, opt}(\lambda))-D_i( \bar p_{m,b, opt}(\lambda))\right]}.
\end{align}
By substituting $\lambda=W_{m, b}(\delta, 0)$ and \eqref{eq_whittle_2} into \eqref{pbm}. and then, re-arranging, we get \eqref{Whittle_Index}. This concludes the proof.

\section{Proof of Lemma \ref{divergenceL}}\label{PdivergenceL}
To prove Lemma \ref{divergenceL}, we will use the sub-gradient mean value theorem \cite{bertsekas2003convex}. When $H_L(Y)$ is twice differentiable in $P_Y$, we can use a second order Taylor series expansion. 

If \eqref{condition_divergenceL} holds, we can obtain 
\begin{align}\label{L2}
\sum_{y \in \mathcal Y} (P_Y(y)-Q_Y(y))^2 \leq \beta^2,
\end{align}
\begin{align}\label{L1}
\sum_{y \in \mathcal Y} |P_Y(y)-Q_Y(y)| \leq \beta,
\end{align}
and 
\begin{align}\label{maxnorm}
\max_{y \in \mathcal Y} |P_Y(y)-Q_Y(y)| \leq \beta.
\end{align}

Let us define a convex function $g: \mathbb R^{|\mathcal Y|} \mapsto \mathbb R$ as
\begin{align}\label{gfunction}
g(\mathbf z)=\sum_{i=1}^{|\mathcal Y|} z_i~L(y_i, a_{Q_Y})-\min_{a\in\mathcal A} \sum_{i=1}^{|\mathcal Y|}  z_i~ L(y_i, a),
\end{align}
where $a_{Q_Y}$ is a Bayes action associated with the distribution $Q_Y$, i.e., $a_{Q_Y}$ is the minimizer of 
\begin{align}
a_{Q_Y}=\argmin_{a \in \mathcal A} \mathbb E_{Y \sim Q_Y} [L(Y, a]]. 
\end{align}

Because $g(\mathbf z)$ is a convex function and the set of sub-gradients of $g(\mathbf z)$ is bounded \cite[Proposition 4.2.3]{bertsekas2003convex}, we can apply the sub-gradient mean value theorem \cite{bertsekas2003convex} along with \eqref{divergence}, and \eqref{L1} to obatin
\begin{align}\label{Lemma3_e1}
    g(\mathbf{p}_Y)=&D_L(P_Y||Q_Y)\nonumber\\
    =&g(\mathbf{q}_Y)+O\left(\sum_{y \in \mathcal Y} \left|P_Y(y)-Q_Y(y)\right|\right) \nonumber\\
 =&D_L(Q_Y||Q_Y)+O\left(\sum_{y \in \mathcal Y} \left|P_Y(y)-Q_Y(y)\right|\right)\nonumber\\
 =&O(\beta).
\end{align}

Now, let us consider the case where $H_L(Y)$ is assumed to be twice differentiable in $P_Y$. The function $g(\mathbf{p}_Y)$ can also be expressed in terms of $H_L(Y)$ as:
\begin{align}\label{gH}
g(\mathbf{p}_Y)=&\sum_{i=1}^{|\mathcal Y|} P_Y(y_i)~L(y_i, a_{Q_Y})-\min_{a\in\mathcal A} \sum_{i=1}^{|\mathcal Y|}  P_Y(y_i)~ L(y_i, a) \nonumber\\
=&\sum_{i=1}^{|\mathcal Y|} P_Y(y_i)~L(y_i, a_{Q_Y})-H_L(Y).
\end{align}
Because $H_L(Y)$ is assumed to be twice differentiable in $P_Y$, from \eqref{gH}, we can conclude that $g(\mathbf{p}_Y)$ is twice differentiable with respect to $\mathbf{p}_Y$. Furthermore, we have
\begin{align}
g(\mathbf{p}_Y) \geq 0, \forall \mathbf{p}_Y \in \mathbb R^{|\mathcal Y|}
\end{align} and 
\begin{align}\label{dqq}
g(\mathbf{q}_Y)=D_L(Q_Y||Q_Y)=0.
\end{align}
Using the first-order necessary condition for optimality, we find that the gradient of $g(\mathbf{p}_Y)$ at point $\mathbf{p}_Y=\mathbf{q}_Y$ is zero, i.e.,
\begin{align}\label{first_order_optimality}
\nabla g(\mathbf{q}_Y)=0.
\end{align}
Next, based on \eqref{dqq} and \eqref{first_order_optimality}, we can perform a second-order Taylor series expansion of $g(\mathbf p_Y)$ at $\mathbf p_Y=\mathbf q_Y$:
\begin{align}\label{Taylor}
g(\mathbf{p}_Y)
=&g(\mathbf{q}_Y)+(\mathbf{p}_Y-\mathbf{q}_Y)^T \nabla g(\mathbf{q}_Y)\nonumber\\
&+ \frac{1}{2} (\mathbf{p}_Y-\mathbf{q}_Y)^T \mathcal H(\mathbf{q}_Y) (\mathbf{p}_Y-\mathbf{q}_Y)\nonumber\\
&+ o\left(\sum_{y \in \mathcal Y} (P_Y(y)-Q_Y(y))^2\right) \nonumber\\
=&\frac{1}{2}\! (\mathbf{p}_Y\!-\mathbf{q}_Y)^T\! \mathcal H(\mathbf{q}_Y) (\mathbf{p}_Y-\mathbf{q}_Y)\! \nonumber\\
&+ o\left(\sum_{y \in \mathcal Y} \!(P_Y(y)\!-Q_Y(y))^2\right),
\end{align}
where $\mathcal H(\mathbf{q}_Y)$ is the Hessian matrix of $g(\mathbf{p}_Y)$ at point $\mathbf{p}_Y=\mathbf{q}_Y$. 

Because $g(\mathbf{p}_Y)$ is a convex function, we have $$(\mathbf{p}_Y-\mathbf{q}_Y)^T \mathcal H(\mathbf{q}_Y) (\mathbf{p}_Y-\mathbf{q}_Y) \geq 0.$$ Moreover, we can express
\begin{align}\label{big_O_H}
&\frac{1}{2}(\mathbf{p}_Y\!-\mathbf{q}_Y)^T\! \mathcal H(\mathbf{q}_Y) (\mathbf{p}_Y-\mathbf{q}_Y)\! \nonumber\\
&=\frac{1}{2}\sum_{y, y'} (P_{Y}(y)-Q_Y(y))\mathcal H(\mathbf{q}_Y)_{y, y'}(P_{Y}(y')-Q_Y(y'))\nonumber\\
&=O\left(\sum_{y, y'}\! (P_Y(y)-Q_Y(y))(P_{Y}(y')-Q_Y(y'))\right).
\end{align} 
Now, by substituting \eqref{big_O_H} into \eqref{Taylor}, we obtain
\begin{align}\label{last}
g(\mathbf{p}_Y)=&O\left(\sum_{y, y'}\! (P_Y(y)-Q_Y(y))(P_{Y}(y')-Q_Y(y'))\right)\nonumber\\
&+ o\left(\sum_{y \in \mathcal Y} \!(P_Y(y)\!-Q_Y(y))^2\right).
\end{align}
Using \eqref{L2} and \eqref{maxnorm}, we obtain from \eqref{last} that
\begin{align}
g(\mathbf{p}_Y)=D_L(P_Y||Q_Y)=O(\beta^2)+o(\beta^2)=O(\beta^2).
\end{align}

This completes the proof.

\ignore{\section{Proof of Equation \eqref{bounded_alpha}}\label{pbounded_alpha}
By using \eqref{alpha-discounted-problem}, we get 
\begin{align}\label{bound1}
J_{\alpha}(\delta) &\leq Q^{\alpha}_b(\delta, 1) \nonumber\\
&=\mathbb E \left [ \sum_{k=0}^{T_{i+1}-1} \alpha^k p(\delta+k)\right]+ \mathbb E[\alpha^{T_{i+1}} J_{\alpha}(T_{i+1}+b)]\nonumber\\
& \leq M~\mathbb E \left [ \sum_{k=0}^{T_{i+1}-1} \alpha^k \right]+\mathbb E[ J_{\alpha}(T_{i+1}+b)] \nonumber\\
& \leq M~\mathbb E[ T_{i+1}]+\mathbb E[ J_{\alpha}(T_{i+1}+b)].
\end{align}

Moreover, by using \eqref{alpha-discounted-problem}, we obtain
\begin{align}
J_{\alpha}(\delta) \geq \mathbb E[\alpha^{T_{i+1}} J_{\alpha}(T_{i+1}+b)],
\end{align}
which yields
\begin{align}\label{bound2}
\mathbb E[J_{\alpha}(T_{i+1}+b)]-J_{\alpha}(\delta) & \leq  \mathbb E[(1-\alpha^{T_{i+1}}) J_{\alpha}(T_{i+1}+b)] \nonumber\\
& \leq \mathbb E[(1-\alpha^{T_{i+1}})] \frac{M}{1-\alpha} \nonumber\\
& \leq (1-\alpha^{\mathbb E[T_{i+1}]}) \frac{M}{1-\alpha}\nonumber\\
&=(1+\alpha+\cdots+\alpha^{\mathbb E[T_{i+1}]-1})~M \nonumber\\
&\leq M~\mathbb E[T_{i+1}].
\end{align}
The third inequality is due to the convexity of $\alpha^{k}$ as a function of $k$ and Jensen's inequality. 

Equations \eqref{bound1} and \eqref{bound2} implies \eqref{bounded_alpha}. This concludes the proof.}

 
\section{Proof of Lemma \ref{ToyExampleLemma1}}\label{ToyExample}

Because $Y_t = f(X_{t-d})$ and $X_t$ is a Markov chain, $Y_t \leftrightarrow X_{t-\delta} \leftrightarrow X_{t-(\delta-1)}$ is a Markov chain for all $0 \leq \delta \leq d$. By the data processing inequality for $L$-conditional entropy \cite[Lemma 12.1] {Dawid1998}, one can show that  for all $0 \leq \delta \leq d$, 
\begin{align}
 H_L(Y_{t} | X_{t-\delta}) \leq H_L(Y_{t} | X_{t-(\delta-1)}).
\end{align}

Moreover, since $Y_t = f(X_{t-d})$ and $X_t$ is a Markov chain, $Y_t \leftrightarrow X_{t-\delta} \leftrightarrow X_{t-(\delta+1)}$ is a Markov chain for all $\delta \geq d$. By the data processing inequality \cite[Lemma 12.1] {Dawid1998}, one can show that for all $ \delta \geq d$, 
\begin{align}
H_L(Y_{t} | X_{t-\delta}) \leq H_L(Y_{t} | X_{t-(\delta+1)}).
\end{align}
Because $Y_t=f(X_{t-d})$ and $f(\cdot)$ is a function, we have $P_{Y_t|X_{t-d}}=P_{Y_t|f(X_{t-d})}=P_{Y_t|Y_t}$. Hence, 
\begin{align}
H_L(Y_t|X_{t-d})=H_L(Y_t|Y_t)\leq  H_L(Y_t|Z).
\end{align}
\section{Proof of Theorem \ref{asymptotic_optimal}}\label{pasymptotic_optimal}

\subsection{Notations and Definitions}
In \eqref{relaxed_multiple}-\eqref{Sconstraint3}, we have $M+N$ bandits with $M$ sources and $N$ dummy bandits, and $N$ channels. Two bandits are considered to be in the same class if they share identical penalty functions, weights, and transition probabilities. The dummy bandits belong to the same class. However, among the remaining $M$ sources, no two sources share the same combination of penalty function, weight, and transition probabilities. Therefore, we start with $M+1$ distinct classes of bandits. Then, we increase the number of sources $rM$, number of dummy bandits $rN$, and number of channels $rN$ with scaling factor $r$, while maintaining the ratio $\frac{rM}{rN}$.
 
{\blue The state of a class $m$ bandit is represented by two tuple $(\Delta_m(t), d_m(t))$, where $\Delta_m(t)$ represents the AoI and $d_m(t)$ is the amount of time spent to send the current feature from a class $m$ bandit. Let $\mu_m(t) \in \{0, 1, \ldots, B_m\}$ be the action taken for a class $m$ bandit. If $\mu_m(t) =0$, no feature is selected for transmission; otherwise, if $\mu_m(t) =b$, a feature from buffer position $(b-1) \in \{0, 1, \ldots, B_m-1\}$ is selected for transmission. Given state $(\delta, d)$ and action $\mu$ of a class $m$ bandit, we denote by $P^{m,\mu}_{(\delta',d'),(\delta, d)}$ the transition probability to a state $(\delta',d')$.}

{\blue Recall that the transmission times $T_{m, i}$ are assumed to be bounded for all $m$ and $i$. Hence, we can find a $d_{\text{bound}} \in \mathbb N$ such that $0< T_{m, i} \leq d_{\text{bound}}$ for all $m$ and $i$ and $d_m(t)\in \{0, 1, \ldots, d_{\text{bound}}\}.$ Also, we can find an AoI $\delta_{\text{bound}}$ such that $p_m(\delta)=p_m(\delta_{\text{bound}})$ for all $\delta \geq \delta_{\text{bound}}$ and for all $m$.}

We denote by $V^m_{\delta, d}(t)$ the fraction of class $m$ bandits with $\Delta_m(t)=\delta$ and $d_m(t)=d$. Let $U^m_{\delta, d, \mu}(t)$ be the fraction of class $m$ bandits with $\Delta_m(t)=\delta$, $d_m(t)=d$, and $\mu_m(t)=\mu$. We denote by $\tilde V^m_{d}(t)$ the fraction of class $m$ bandits with $d_m(t)=d$ and $\Delta_m(t)> \delta_{\text{bound}}$. Moreover, let $\tilde U^m_{d, \mu}(t)$ be the fraction of class $m$ bandits with $d_m(t)=d$, $\Delta_m(t)> \delta_{\text{bound}}$, and $\mu(t)=\mu$. We define variables $v^m_{\delta, d}$, $u^m_{\delta, d, \mu}$, $\tilde v^m_{d}$, and $\tilde u^m_{d, \mu}$ for all $\delta, d, \mu$, and $m$ as follows:
\begin{align}
 v^m_{\delta, d}:=\limsup_{T\rightarrow \infty}\sum_{t=0}^{T-1} \frac{1}{T} \mathbb E [ V^m_{\delta, d}(t)], \\
 u^m_{\delta, d, \mu}:=\limsup_{T\rightarrow \infty}\sum_{t=0}^{T-1} \frac{1}{T} \mathbb E [ U^m_{\delta, d, \mu}(t)], \\
  \tilde v^m_{d}:=\limsup_{T\rightarrow \infty}\sum_{t=0}^{T-1} \frac{1}{T} \mathbb E [ \tilde V^m_{d}(t)], \\
 \tilde u^m_{d, \mu}:=\limsup_{T\rightarrow \infty}\sum_{t=0}^{T-1} \frac{1}{T} \mathbb E [ \tilde U^m_{d, \mu}(t)].
\end{align}
A channel is occupied by a bandit, if either $d>0$ or $\mu> 0$. Then, the time-averaged expected fraction of class $m$-bandits occupied a channel is given by 
\begin{align}
c_m=&\sum_{\substack{\delta< \delta_{\text{bound}}, d>0 \\ \mu=0}} u^m_{\delta, d, \mu}+\sum_{\substack{\delta< \delta_{\text{bound}}, d=0 \\ \mu >0}} u^m_{\delta, d, \mu}\nonumber\\
&+\sum_{d>0, \mu=0} \tilde u^m_{d, \mu}+\sum_{d=0, \mu>0} \tilde u^m_{d, \mu}.
\end{align}
By using $c_m$, the time-average channel constraint \eqref{Sconstraint3} can be expressed as 
\begin{align}
\sum_{m=0}^M c_m=N.
\end{align}

\ignore{\begin{align}\label{LP1}
&\min_{(\mathbf u^m)_{m=0}^M} \sum_{m=1}^M w_m \sum_{\substack{\delta<\delta_{\text{bound}},\\ d, \mu}} p_m(\delta) u^m_{d, \mu}+\sum_{d, \mu} p_m(\delta_{\text{bound}}) \tilde u^m_{d, \mu}\\
&~~~~\text{s.t.}~~~\sum_{m=1}^M c_m=N,\\
& ~~~~~~~~~~~ \sum_{\mu} u^m_{\delta, d, \mu}=\sum_{\delta', d', \mu}u^m_{\delta', d', \mu} P^{m,\mu}_{(\delta,d),(\delta', d')}, \forall m, \delta, d,\\ 
& ~~~~~~~~~~~\sum_{\delta, d, \mu} u^m_{\delta, d, \mu}=1, \forall m, \\ 
& ~~~~~~~~~~\sum_{\delta, d>0,\mu\neq 0} u^m_{\delta, d, \mu}=0, \forall m, \\ \label{LPlast}
& ~~~~~~~~~~~ \mathbf u^m \geq 0, \forall m.
\end{align}}
\subsection{Asymptotic Optimality}
{\blue Policy $\pi$ in \eqref{relaxed_multiple}-\eqref{Sconstraint3} can be expressed as a sequence of actions $\pi=(\mu(t))_{t=0, 1, \ldots}$, with  $\mu(t)=\mu_{m}(t))_{m=0}^M$ representing actions taken at successive time slots. Let $(\mu^*(t))_{t=0, 1, \ldots}$ denote the optimal policy for solving \eqref{relaxed_multiple}-\eqref{Sconstraint3}. Given the optimal dual variable $\lambda=\lambda^*$, Theorems \ref{decoupled}-\ref{decoupled1} imply that in the optimal policy for \eqref{relaxed_multiple}-\eqref{Sconstraint3}, class $m$ bandits choose either action $\mu_m(t)=0$ or $\mu_m(t)=b^*_m(\lambda^*)+1$ at every time slot $t$. Therefore, the optimal state-action frequency satisfies: $u^{m*}_{\delta, d, \mu}=0$ for all $\mu\neq 0, b^*_m(\lambda^*)+1$. Hence, for all class $m$, only two actions can occur for every bandit in class $m$. Thus, our multiple-action RMAB problem can be reduced to a binary action RMAB problem. After that, for the truncated state space $(\delta, d) \in \{1, 2, \ldots, \delta_{\text{bound}}\} \times \{0, 1, \ldots, d_{\text{bound}}\}$, we directly use \cite[Theorem 13 and Proposition 14]{gast2021lp} to prove the asymptotic optimality of our policy. Though \cite[Theorem 13 and Proposition 14]{gast2021lp} is proved for a single class of bandits, we can use the results for multiple classes of bandits because of a similar argument provided in \cite[Section 5]{gast2023exponential}: we argue that having
multiple classes of bandits can be represented by a single class of bandits by considering a larger state space: the state of a bandit would be $(m, \delta, d)$, where $m$ is its class.}

\ignore{\subsection{Proof of Theorem \ref{asymptotic_optimal}}
We begin by introducing a set of linear programming (LP)-based priority policies that are asymptotic optimal under the uniform global attractor condition in Definition \ref{def2}. After that, we establish that $\pi_{our-dummy}$ belongs to this set of priority policies. Because the system performances under both $\pi_{our-dummy}$ and $\pi_{\mathrm{our}}$ are same, the asymptotic optimality for $\pi_{\mathrm{our}}$ follows directly.

Now, we define the following three sets:
\begin{align*}
\mathcal{S}^{m}_{+} & :=\{(\delta, d): u^m_{\delta, d, b^*_m(\lambda^*)}>0~\text{and}~u^m_{\delta, d, \mathrm{idle}}=0\},\\
\mathcal{S}^{m}_{0} & :=\{(\delta, d): u^m_{\delta, d, b^*_m(\lambda^*)}>0~\text{and}~u^m_{\delta, d, \mathrm{idle}}>0\},\\
\mathcal{S}^{m}_{-} & :=\{(\delta, d): u^m_{\delta, d, b^*_m(\lambda^*)}=0~\text{and}~u^m_{\delta, d, \mathrm{idle}}>0\}.
\end{align*}

\begin{definition}\label{LPP}
\textbf{LP-Priority Policies.}\cite{verloop2016asymptotically, gast2021lp}
We define a set $\Pi_{\mathrm{LP-Priority}}$ that consists of priority policies that satisfy the following conditions:

\begin{enumerate}
\item A class $k$ bandit in state $(\delta_k, d_k) \in \mathcal{S}^{k}_{+}$ is given higher priority than a class $j$ bandit in state $(\delta_j, d_j) \in \mathcal{S}^{j}_{0}$.
\item A class $k$ bandit in state $(\delta_k, d_k) \in \mathcal{S}^{k}_{0}$ is given higher priority than a class $j$ bandit in state $(\delta_j, d_j) \in \mathcal{S}^{j}_{-}$.
\item If a class $k$ bandit is activated, then the action $\mu=b_k(\lambda^*)$ is selected; otherwise, the action $\mu=\mathrm{idle}$ is selected.
\end{enumerate}

\end{definition}

\begin{lemma}\label{LP-Priority}
For any policy $\pi \in \Pi_{\mathrm{LP-Priority}}$, if the uniform global attractor condition in Definition \ref{def2} holds, then the policy is asymptotically optimal. 
\end{lemma}
By directly using Theorem 13 of \cite{gast2021lp}, we can prove Lemma \ref{LP-Priority}. While Theorem 13 of \cite{gast2021lp} deals with binary actions, our problem involves multiple actions for each bandit $m$. Nevertheless, using Theorems \ref{decoupled}-\ref{decoupled1}, we establish that in the optimal solution, only two actions, $\mu=\mathrm{idle}$ and $\mu=b^*_m(\lambda^*)$ for class $m$ bandit, occur when the optimal dual variable is $\lambda=\lambda^*$. Therefore, we can directly apply Theorem 12 from \cite{gast2021lp} to prove Lemma \ref{LP-Priority}.
 
Utilizing Proposition 14 from \cite{gast2021lp}, we observe the following for class $k$-bandits:
\begin{itemize}
\item[1.] For any class $k$-bandit, $(\delta_k, d_k)\in \mathcal S^k_{+}$ implies
\begin{align}\label{W1}
    W_m(\delta_k, d_k)>\lambda^*.
\end{align}
\item[2.] For any class $k$-bandit, $(\delta_k, d_k)\in \mathcal S^k_{0}$ implies
\begin{align}\label{W2}
    W_m(\delta_k, d_k)=\lambda^*.
\end{align}
\item[3.] For any class $k$-bandit, $(\delta_k, d_k)\in \mathcal S^k_{-}$ implies
\begin{align}\label{W3}
    W_m(\delta_k, d_k)<\lambda^*.
\end{align}
\end{itemize}
Because (1) $\pi_{our-dummy}$ activates exactly $N$ bandits with the highest Whittle index and (2) if a bandit $k$ is activated, $\pi_{our-dummy}$ chooses $\mu=b_k(\lambda^*)$, we can deduce from \eqref{W1}-\eqref{W3} and Definition \ref{LPP} that $\pi_{our-dummy}$ belongs to $\Pi_{\mathrm{LP-Priority}}$. This concludes the proof.

\ignore{We define $Z_m(\pi, \{\delta, d\}, \mu)$ as the state-action frequency for a class-$m$ bandit under policy $\pi$. In order to prove asymptotic optimality of $\pi_{\mathrm{our}}(\lambda^*)$, it is enough to show that  $Z_m(\pi^*(\lambda^*), \{\delta, d\}, \mu\}$ is a global attractor of $Z_m(\pi_{\mathrm{our}}(\lambda^*), \{\delta, d\}, \mu, t\}$ for any $\delta, d$, and $\mu$ \cite{verloop2016asymptotically}, where $\pi^*(\lambda^*)$ is the optimal policy of the relaxed problem \eqref{relaxed_multiple}-\eqref{Changed_constraint}.}

By using Theorem \ref{decoupled1}, we can deduce that under policy $\pi^*(\lambda^*)$, the possible actions for a class-$m$ bandit are $\mu=\text{Idle}$ and $\mu=b_m^*(\lambda^*)$, where $b_m^*(\lambda^*)$ is determined by \eqref{optimal_buffer_length_1}.

\begin{definition}[\textbf{Sufficient Condition}]\label{priority}
According to \cite[Proposition 4.14]{verloop2016asymptotically}, $Z_m(\pi^*(\lambda^*), \{\delta, d\}, \mu\}$ is a global attractor of $Z_m(\pi_{\mathrm{our}}(\lambda^*), \{\delta, d\}, \mu, t\}$ for any $\delta, d$, and $\mu$, if the following conditions hold under policy $\pi_{\mathrm{our}}(\lambda^*)$:
\begin{itemize}
\item[1.] A class-k bandit in state $(\delta, d)$ with $$Z_k(\pi^*(\lambda^*), \{\delta, d\}, b_k^*(\lambda^*))>0$$ and $$Z_k(\pi^*(\lambda^*), \{\delta, d\}, \mathrm{Idle})=0$$ is given higher priority than a class-j bandit in state $(\delta', d')$ with $$Z_j(\pi^*(\lambda^*), \{\delta', d'\}, \mathrm{Idle})>0.$$
\item[2.] A class-k bandit in state $(\delta, d)$ with $$Z_k(\pi^*(\lambda^*), \{\delta, d\}, b_k^*(\lambda^*))>0$$ and $$Z_k(\pi^*(\lambda^*), \{\delta, d\}, \mathrm{Idle})>0$$ is given higher priority than a class-j bandit in state $(\delta', d')$ with $$Z_j(\pi^*(\lambda^*), \{\delta', d'\}, \mathrm{Idle})>0$$ and $$Z_j(\pi^*(\lambda^*), \{\delta', d'\}, b_j^*(\lambda^*))=0.$$
\item[3.] if $$\sum_{\substack{\delta \in \{1, 2, \ldots\}\\ d \in \{0, 1, \ldots\}}}\sum_{k=1}^M Z_k(\pi^*(\lambda^*), \{\delta, d\}, b_k^*(\lambda^*))< N$$, then for any class-k bandit in state $(\delta, d)$ with $$Z_k(\pi^*(\lambda^*), \{\delta, d\}, b_k^*(\lambda^*))=0$$ and $$Z_k(\pi^*(\lambda^*), \{\delta, d\}, \mathrm{Idle})>0,$$
the policy always chooses the passive action $\mu=\mathrm{Idle}$. 
\end{itemize}
\end{definition}
Now, we examine each condition in detail. By utilizing the Whittle index, we can express $Z_m(\pi^*(\lambda^*), \{\delta, d\}, \mu)$ as follows:
\begin{align}\label{passiveaction}
Z_m(\pi^*(\lambda^*), \{\delta, d\}, \mathrm{Idle})&=0,~\mathrm{if}~W_m(\delta, d) \geq \lambda^*, \nonumber\\
Z_m(\pi^*(\lambda^*), \{\delta, d\}, \mathrm{Idle})&>0,~ \mathrm{otherwise}. 
\end{align}
The state action frequency under policy $\pi^*(\lambda^*)$ for class-m bandit with action $\mu=b_m^*(\lambda^*)$ satisfy
\begin{align}\label{activeaction}
Z_m(\pi^*(\lambda^*), \{\delta, d\}, b_m^*(\lambda^*))&=0,~\mathrm{if}~W_m(\delta, d) <\lambda^*, \nonumber\\
Z_m(\pi^*(\lambda^*), \{\delta, d\}, b_m^*(\lambda^*))&>0,~ \mathrm{otherwise}.
\end{align}

Our policy $\pi_{\mathrm{our}}(\lambda^*)$ gives higher priority to class-k bandit in state $(\delta, d)$ with $W_k(\delta, d)>W_j(\delta', d')$ than class-j bandit in state $(\delta', d')$. Comparing this with \eqref{passiveaction}-\eqref{activeaction}, we get that the policy $\pi_{\mathrm{our}}(\lambda^*)$ satisfies point 1 and 2 of Definition \ref{priority}.

Now, we treat point 3 of Definition \ref{priority}. If
\begin{align}\label{inequality}
\sum_{\substack{\delta \in \{1, 2, \ldots\}\\ d \in \{0, 1, \ldots\}}}\sum_{k=1}^M Z_k(\pi^*(\lambda^*), \{\delta, d\}, b_k^*(\lambda^*))< N,
\end{align}
then the remaining fraction of time the policy $\pi^*(\lambda^*)$ makes dummy bandits active. Hence, $W_0(\delta, 0)\geq \lambda^*$ for all $N$ dummy bandits. Since, by using Lemma \ref{dummyWhittle}, $W_0(\delta, 0)=0$ for all $N$ dummy bandits, we must have $\lambda^*=0$. If $\lambda^*=0$, and 
$$Z_k(\pi^*(\lambda^*), \{\delta, d\}, b_k^*(\lambda^*))=0$$ and $$Z_k(\pi^*(\lambda^*), \{\delta, d\}, \mathrm{Idle})>0,$$
then $W_k(\delta, d)<0$. Since our policy $\pi_{\mathrm{our}}(\lambda^*)$ always chooses $\mu=\mathrm{Idle}$ for bandit $k$ in state $(\delta, d)$ with $W_k(\delta, d)<0$, point 3 of Definition \ref{priority} is also satisfied. This completes the proof.}

\section{Proof of Lemma \ref{jB}}\label{pjB}

We denote $D_{i+1}(\beta)=S_{i+1}(\beta)+T_{i+1}$, $S_{i+1}(\beta)=\tau(\Delta(D_i), \beta)+D_i(\beta)$, where  $\tau(\delta, \beta)$ is defined as the optimal solution of 
\begin{align}
\inf_{\tau \in \{0, 1, 2, \ldots\}}\mathbb E \left [ \sum_{k=0}^{\tau+T_{1}-1} (wp(\delta+k) -\beta)\right].
\end{align}
Because $T_i$'s are i.i.d., and $\Delta(D_i)=T_i+b$, we can express \eqref{Jfunction} as 
\begin{align}\label{Jfunction1}
J(\beta)=&\mathbb{E}\left[\sum_{t=D_i(\beta)}^{D_{i+1}(\beta)-1}  wp\big(\Delta(t)\big)\right]\nonumber\\
&-\beta \mathbb{E}\left[D_{i+1}(\beta)-D_i(\beta)\right]+\lambda \mathbb E[T_1]\nonumber\\
=&\mathbb E \left [ \sum_{k=0}^{\tau(T_1+b, \beta)+T_{1}-1} (wp(T_1+b+k) -\beta)\right]\nonumber\\
&+\lambda \mathbb E[T_1]\nonumber\\
=&\inf_{\tau \in \{0, 1, \ldots}\mathbb E \left [ \sum_{k=0}^{\tau+T_{1}-1} (wp(T_1+b+k) -\beta)\right]\nonumber\\
&+\lambda \mathbb E[T_1].
\end{align}
Since the right-hand side of \eqref{Jfunction1} is the pointwise infimum of the linear decreasing functions of $\beta$, $J(\beta)$ is concave, continuous, and strictly decreasing in $\beta$. This completes the proof of part (i) of Lemma \ref{jB}. 

Part (ii) of Lemma \ref{jB} holds because for any $\tau\geq0$
\begin{align}
\lim_{\beta \to \infty} \mathbb E \left [ \sum_{k=0}^{\tau+T_{1}-1} (wp(T_1+b+k) -\beta)\right]=-\infty
\end{align}
and 
\begin{align}
\lim_{\beta \to -\infty} \mathbb E \left [ \sum_{k=0}^{\tau+T_{1}-1} (wp(T_1+b+k) -\beta)\right]=\infty.
\end{align}
This completes the proof.

%% file: ToNPaper.bbl
\begin{thebibliography}{10}
\providecommand{\url}[1]{#1}
\csname url@samestyle\endcsname
\providecommand{\newblock}{\relax}
\providecommand{\bibinfo}[2]{#2}
\providecommand{\BIBentrySTDinterwordspacing}{\spaceskip=0pt\relax}
\providecommand{\BIBentryALTinterwordstretchfactor}{4}
\providecommand{\BIBentryALTinterwordspacing}{\spaceskip=\fontdimen2\font plus
\BIBentryALTinterwordstretchfactor\fontdimen3\font minus
  \fontdimen4\font\relax}
\providecommand{\BIBforeignlanguage}[2]{{%
\expandafter\ifx\csname l@#1\endcsname\relax
\typeout{** WARNING: IEEEtran.bst: No hyphenation pattern has been}%
\typeout{** loaded for the language `#1'. Using the pattern for}%
\typeout{** the default language instead.}%
\else
\language=\csname l@#1\endcsname
\fi
#2}}
\providecommand{\BIBdecl}{\relax}
\BIBdecl

\bibitem{ShisherMobihoc}
M.~K.~C. Shisher and Y.~Sun, ``How does data freshness affect real-time
  supervised learning?'' \emph{ACM MobiHoc}, 2022.

\bibitem{song1990performance}
X.~Song and J.~W.-S. Liu, ``Performance of multiversion concurrency control
  algorithms in maintaining temporal consistency,'' in \emph{Fourteenth Annual
  International Computer Software and Applications Conference}.\hskip 1em plus
  0.5em minus 0.4em\relax IEEE, 1990, pp. 132--133.

\bibitem{kaul2012real}
S.~Kaul, R.~Yates, and M.~Gruteser, ``Real-time status: How often should one
  update?'' in \emph{IEEE INFOCOM}, 2012, pp. 2731--2735.

\bibitem{yates2015lazy}
R.~D. Yates, ``Lazy is timely: Status updates by an energy harvesting source,''
  in \emph{IEEE ISIT}, 2015, pp. 3008--3012.

\bibitem{sun2017update}
Y.~Sun, E.~Uysal-Biyikoglu, R.~D. Yates, C.~E. Koksal, and N.~B. Shroff,
  ``Update or wait: How to keep your data fresh,'' \emph{IEEE Trans. Inf.
  Theory}, vol.~63, no.~11, pp. 7492--7508, 2017.

\bibitem{YinUpdateInfocom}
Y.~Sun, E.~Uysal-Biyikoglu, R.~Yates, C.~E. Koksal, and N.~B. Shroff, ``Update
  or wait: How to keep your data fresh,'' in \emph{IEEE INFOCOM}, 2016, pp.
  1--9.

\bibitem{SunNonlinear2019}
Y.~Sun and B.~Cyr, ``Sampling for data freshness optimization: Non-linear age
  functions,'' \emph{J. Commun. Netw.}, vol.~21, no.~3, pp. 204--219, 2019.

\bibitem{SunSPAWC2018}
------, ``Information aging through queues: A mutual information perspective,''
  in \emph{Proc. IEEE SPAWC Workshop}, 2018.

\bibitem{orneeTON2021}
T.~Z. Ornee and Y.~Sun, ``Sampling and remote estimation for the
  ornstein-uhlenbeck process through queues: Age of information and beyond,''
  \emph{IEEE/ACM Trans. on Netw.}, vol.~29, no.~5, pp. 1962--1975, 2021.

\bibitem{Tripathi2019}
V.~Tripathi and E.~Modiano, ``A {Whittle} index approach to minimizing
  functions of age of information,'' in \emph{IEEE Allerton}, 2019, pp.
  1160--1167.

\bibitem{klugel2019aoi}
M.~Kl{\"u}gel, M.~H. Mamduhi, S.~Hirche, and W.~Kellerer, ``Ao{I}-penalty
  minimization for networked control systems with packet loss,'' in \emph{IEEE
  INFOCOM Age of Information Workshop}, 2019, pp. 189--196.

\bibitem{bedewy2021optimal}
A.~M. Bedewy, Y.~Sun, S.~Kompella, and N.~B. Shroff, ``Optimal sampling and
  scheduling for timely status updates in multi-source networks,'' \emph{IEEE
  Trans. Inf. Theory}, vol.~67, no.~6, pp. 4019--4034, 2021.

\bibitem{kadota2018optimizing}
I.~Kadota, A.~Sinha, and E.~Modiano, ``Optimizing age of information in
  wireless networks with throughput constraints,'' in \emph{IEEE INFOCOM},
  2018, pp. 1844--1852.

\bibitem{hsu2018age}
Y.~Hsu, ``Age of information: Whittle index for scheduling stochastic
  arrivals,'' in \emph{IEEE ISIT}, 2018, pp. 2634--2638.

\bibitem{sun2019closed}
J.~Sun, Z.~Jiang, B.~Krishnamachari, S.~Zhou, and Z.~Niu, ``Closed-form
  {W}hittle’s index-enabled random access for timely status update,''
  \emph{IEEE Transactions on Communications}, vol.~68, no.~3, pp. 1538--1551,
  2019.

\bibitem{Kadota2018}
I.~Kadota, A.~Sinha, E.~Uysal-Biyikoglu, R.~Singh, and E.~Modiano, ``Scheduling
  policies for minimizing age of information in broadcast wireless networks,''
  \emph{IEEE/ACM Trans. Netw.}, vol.~26, no.~6, pp. 2637--2650, 2018.

\bibitem{verloop2016asymptotically}
I.~M. Verloop, ``Asymptotically optimal priority policies for indexable and
  nonindexable restless bandits,'' \emph{The Annals of Applied Probability},
  vol.~26, no.~4, p. 1947–1995, 2016.

\bibitem{gast2021lp}
N.~Gast, B.~Gaujal, and C.~Yan, ``{LP}-based policies for restless bandits:
  necessary and sufficient conditions for (exponentially fast) asymptotic
  optimality,'' \emph{Mathematics of Operations Research}, pp. 1--29, 2023.

\bibitem{lee2018stochastic}
A.~X. Lee, R.~Zhang, F.~Ebert, P.~Abbeel, C.~Finn, and S.~Levine, ``Stochastic
  adversarial video prediction,'' \emph{arXiv:1804.01523}, 2018.

\bibitem{shisher2021age}
M.~K.~C. Shisher, H.~Qin, L.~Yang, F.~Yan, and Y.~Sun, ``The age of correlated
  features in supervised learning based forecasting,'' in \emph{IEEE INFOCOM
  Age of Information Workshop}, 2021.

\bibitem{KamKompellaEphremidesTIT}
C.~Kam, S.~Kompella, G.~D. Nguyen, and A.~Ephremides, ``Effect of message
  transmission path diversity on status age,'' \emph{IEEE Trans. Inf. Theory},
  vol.~62, no.~3, pp. 1360--1374, March 2016.

\bibitem{soleymani2016optimal}
T.~Soleymani, S.~Hirche, and J.~S. Baras, ``Optimal self-driven sampling for
  estimation based on value of information,'' in \emph{IEEE WODES}, 2016, pp.
  183--188.

\bibitem{chen2021uncertainty}
G.~Chen, S.~C. Liew, and Y.~Shao, ``Uncertainty-of-information scheduling: A
  restless multiarmed bandit framework,'' \emph{IEEE Trans. Inf. Theory},
  vol.~68, no.~9, pp. 6151--6173, 2022.

\bibitem{wang2022framework}
Z.~Wang, M.-A. Badiu, and J.~P. Coon, ``A framework for characterising the
  value of information in hidden {Markov} models,'' \emph{IEEE Trans. Inf.
  Theory}, 2022.

\bibitem{ornee2023whittle}
T.~Z. Ornee and Y.~Sun, ``A {Whittle} index policy for the remote estimation of
  multiple continuous {G}auss-{M}arkov processes over parallel channels,''
  \emph{ACM MobiHoc}, 2023.

\bibitem{pan2023sampling}
J.~Pan, Y.~Sun, and N.~B. Shroff, ``Sampling for remote estimation of the
  wiener process over an unreliable channel,'' \emph{ACM Sigmetrics}, 2023.

\bibitem{sun2023age}
Y.~Sun and S.~Kompella, ``Age-optimal multi-flow status updating with errors: A
  sample-path approach,'' \emph{J. Commun. Netw.}, vol.~25, no.~5, pp.
  570--584, 2023.

\bibitem{sun2022age}
Y.~Sun, I.~Kadota, R.~Talak, and E.~Modiano, \emph{Age of information: A new
  metric for information freshness}.\hskip 1em plus 0.5em minus 0.4em\relax
  Springer Nature, 2022.

\bibitem{ornee2023context}
T.~Z. Ornee, M.~K.~C. Shisher, C.~Kam, and Y.~Sun, ``Context-aware status
  updating: Wireless scheduling for maximizing situational awareness in
  safety-critical systems,'' in \emph{IEEE MILCOM}, 2023, pp. 194--200.

\bibitem{ayan2023optimal}
O.~Ayan, S.~Hirche, A.~Ephremides, and W.~Kellerer, ``Optimal finite horizon
  scheduling of wireless networked control systems,'' \emph{IEEE/ACM Trans. on
  Netw.}, vol.~32, no.~2, pp. 927--942, 2024.

\bibitem{yates2021age}
R.~D. Yates, Y.~Sun, D.~R. Brown, S.~K. Kaul, E.~Modiano, and S.~Ulukus, ``Age
  of information: An introduction and survey,'' \emph{IEEE J. Select. Areas in
  Commun.}, vol.~39, no.~5, pp. 1183--1210, 2021.

\bibitem{whittle1988restless}
P.~Whittle, ``Restless bandits: Activity allocation in a changing world,''
  \emph{Journal of applied probability}, vol.~25, no.~A, pp. 287--298, 1988.

\bibitem{Ornee2021performance}
T.~Z. Ornee and Y.~Sun, ``Performance bounds for sampling and remote estimation
  of gauss-markov processes over a noisy channel with random delay,'' in
  \emph{IEEE SPAWC}, 2021.

\bibitem{SunTIT2020}
Y.~Sun, Y.~Polyanskiy, and E.~Uysal, ``Sampling of the {Wiener} process for
  remote estimation over a channel with random delay,'' \emph{IEEE Trans. Inf.
  Theory}, vol.~66, no.~2, pp. 1118--1135, 2020.

\bibitem{goodfellow2016deep}
I.~Goodfellow, Y.~Bengio, and A.~Courville, \emph{Deep learning}.\hskip 1em
  plus 0.5em minus 0.4em\relax MIT press, 2016.

\bibitem{brockman2016openai}
G.~Brockman, V.~Cheung, L.~Pettersson, J.~Schneider, J.~Schulman, J.~Tang, and
  W.~Zaremba, ``Openai gym,'' \emph{arXiv:1606.01540}, 2016.

\bibitem{Dawid2004}
P.~D. Gr{\"u}nwald and A.~P. Dawid, ``Game theory, maximum entropy, minimum
  discrepancy and robust {Bayesian} decision theory,'' \emph{Annals of
  Statistics}, vol.~32, no.~4, pp. 1367--1433, 08 2004.

\bibitem{Dawid1998}
A.~P. Dawid, ``Coherent measures of discrepancy, uncertainty and dependence,
  with applications to {Bayesian} predictive experimental design,''
  \emph{Technical Report 139}, 1998.

\bibitem{farnia2016minimax}
F.~Farnia and D.~Tse, ``A minimax approach to supervised learning,''
  \emph{NIPS}, vol.~29, pp. 4240--4248, 2016.

\bibitem{shisher2022local}
M.~K.~C. Shisher, T.~Z. Ornee, and Y.~Sun, ``A local geometric interpretation
  of feature extraction in deep feedforward neural networks,''
  \emph{arXiv:2202.04632}, 2022.

\bibitem{dhillon2008matrix}
I.~S. Dhillon and J.~A. Tropp, ``Matrix nearness problems with bregman
  divergences,'' \emph{SIAM Journal on Matrix Analysis and Applications},
  vol.~29, no.~4, pp. 1120--1146, 2008.

\bibitem{csiszar2004information}
I.~Csisz{\'a}r and P.~C. Shields, ``Information theory and statistics: A
  tutorial,'' 2004.

\bibitem{huang2019universal}
S.-L. Huang, A.~Makur, G.~W. Wornell, and L.~Zheng, ``Universal features for
  high-dimensional learning and inference,'' \emph{Foundations and Trends in
  Communications and Information Theory}, vol.~21, no. 1-2, pp. 1--299, 2024.

\bibitem{polyanskiy2014lecture}
Y.~Polyanskiy and Y.~Wu, ``Lecture notes on information theory,'' \emph{Lecture
  Notes for MIT (6.441), UIUC (ECE 563), Yale (STAT 664)}, no. 2012-2017, 2014.

\bibitem{cover1999elements}
T.~M. Cover, \emph{Elements of information theory}.\hskip 1em plus 0.5em minus
  0.4em\relax John Wiley \& Sons, 1999.

\bibitem{stochasticOrder}
M.~Shaked and J.~G. Shanthikumar, \emph{Stochastic orders}.\hskip 1em plus
  0.5em minus 0.4em\relax Springer Science \& Business Media, 2007.

\bibitem{Shannon1948}
C.~E. Shannon, ``A mathematical theory of communication,'' \emph{The Bell
  system technical journal}, vol.~27, no.~3, pp. 379--423, 1948.

\bibitem{shisher2023learning}
M.~K.~C. Shisher, B.~Ji, I.~Hou, and Y.~Sun, ``Learning and communications
  co-design for remote inference systems: Feature length selection and
  transmission scheduling,'' \emph{IEEE J. Select. Areas in Inf. Theory}, 2023.

\bibitem{shisher2024monotonicity}
M.~K.~C. Shisher and Y.~Sun, ``On the monotonicity of information aging,''
  \emph{IEEE INFOCOM ASoI Workshop}, 2024.

\bibitem{bertsekasdynamic1}
D.~Bertsekas, \emph{Dynamic programming and optimal control: Volume I}.\hskip
  1em plus 0.5em minus 0.4em\relax Athena scientific, 2017, vol.~1.

\bibitem{gittins2011multi}
J.~Gittins, K.~Glazebrook, and R.~Weber, \emph{Multi-armed bandit allocation
  indices}.\hskip 1em plus 0.5em minus 0.4em\relax John Wiley \& Sons, 2011.

\bibitem{katehakis1987multi}
M.~N. Katehakis and A.~F. Veinott~Jr, ``The multi-armed bandit problem:
  decomposition and computation,'' \emph{Mathematics of Operations Research},
  vol.~12, no.~2, pp. 262--268, 1987.

\bibitem{papadimitriou1994complexity}
C.~H. Papadimitriou and J.~N. Tsitsiklis, ``The complexity of optimal queueing
  network control,'' in \emph{Proceedings of IEEE 9th Annual Conference on
  Structure in Complexity Theory}, 1994, pp. 318--322.

\bibitem{palomar2006tutorial}
D.~P. Palomar and M.~Chiang, ``A tutorial on decomposition methods for network
  utility maximization,'' \emph{IEEE J. Select. Areas in Commun.}, vol.~24,
  no.~8, pp. 1439--1451, 2006.

\bibitem{ebert17sna}
F.~Ebert, C.~Finn, A.~Lee, and S.~Levine, ``Self-supervised visual planning
  with temporal skip connections,'' \emph{arXiv:1710.05268}, 2017.

\bibitem{berenson2009manipulation}
D.~Berenson, S.~S. Srinivasa, D.~Ferguson, A.~Collet, and J.~J. Kuffner,
  ``Manipulation planning with workspace goal regions,'' in \emph{2009 IEEE
  international conference on robotics and automation}.\hskip 1em plus 0.5em
  minus 0.4em\relax IEEE, 2009, pp. 618--624.

\bibitem{mnih2015human}
V.~Mnih, K.~Kavukcuoglu, D.~Silver \emph{et~al.}, ``Human-level control through
  deep reinforcement learning,'' \emph{nature}, vol. 518, no. 7540, pp.
  529--533, 2015.

\bibitem{kerasexample}
P.~Attri, Y.~Sharma, K.~Takach, and F.~Shah, ``Timeseries forecasting for
  weather prediction,'' 2020, online:
  \url{https://keras.io/examples/timeseries/timeseries_weather_forecasting/}.

\bibitem{baddour2005autoregressive}
K.~E. Baddour and N.~C. Beaulieu, ``Autoregressive modeling for fading channel
  simulation,'' \emph{IEEE Trans. Wireless Commun.}, vol.~4, no.~4, pp.
  1650--1662, 2005.

\bibitem{alpha}
J.~{Liao}, O.~{Kosut}, L.~{Sankar}, and F.~P. {Calmon}, ``A tunable measure for
  information leakage,'' in \emph{IEEE ISIT}, 2018, pp. 701--705.

\bibitem{zhang2017matrix}
X.-D. Zhang, \emph{Matrix analysis and applications}.\hskip 1em plus 0.5em
  minus 0.4em\relax Cambridge University Press, 2017.

\bibitem{Amari}
S.-I. Amari, ``$\alpha$ -divergence is unique, belonging to both $f$-divergence
  and bregman divergence classes,'' \emph{IEEE Trans. Inf. Theory}, vol.~55,
  no.~11, pp. 4925--4931, 2009.

\bibitem{durrett2019probability}
R.~Durrett, \emph{Probability: theory and examples}.\hskip 1em plus 0.5em minus
  0.4em\relax Cambridge university press, 2019, vol.~49.

\bibitem{puterman2014markov}
M.~L. Puterman, \emph{Markov decision processes: discrete stochastic dynamic
  programming}.\hskip 1em plus 0.5em minus 0.4em\relax John Wiley \& Sons,
  2014.

\bibitem{bertsekas2003convex}
D.~Bertsekas, A.~Nedic, and A.~Ozdaglar, \emph{Convex analysis and
  optimization}.\hskip 1em plus 0.5em minus 0.4em\relax Athena Scientific,
  2003, vol.~1.

\bibitem{gast2023exponential}
N.~Gast, B.~Gaujal, and C.~Yan, ``Exponential asymptotic optimality of whittle
  index policy,'' \emph{Queueing Systems}, vol. 104, no.~1, pp. 107--150, 2023.

\end{thebibliography}
